\documentclass[a4paper,12pt]{article}

\usepackage{bm}
\usepackage{enumerate}
\usepackage{caption}

\usepackage{epstopdf}
\usepackage {svg}
\usepackage{here}

\usepackage{amsmath}
\usepackage{amssymb}
\usepackage{amsthm}
\usepackage{mathrsfs}
\usepackage[all]{xy}
\usepackage{float}
\usepackage{graphicx}

\usepackage{hyperref}

\usepackage{tikz}
\usetikzlibrary{intersections,calc,arrows.meta}

\newtheorem{thm}{Theorem}[section]
\newtheorem{prop}[thm]{Proposition}
\newtheorem{lem}[thm]{Lemma}
\newtheorem{dfn}[thm]{Definition}

\newtheorem{Remark}[thm]{Remark}

\numberwithin{equation}{section}

\makeatletter 

\@addtoreset{figure}{section}
\@addtoreset{table}{section}

\makeatother 

\usepackage{lipsum} 

\usepackage{braket}

\newcommand\restr[2]{{
  \left.\kern-\nulldelimiterspace 
  #1 
  \vphantom{\big|} 
  \right|_{#2} 
  }}

\usepackage{comment}

\title{Derivations on the triplet $W$-algebras with $\mathfrak{sl}_2$-symmetry}
\author{Hiromu Nakano}
\date{}

\begin{document}

\maketitle
\begin{abstract}
We construct derivations on the triplet $W$-algebras $\mathcal{W}_{p_+,p_-}$ by refining the Frobenius homomorphisms of Tsuchiya-Wood
and show that the property of the Adamovi\'{c}-Milas derivation for $\mathcal{W}_{2,p}$ extends to our derivations.
As an application, we show that the $\mathfrak{sl}_2$-symmetry of $\mathcal{W}_{p_+,p_-}$ arises naturally from our construction. 
We further show that our method applies to the triplet $W$-superalgebra $\mathcal{SW}(m)$ and that the full automorphism group ${\rm Aut}(\mathcal{SW}(m))$ is $PSL_2(\mathbb{C})\times \mathbb{Z}_2$.
\end{abstract}

\setcounter{tocdepth}{1}
\tableofcontents

\section{Introduction}
The triplet $W$-algebras $\mathcal{W}_{p_+,p_-}$ are a family of VOAs parametrised by two coprime integers $p_-> p_+\geq 1$. The special case $p_+=1$ is usually denoted by $\mathcal{W}_{p_-}$.
They are defined in terms of the lattice vertex operator algebra $\mathcal{V}_{[p_+,p_-]}$ associated with the rank-one lattice $\mathbb{Z}\sqrt{2p_+p_-}$, and the two screening operators $Q_+$ and $Q_-$ \cite{FGST,FHST}.
It is known that the triplet $W$-algebras $\mathcal{W}_{p_+,p_-}$ satisfy Zhu's $c_2$-cofiniteness condition
and that $\mathcal{W}_{1,p}$ is simple, whereas $\mathcal{W}_{p_+,p_-}$ $({\rm for}\ p_+\geq 2)$ is not. Instead, for $p_+\geq 2$, the $W$-algebra $\mathcal{W}_{p_+,p_-}$ has the Virasoro minimal VOA as a simple quotient.
These results were proved in the case $p_+=1$ in \cite{AdamovicD/MilasA:2008}, in the case $p_+=2$ in \cite{AdamovicD/MilasA:2010,AdamovicD/MilasA:2011}, and for $p_+\geq 3$ in \cite{TW}.
The proofs of these facts, particularly in the case $p_+ \geq  2$, rely on derivations for $\mathcal{W}_{p_+,p_-}$ that commutes with the Virasoro action on $\mathcal{W}_{p_+,p_-}$.
For the case $p_+=2$, Adamovi\'{c} and Milas \cite{AdamovicD/MilasA:2010} were the first to construct the derivation 
\begin{align}
\label{eq:int-deri}
{\rm Res}_{x_1}{\rm Res}_{x_2}\log\Bigl(1-\frac{x_2}{x_1}\Bigr)Q_+(x_1)Q_+(x_2),
\end{align}
where $Q_+(x)$ is a primary field whose zero mode is $Q_+$.
This derivation can be regarded as a $\log(1-x_2/x_1)$-twisted version of the relation $Q_+ \circ Q_+ = 0$ that defines the Felder complex~\cite{Felder}.
For the general case of $p_+ \geq 2$, motivated by the results in \cite{AdamovicD/MilasA:2010,AdamovicD/MilasA:2011}, Tsuchiya and Wood \cite{TW} define derivations, called Frobenius homomorphisms, as a certain limit of an $\epsilon$-deformation of the relations $Q^{[p_+-1]}_+ \circ Q_+ = 0$ and $Q^{[p_--1]}_- \circ Q_- = 0$, which define two Felder complexes. Here, $Q^{[p_+-1]}_+$ and $Q^{[p_--1]}_-$ are the (multiple) screening operators constructed in \cite{TK}.
However, the twisted cycles considered therein require further clarification regarding their well-definedness
 (see Remark \ref{non-well-tw}  in Subsection \ref{sec-deri-main}). 
 For this reason, the construction of derivations in \cite{TW} leaves a certain technical point to be investigated.

If we formally extend the derivation in (\ref{eq:int-deri}) to the case of $p_+\geq 3$, 
we obtain the following operators:
\begin{align}
\label{eq:int-deri-m}
\int_{\Gamma^{(\alpha_\pm)}_{p_\pm-1}}{\rm d}\bm{z}'\oint_{z=0}\frac{{\rm d}z}{2\pi i} \log\Bigl(\prod_{i=1}^{p_\pm-1}\bigl(1-\frac{z}{z'_i}\bigr)\Bigr){Q}_\pm(z'_1){Q}_\pm(z'_2)\cdots {Q}_\pm(z'_{p_\pm-1}) {Q}_\pm(z).
\end{align}
Here $\Gamma^{(\alpha_\pm)}_{p_\pm-1}$ denotes a certain twisted cycle constructed in \cite{TK}. 
However, as each expression is expressed as a multivariable integral, the method of \cite{AdamovicD/MilasA:2010} does not seem to be directly applicable to the analysis of their properties.
Therefore, in this paper, we use the $\epsilon$-deformation method of \cite{TW} 
to investigate a more tractable representation of the operators analogous to (\ref{eq:int-deri-m}).
As stated in our main result, Theorem~\ref{G+G-prop}, the approach of \cite{TW} applies to more elementary $\epsilon$-deformed operators
\begin{align}
\label{eq:int-deri-ep}
\frac{1}{\epsilon}\widetilde{Q}^{[p_\pm-1]}_\pm\circ {\rm e}^{\bullet}\circ \widetilde{Q}_\pm\circ {\rm e}^{\bullet}.
\end{align}
Here, $\widetilde{Q}^{[\bullet]}_\pm$ denote $\epsilon$-deformed screening operators and ${\rm e}^{\bullet}$ are certain weight-shifting operators of Fock modules.
As shown in \cite{TW}, the $\epsilon$-deformed screening operators satisfy a certain $\mathbb{C}[[\epsilon]]$-integrality property.
Using the relations $Q^{[p_\pm-1]}_\pm \circ Q_\pm = 0$, it follows that (\ref{eq:int-deri-ep}) also satisfies the $\mathbb{C}[[\epsilon]]$-integrality property.
As discussed in Subsection~\ref{cons-main0}, by using a recent study of the analytic continuation of the Selberg integral by Sussman~\cite{Su}, it follows that the discrete valuation ring $\mathbb{C}[[\epsilon]]$ can be replaced by the ring $\mathcal{O}_{0;\epsilon}$ of functions holomorphic at $\epsilon=0$.
This replacement allows us to formulate the analytic properties of the deformed screening operators more clearly and to study the properties of the operators~(\ref{eq:int-deri-ep}) in detail.
As stated in Theorems~\ref{G+G-prop} and~\ref{G-deri}, we show that, in the limit as $\epsilon$ tends to zero, the operators~(\ref{eq:int-deri-ep}) take a form analogous to (\ref{eq:int-deri-m}) on $\mathcal{W}_{p_+,p_-}$, and that the limiting operators define derivations on $\mathcal{W}_{p_+,p_-}$. 
Moreover, as stated in Theorems~\ref{G-hom} and~\ref{non-Gop}, 
properties that are important for studying the representation theory of $\mathcal{W}_{p_+,p_-}$ in \cite[Theorem~3.1, Lemma~5.1]{AdamovicD/MilasA:2010} also hold in our setting.

In Section~\ref{symWpq}, we further investigate the properties of the operators~(\ref{eq:int-deri-ep}) and show that the $\mathfrak{sl}_2$-symmetry of $\mathcal{W}_{p_+,p_-}$ established by Adamovi\'{c}-Lin-Milas~\cite{AdamovicD/LinX/MilasA:2013} and McRae-Sopinc~\cite{McRaeR/SopinV:2026} can be naturally derived from our derivations. The approach of \cite{AdamovicD/LinX/MilasA:2013} derives this symmetry from a certain hidden structure of the lattice VOA in the case $p_+=1$, whereas that of \cite{McRaeR/SopinV:2026} achieves it in the case $p_+\geq 2$ by applying the theory of commutative algebras in braided tensor categories.
Our method differs from theirs in that it is more closely related to the analytic side of the representation theory of Virasoro algebras.
More precisely, as discussed in Subsections~\ref{hidden-1} and~\ref{hidden-2}, our method depends crucially on the structure of the Felder complexes and the $\mathcal{O}_{0;\epsilon}$-integrality of the screening operators.
In our method, the $\mathfrak{sl}_2$-symmetry of $\mathcal{W}_{p_+,p_-}$ can be derived by analyzing the order in $\epsilon$ at which the commutation relations among the deformed screening operators on the Felder complex vanish.
It is therefore applicable to vertex operator (super)algebras in which similar theories of the Felder complex and the screening operators are developed.
In Section~\ref{tri-swm}, we show that our method can be applied to the triplet $W$-superalgebra \(\mathcal{SW}(m)\) introduced by Adamovi\'{c} and Milas \cite{AdamovicD/MilasA:2009-2}.
This application is possible due to the classification results of the Neveu-Schwarz Fock-modules by Iohara and Koga~\cite{IK2}, and the precise properties of Neveu-Schwarz screening operators established by Blondeau-Fournier et al.~\cite{BMRW}.
As a consequence, it follows that the Lie algebra $\mathfrak{sl}_2(\mathbb{C})$ acts on $\mathcal{SW}(m)$ by derivations. 
By this result and the structure of the Zhu-algebra determined by \cite{AdamovicD/MilasA:2009-2}, 
we determine the full automorphism group ${\rm Aut}(\mathcal{SW}(m))$ and show that it is isomorphic $PSL_2(\mathbb{C})\times \mathbb{Z}_2$.
These results are stated in Subsection~\ref{subs:swm}.

\section{Representation theories of the Virasoro algebra}
\label{Basic}
In this section, we briefly review theories of the Virasoro algebra and the rank 1 Fock modules following \cite{TW}.
We also briefly review some recent results on the Selberg integral \cite{Su}.

\subsection{Free field theory}
The Virasoro algebra $\mathfrak{Vir}$ is the Lie algebra over $\mathbb{C}$ generated by $\mathcal{L}_n(n\in \mathbb{Z})$ and $C$ (the central charge) with the relations
\begin{align*}
&[\mathcal{L}_m,\mathcal{L}_n]=(m-n)\mathcal{L}_{m+n}+\frac{m^3-m}{12}C\delta_{m+n,0},
&[\mathcal{L}_n,C]=0.
\end{align*}
$\mathfrak{Vir}$ has the triangular decomposition $\mathfrak{Vir}=\mathfrak{Vir}_+\oplus \mathfrak{Vir}_0\oplus \mathfrak{Vir}_-$ with 
\begin{align*}
&\mathfrak{Vir}_\pm=\bigoplus_{\pm n\geq 0}\mathbb{C}\mathcal{L}_n,
&\mathfrak{Vir}_0=\mathbb{C}\mathcal{L}_0\oplus \mathbb{C}C.
\end{align*}
For $h,c\in\mathbb{C}$, let $\mathbb{C}{\mid}h,c\rangle$ be the one dimensional $\mathfrak{Vir}^\geq$-module defined by
\begin{align*}
&\mathcal{L}_n{\mid}h,c\rangle=\delta_{n,0}h{\mid}h,c\rangle\qquad n\geq 0,
&C{\mid}h,c\rangle=c{\mid}h,c\rangle. 
\end{align*}
For $h,c\in \mathbb{C}$, the left Verma $\mathfrak{Vir}$-module is defined by the induced module
\begin{align*}
M(h,c)={\rm Ind}_{\mathfrak{Vir}^\geq}^{\mathfrak{Vir}}\mathbb{C}{\mid}h,c\rangle.
\end{align*} 

The Heisenberg Lie algebra 
\begin{align*}
\mathcal{H}=\bigoplus_{n\in\mathbb{Z}}\mathbb{C} a_{n}\oplus \mathbb{C} K_{\mathcal{H}}
\end{align*}
is the Lie algebra whose commutation is given by
\begin{align*} 
&[a_m,a_n]=m\delta_{m+n,0}K_{\mathcal{H}},
&[K_{\mathcal{H}},\mathcal{H}]=0.
\end{align*}
Let
\begin{align*}
\mathcal{H}^\pm&=\bigoplus_{n>0}\mathbb{C} a_{\pm n},&\mathcal{H}^0&=\mathbb{C} a_0\oplus \mathbb{C} K_{\mathcal{H}},\\
\mathcal{H}^{\geq}&=\mathcal{H}^+\oplus \mathcal{H}^0,&\mathcal{H}^{\leq }&=\mathcal{H}^-\oplus \mathcal{H}^0.
\end{align*}
For $\alpha\in\mathbb{C}$, let $\mathbb{C}{\mid}\alpha\rangle$ be the one dimensional $\mathcal{H}^\geq$-module defined by
\begin{align*}
&a_n{\mid}\alpha\rangle=\delta_{n,0}\alpha{\mid}\alpha\rangle\qquad n\geq 0,
&K_{\mathcal{H}}{\mid}\alpha\rangle={\mid}\alpha\rangle. 
\end{align*}
For $\alpha\in \mathbb{C}$, the left Fock module is defined by the induced module
\begin{align*}
F_{\alpha}={\rm Ind}_{\mathcal{H}^\geq}^{\mathcal{H}}\mathbb{C}{\mid}\alpha\rangle.
\end{align*} 
Similarly the right Fock module (or the dual Fock module) $F^\vee_{\alpha}$ is defined by
\begin{align*}
F^\vee_{\alpha}={\rm Ind}_{\mathcal{H}^\leq}^{\mathcal{H}}\mathbb{C}\langle\alpha{\mid},
\end{align*} 
where $\mathbb{C}\langle\alpha{\mid}$ is the one dimensional $\mathcal{H}^\leq$-module defined by
\begin{align*}
&\langle\alpha{\mid}a_n=\delta_{n,0}\alpha\langle\alpha{\mid},\qquad n\leq 0,
&K_{\mathcal{H}}\langle\alpha{\mid}=\langle\alpha{\mid}. 
\end{align*}
We see that the two Fock modules $F_{\alpha},F^\vee_{\alpha}$ are equipped with an inner product
\begin{align}
\label{inn-pro}
(\ \cdot\ ,\ \cdot\ )_{F_{\alpha}}:\ F^\vee_{\alpha}\times F_{\alpha}\rightarrow \mathbb{C}
\end{align}
defined by 
\begin{align*}
(\bra{\alpha},\ket{\alpha})_{F_{\alpha}}=1,\qquad (\bra{\alpha}u_1,u_2\ket{\alpha})_{F_{\alpha}}=(\bra{\alpha}u_1u_2,\ket{\alpha})_{F_{\alpha}}=(\bra{\alpha},u_1u_2\ket{\alpha})_{F_{\alpha}}
\end{align*}
for $u_1,u_2\in U(\mathcal{H})$, where $U(\mathcal{H})$ is the universal enveloping algebra of $\mathcal{H}$.
In some cases, we use the notation $\bra{\alpha}u_1u_2\ket{\alpha}=(\bra{\alpha},u_1u_2\ket{\alpha})_{F_{\alpha}}$.

Define 
\begin{equation*}
a(z):=\sum_{n\in\mathbb{Z}}a_nz^{-n-1}
\end{equation*}
This field satisfies the operator product expansion
\begin{equation*}
 a(z)  a(w)=\frac{1}{(z-w)^2}+\cdots,
\end{equation*}
where $\cdots$ denotes the regular part in $z=w$.
For $\rho\in \mathbb{C}$, we define a field
\begin{equation*}
T^{(\rho)}(z):=\frac{1}{2}:a(z)a(z):+\frac{\rho}{2}\partial a(z),
\end{equation*}
where $:\ :$ is the normal ordered product. 
This field is called the energy–momentum tensor and satisfies the operator product expansion
\begin{align*}
T^{(\rho)}(z)T^{(\rho)}(w)= \frac{c_{\rho}}{2(z-w)^4}+\frac{2T^{(\rho)}(w)}{(z-w)^2}+\frac{\partial T^{(\rho)}(w)}{z-w}+\cdots,
\end{align*}
where ``\ $\cdots$'' denote the holomorphic part about $z=w$, and we set
\begin{align*}
c_{\rho}=1-3\rho^2.
\end{align*}
Expand $T^{(\rho)}(z)$ as
$
T^{(\rho)}(z)=\sum_{n\in\mathbb{Z}}L^{(\rho)}_{n}z^{-n-2}.
$
From the above operator product expansion, the Fourier modes of 
$
\{L^{(\rho)}_{n}\}
$
generate the Virasoro algebra with central charge $C=c_{\rho}$.
Thus, by the energy-momentum tensor $T^{(\rho)}(z)$, the Fock module $F_{\alpha}$ admits the structure of a Virasoro module with the central charge $c_{\rho}$.
More precisely, there exists a $\mathfrak{Vir}$-module map 
\begin{align}
\label{eq:pi-mf}
\pi^{(\rho)}_{\alpha}:M(h^{(\rho)}_{\alpha},c_\rho)\rightarrow F_{\alpha} 
\end{align}
satisfying $\pi^{(\rho)}_{\alpha}(\mathcal{L}_n)=L^{(\rho)}_n$, $\pi^{(\rho)}_{\alpha}(\ket{h^{(\rho)}_\alpha,c_\rho})=\ket{\alpha}$, where
\begin{align}
\label{h_alpha}
h^{(\rho)}_{\alpha}:=\frac{1}{2}\alpha(\alpha-\rho)
\end{align}
which is the $L^{(\rho)}_0$-weight of $\ket{\alpha}$:
\begin{align*}
L^{(\rho)}_0{\mid}\alpha\rangle=\frac{1}{2}\alpha(\alpha-\rho){\mid}\alpha\rangle.
\end{align*}
Then, the Fock module $F_{\alpha}$ has the following $L^{(\rho)}_0$ weight decomposition
\begin{align*}
&F_{\alpha}=\bigoplus_{n\in \mathbb{Z}_{\geq 0}}F_{\alpha}[h^{(\rho)}_{\alpha}+n],
&F_{\alpha}[h^{(\rho)}_{\alpha}+n]:=\{v\in F_{\alpha}\setminus\{0\}\mid L^{(\rho)}_0v=(h^{(\rho)}_{\alpha}+n)v\},
\end{align*}
where each weight space $F_{\alpha}[h^{(\rho)}_{\alpha}+n]$, $n>0$ admits a basis
\begin{align}
\label{a-lambda}
\{a_{-\lambda}{\mid}\alpha\rangle\mid \lambda\vdash n\}
\end{align}
with $a_{-\lambda}=a_{-\lambda_k}\cdots a_{-\lambda_1}$ for a partition $\lambda=(\lambda_1,\dots,\lambda_k)$, $\lambda_i>0$.
\begin{dfn}
Define the conformal vector
\begin{align}
\label{eq:com-vec}
T^{(\rho)}:=\frac{1}{2}(a^2_{-1}+\rho a_{-2})\ket{0}\in F_0[2].
\end{align}
Then, the Fock module ${F}_0$ carries the structure of a $\mathbb{Z}_{\geq 0}$-graded vertex operator algebra, with
\begin{align*}
&Y(\ket{0};z)={\rm id},\ \ \ \ \ \ Y(a_{-1}\ket{0};z)=a(z),\ \ \ \ \ \ Y(T^{(\rho)};z)=T^{(\rho)}(z).
\end{align*}
\end{dfn}
We denote the above vertex operator algebra by $\mathcal{F}^{}_{\rho}$.
For any Fock module $F_{\alpha}$ in $\mathcal{F}_{\rho}$, we denote it by $F^{(\rho)}_{\alpha}$ to emphasize the Virasoro module structure.

To define an appropriate $\mathcal{F}_{\rho}$-module structure for the right Fock modules, we introduce the following definition.
\begin{dfn}
 Let $V$ be a vertex operator algebra. 
 For a $V$-module $M$, we define its contragredient as
\begin{align*}
&M^*=\bigoplus_{h\in \mathbb{C}}{\rm Hom}(M{[h]},\mathbb{C}),
\end{align*}
where
$
M{[h]}=\{m\in M\ |\ (L_0-h)^nm=0,\ n\gg 0\}.
$
\end{dfn}
It is a known fact that the contragredient space $M^*$ admits a module action \cite{FHL} defined by
\begin{align*}
\langle Y(v,z)\phi^*,m\rangle_{M}=\langle \phi^*,Y^{\rm opp}(v,z)m\rangle_{M}
\end{align*}
where $\langle\ \cdot\ ,\ \cdot\ \rangle_M$ is the canonical pairing, and $Y^{\rm opp}(v,z)$ is the opposite action defined by
\begin{align}
\label{eq:opp}
Y^{\rm opp}(v,z)=Y(e^{zL_1}(-z^{-2})^{L_0}v,z^{-1}).
\end{align}
We see that the right Fock module $F^\vee_{\alpha}$ is isomorphic to the contragredient space of $F^{(\rho)}_\alpha$
\begin{align*}
F^*_{\alpha}=\bigoplus_{n\geq 0}{\rm Hom}(F^{(\rho)}_\alpha[h^{(\rho)}_\alpha+n],\mathbb{C}).
\end{align*}
Then by (\ref{eq:opp}), $F^*_{\alpha}$ admits the structure of a $\mathcal{F}_{\rho}$-module. From (\ref{eq:com-vec}), (\ref{eq:opp}), we see that 
\begin{align}
\label{a-opp}
&a^{\rm opp}_n=\rho\delta_{n,0}-a_{-n},
&(L^{(\rho)}_{n})^{\rm opp}=L^{(\rho)}_{-n}
\end{align}
on $F^*_{\alpha}$. In particular, by the simplicity of $F^*_{\alpha}$ and $F_{\alpha}$, the contragredient $F^*_{\alpha}$ is isomorphic to $F^{(\rho)}_{\rho-\alpha}$ as a $\mathcal{F}_{\rho}$-module.
Hereafter we use the normalization $\langle \ket{\rho-\alpha},\ket{\alpha}\rangle_{F_{\alpha}}=1$ which is compatible with $(\bra{\alpha},\ket{\alpha})_{F_{\alpha}}=\braket{\alpha|\alpha}=1$.

\subsection{Selberg integrals}
\label{sub:sus}
 As shown in Figure \ref{one-dimensional-cycle}, we take three paths $\varDelta,\mathcal{C}_1, \mathcal{C}_0$, and define the 1-chain 
 \begin{align}
 \label{twistonedel}
 [\Delta_1(\alpha,\beta)]:=\frac{-1}{1-e^{2\pi i\alpha}}\mathcal{C}_0+\varDelta+\frac{1}{1-e^{2\pi i\beta}}\mathcal{C}_1.
 \end{align}
 It is well known that this 1-chain defines a regularization of the Euler beta integrals (cf.~\cite{AK}):
 \begin{align*}
 \int_{[\Delta_1(\alpha,\beta)]}u^\alpha(1-u)^\beta{\rm d}u=\frac{\Gamma(\alpha+1)\Gamma(\beta+1)}{\Gamma(\alpha+\beta+2)}.
 \end{align*}
 
 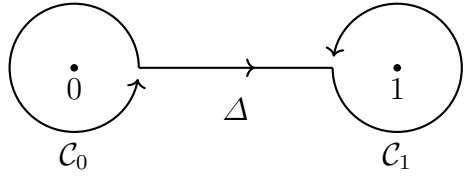
\begin{figure}[htbp]
  \centering
\begin{tikzpicture}[scale=0.85]
\draw[thick, ->] (6.5,0) arc (0:350:1);
\draw [thick](6.5,0)--(9.5,0);
\draw [->,thick](8.25,0)--(8.3,0);
\draw[thick, ->] (9.5,0) arc (-180:170:1);
\filldraw[black](5.5,0)circle (0.5mm) node[below]{$0$};
\filldraw[black](10.5,0)circle (0.5mm) node[below]{$1$};
\filldraw[](5.5,-1) node[below]{$\mathcal{C}_0$};
\filldraw[](10.5,-1) node[below]{$\mathcal{C}_1$};
\filldraw[](8,-0.3) node[below]{$\varDelta$};
\end{tikzpicture}
\caption{One-dimensional twisted cycle associated to $u^\alpha(1-u)^\beta$.
}
\label{one-dimensional-cycle}
\end{figure}

 The following Selberg integral is known as a higher-dimensional generalization of the Euler beta integral (cf.~\cite{FoW,Se}).
\begin{prop}[\cite{Se}]
For $\Re \alpha,\Re \beta>-1$, $\Re \gamma>0$, the integral 
\begin{align*}
S_n(\alpha,\beta,\gamma)=\int_{\Delta_n}\prod_{i=1}^ny^{\alpha}_i(1-y_i)^\beta\prod_{1\leq i\neq j\leq n}(y_i-y_j)^\gamma{\rm d}y_1\cdots{\rm d}y_n
\end{align*}
converges, and its value is given by the following explicit formula
\begin{align*}
S_n(\alpha,\beta,\gamma)=\frac{1}{n!}\prod_{i=1}^n\frac{\Gamma(1+i\gamma)\Gamma(1+\alpha+(i-1)\gamma)\Gamma(1+\beta+(i-1)\gamma)}{\Gamma(1+\gamma)\Gamma(2+\alpha+\beta+(n+i-2)\gamma)},
\end{align*}
where $\Delta_n=\{(y_1,y_2,\dots,y_n)\in \mathbb{R}^n\ |\ 1>y_1>y_1>\cdots>y_n>0\}$.
\end{prop}
For $f\in \mathbb{C}[y^{\pm 1}_i|\ 1\leq i\leq n]$, let 
\begin{align*}
S_n[f](\alpha,\beta,\gamma)=\int_{\Delta_n}\prod_{i=1}^ny^{\alpha}_i(1-y_i)^\beta\prod_{1\leq i\neq j\leq n}(y_i-y_j)^\gamma f(\bm{y}){\rm d}y_1\cdots{\rm d}y_n.
\end{align*}
Note that $S_n[1](\alpha,\beta,\gamma)=S_n(\alpha,\beta,\gamma)$. The following analytic result is known for this integral.
\begin{prop}[\cite{TK}]
\label{TK-thm}
Define a collection of hyperplanes
\begin{align*}
  \begin{aligned}
     \mathcal{A}_n=&\bigcup_{j=1}^n\{(\alpha,\beta,\gamma)\in\mathbb{C}^3\ |\ j(\alpha+(j-1)\gamma)\in \mathbb{Z}\}\\
     &\cup\bigcup_{j=1}^n\{(\alpha,\beta,\gamma)\in\mathbb{C}^3\ |\ j(\beta+(j-1)\gamma)\in \mathbb{Z}\}\\
     &\cup \bigcup_{j=1}^{n-1}\{(\alpha,\beta,\gamma)\in\mathbb{C}^3\ |\ j(j+1)\gamma\in \mathbb{Z}\}.
    \end{aligned}
  \end{align*}
  Then there exists a twisted cycle $[\Delta_n(\alpha,\beta,\gamma)]$, obtained as a regularization of $\Delta_n$, such that
\begin{equation}
\label{eq:IntDelta}
\int_{[\Delta_n(\alpha,\beta,\gamma)]}\prod_{i=1}^ny^{\alpha}_i(1-y_i)^\beta\prod_{1\leq i\neq j\leq n}(y_i-y_j)^\gamma f(\bm{y}){\rm d}\bm{y}=S_n[f](\alpha,\beta,\gamma)
\end{equation}
for $(\alpha,\beta,\gamma)\in \mathbb{C}^3\setminus\mathcal{A}_n$. 
In particular, as the function of $(\alpha,\beta,\gamma)$, the integral $S_n[f](\alpha,\beta,\gamma)$ admits the analytical continuation to the domain $\mathbb{C}^3\setminus\mathcal{A}_n$.
\end{prop}
 
We refer to $(\alpha,\beta,\gamma)$ as parameters of the twisted cycle $[\Delta_n(\alpha,\beta,\gamma)]$.
 
\begin{Remark}
\label{del-sus0}
As can be seen from the case of the one-dimensional cycle \eqref{twistonedel}, 
it is known that the twisted cycle $[\Delta_n(\alpha,\beta,\gamma)]$ depends on the monodromy of the multivalued function of the integrand of $S_n[1](\alpha,\beta,\gamma)$ \cite{AK,TK}.
The condition $(\alpha,\beta,\gamma)\notin\mathcal{A}_n$ follows naturally from the construction of $[\Delta_n(\alpha,\beta,\gamma)]$.
\end{Remark} 

The result of the analytic continuation for $S_n[f](\alpha,\beta,\gamma)$ can be refined as in the following theorem.
\begin{thm}[\cite{Su}]
\label{sus-prop}
Define a collection of hyperplanes
\begin{align*}
  \begin{aligned}
     \widehat{\mathcal{A}}_n=&\bigcup_{j=1}^n\{(\alpha,\beta,\gamma)\in\mathbb{C}^3\ |\ \alpha+(j-1)\gamma\in \mathbb{Z}\}\\
     &\cup\bigcup_{j=1}^n\{(\alpha,\beta,\gamma)\in\mathbb{C}^3\ |\ \beta+(j-1)\gamma\in \mathbb{Z}\}\\
     &\cup \bigcup_{j=1}^{n-1}\{(\alpha,\beta,\gamma)\in\mathbb{C}^3\ |\ (j+1)\gamma\in \mathbb{Z}\}.
    \end{aligned}              
  \end{align*}
Suppose that $f\in \mathbb{C}[y^{\pm 1}_i|\ 1\leq i\leq n]^{\mathfrak{S}_n}$, where $\mathbb{C}[y^{\pm 1}_i|\ 1\leq i\leq n]^{\mathfrak{S}_n}$ is the ring of symmetric Laurent polynomials in $y_1,\dots,y_n$.
Then, as the function of $(\alpha,\beta,\gamma)$, the integral $S_n[f](\alpha,\beta,\gamma)$ admits the analytical continuation to the domain 
\begin{align*}
\mathbb{C}^3\setminus\widehat{\mathcal{A}}_n.
\end{align*}
\end{thm}

\begin{Remark}
From the constructions in \cite{AK,TK}, the cycle $[\Delta_n(\alpha,\beta,\gamma)]$ is not well-defined for all $(\alpha,\beta,\gamma)\in \mathbb{C}^3\setminus\widehat{\mathcal{A}}_n$. 
Nevertheless, by Theorem~\ref{sus-prop}, the value of the integral $S_n[f](\alpha,\beta,\gamma)$ analytically continues to the whole space of $\mathbb{C}^3\setminus\widehat{\mathcal{A}}_n$. 
Therefore, in what follows we regard $[\Delta_n(\alpha,\beta,\gamma)]$ as well-defined for arbitrary $(\alpha,\beta,\gamma)\in \mathbb{C}^3\setminus\widehat{\mathcal{A}}_n$, and use the notation in (\ref{eq:IntDelta}) for all $(\alpha,\beta,\gamma)\in \mathbb{C}^3\setminus\widehat{\mathcal{A}}_n$.
As explained in Subsections~\ref{cons-main0} and \ref{hidden-2}, this property of analytic continuation is very useful for discussing the $\mathcal{O}$-integrality of screening operators and the $\mathfrak{sl}_2$-symmetry of the derivations on $\mathcal{W}_{p_+,p_-}$ $(p_+\geq 2)$.
 \label{del-sus2}
 \end{Remark}

\subsection{Screening operators}
\label{subsc}
We define the conjugate $\hat{a}$ of $a_0$ to be the element satisfying the following relation
\begin{equation}
\label{conjugate}
[a_m,\hat{a}]=\delta_{m,0}{\rm id}.
\end{equation}
For each $\gamma \in \mathbb{C}$, $e^{\gamma\hat{a}}$ defines a Heisenberg weight shifting map $e^{\gamma\hat{a}}: F_{\beta}\rightarrow F_{\beta+\gamma}$. From now on, we identify $e^{\gamma\hat{a}}\cdot F_{\beta}$ with $F_{\beta+\gamma}$.

We introduce a free scalar field $\phi(z)$, which is a formal primitive of $a(z)$:
\begin{equation}
\label{eq:scalar}
\phi(z)=\hat{a}+a_0{\rm log}z-\sum_{n\neq 0}\frac{a_n}{n}z^{-n}.
\end{equation}
The scalar field $\phi(z)$ satisfies the operator product expansion
\begin{equation}
\label{eq:bosonf}
\phi(z)\phi(w)={\rm log}(z-w)+\cdots.
\end{equation} 
For $\alpha\in\mathbb{C}$, consider the field $:e^{\alpha\phi(z)}:$ which equals the vertex operator
\begin{align}
\label{Y-ver}
Y(\ket{\alpha},z)=e^{\alpha\hat{a}}z^{\alpha a_{0}}\overline{Y}(\ket{\alpha},z),
\end{align}
where $z^{\alpha a_{0}}=e^{\alpha a_{0}{\rm log}z}$ and
\begin{align*}
\overline{Y}(\ket{\alpha},z):=\prod_{n\geq 1}{\rm exp}\Bigl({\alpha\frac{a_{-n}}{n}z^{n}}\Bigr){\rm exp}\Bigl({-\alpha\frac{a_n}{n}z^{-n}}\Bigr).
\end{align*}
From (\ref{eq:bosonf}), the composition of $m$ vertex operators is given by
\begin{align}
\label{eq:opeVV}
\begin{aligned}
Y(\ket{\alpha_1},z_1)\cdots Y(\ket{\alpha_m},z_m)
&=e^{(\sum_{i=1}^m\alpha_i)\hat{a}}\prod_{i=1}^mz_i^{\alpha_i a_0}\prod_{1\leq i\neq j\leq m}(z_i-z_j)^{\frac{\alpha_i\alpha_j}{2}}\\
&\qquad\qquad \cdot :\prod_{i=1}^m\overline{Y}(\ket{\alpha_i},z_i):
\end{aligned}
\end{align}

Fix $\rho\in \mathbb{C}$. Let $\rho_\pm$ $(\Re \rho_+\geq \Re \rho_-)$ be the solutions of 
\begin{align}
\label{second-order}
x^2-\rho x-2=0.
\end{align}
For $r,s\in\mathbb{Z}$, we set
\begin{align}
h^{(\rho)}_{r,s}&:=\frac{r^2-1}{8}\rho^2_+-\frac{rs-1}{2}+\frac{s^2-1}{8}\rho^2_-,\nonumber\\
\alpha^{(\rho)}_{r,s}&:=\frac{1-r}{2}\rho_++\frac{1-s}{2}\rho_-.
\label{alpha-rs-rho}
\end{align}
Note that $h^{(\rho)}_{\alpha^{(\rho)}_{r,s}}=h^{(\rho)}_{r,s}$.
We set
\begin{align}
\label{F-rho}
F^{(\rho)}_{r,s}:=F^{(\rho)}_{\alpha^{(\rho)}_{r,s}},\qquad r,s\in \mathbb{Z}.
\end{align}
Define two fields
\begin{align*}
\mathcal{Q}^{(\rho)}_\pm(z)=Y(\ket{\rho_\pm},z).
\end{align*} 
The conformal weights of these fields are $h^{(\rho)}_{\rho_{\pm}}=1$:
\begin{align*}
T^{(\rho)}(z)\mathcal{Q}^{(\rho)}_\pm(w)&=\frac{\mathcal{Q}^{(\rho)}_\pm(w)}{(z-w)^2}+\frac{\partial_{w}\mathcal{Q}^{(\rho)}_\pm(w)}{z-w}+\cdots\\
&=\partial_{w}\Bigl(\frac{\mathcal{Q}^{(\rho)}_\pm(w)}{z-w}\Bigr)+\cdots.
\end{align*}
Therefore the zero modes of the fields $\mathcal{Q}^{(\rho)}_\pm(z)$
\begin{align*}
\begin{aligned}
&\mathcal{Q}^{(\rho)}_{+}:={\rm Res}_{z=0}\mathcal{Q}^{(\rho)}_\pm(z){\rm d}z\ :\ F^{(\rho)}_{1,k}\rightarrow F^{(\rho)}_{-1,k},\ \ k\in\mathbb{Z}\\
&\mathcal{Q}^{(\rho)}_{-}:={\rm Res}_{z=0}\mathcal{Q}^{(\rho)}_\pm(z){\rm d}z\ :\ F^{(\rho)}_{k,1}\rightarrow F^{(\rho)}_{k,-1},\ \ k\in\mathbb{Z}
\end{aligned}
\end{align*}
commute with every Virasoro mode. 
The above zero modes are called screening operators.

Let $r,s\in \mathbb{Z}_{\geq 2}$ and $k\in \mathbb{Z}$. For $u_r\in F^{(\rho)}_{r,k}$, $u_s\in F^{(\rho)}_{k,s}$, $\psi^\vee_r\in F^\vee_{\alpha^{(\rho)}_{-r,k}}$, and $\psi^\vee_s\in F^\vee_{\alpha^{(\rho)}_{k,-s}}$, consider the correlation functions
\begin{align}
\label{Phi-cor}
\begin{aligned}
&\Phi_+(\psi^\vee_r,u_r;\bm{z}):=\bigl( \psi^\vee_r,\mathcal{Q}^{(\rho)}_+(z_1)\mathcal{Q}^{(\rho)}_+(z_2)\cdots \mathcal{Q}^{(\rho)}_+(z_r)u_r\bigr)_{F^{(\rho)}_{-r,k}},\\
&\Phi_-(\psi^\vee_s,u_s;\bm{z}):=\bigl( \psi^\vee_s,\mathcal{Q}^{(\rho)}_-(z_1)\mathcal{Q}^{(\rho)}_-(z_2)\cdots \mathcal{Q}^{(\rho)}_-(z_s)u_s\bigr)_{F^{(\rho)}_{k,-s}}.
\end{aligned}
\end{align}
Then by definition, we have
\begin{align}
\label{Phi-cor-ex}
\begin{aligned}
\Phi_+(\psi^\vee_r,u_r;\bm{z})=&\Phi_+(\bra{\alpha^{(\rho)}_{-r,k}},\ket{\alpha^{(\rho)}_{r,k}};\bm{z})\cdot\bigl( \psi^\vee_r,:\prod_{i=1}^r\overline{Y}(\ket{\rho_+},z_i):u_r\bigr)_{F^{(\rho)}_{-r,k}},\\
\Phi_-(\psi^\vee_s,u_s;\bm{z})=&\Phi_-(\bra{\alpha^{(\rho)}_{k,-s}},\ket{\alpha^{(\rho)}_{k,s}};\bm{z})\cdot\bigl( \psi^\vee_s,:\prod_{i=1}^s\overline{Y}(\ket{\rho_-},z_i):u_s\bigr)_{F^{(\rho)}_{k,-s}}.
\end{aligned}
\end{align}
It is known that the second term in each of (\ref{Phi-cor-ex}) is expressed by means of Jack symmetric polynomials (cf. \cite{TW}).
From (\ref{eq:opeVV}), we see that
\begin{align}
\label{eq:cor-phi}
\begin{aligned}
\Phi_+(\bra{\alpha^{(\rho)}_{-r,k}},\ket{\alpha^{(\rho)}_{r,k}};\bm{z})&=\prod_{i=1}^rz^{(1-r)\frac{\rho^2_+}{2}+k-1}_i\prod_{1\leq i\neq j\leq r}(z_i-z_j)^{\frac{\rho^2_+}{2}}\\
&=\prod_{i=1}^rz^{k-1}_i\prod_{1\leq i\neq j\leq r}\Bigl(1-\frac{z_j}{z_i}\Bigr)^{\frac{\rho^2_+}{2}},\\
\Phi_-(\bra{\alpha^{(\rho)}_{k,-s}},\ket{\alpha^{(\rho)}_{k,s}};\bm{z})&=\prod_{i=1}^sz^{(1-s)\frac{\rho^2_-}{2}+k-1}_i\prod_{1\leq i\neq j\leq s}(z_i-z_j)^{\frac{\rho^2_-}{2}}\\
&=\prod_{i=1}^sz^{k-1}_i\prod_{1\leq i\neq j\leq s}\Bigl(1-\frac{z_j}{z_i}\Bigr)^{\frac{\rho^2_-}{2}}.
\end{aligned}
\end{align}
For these correlation functions, we consider the following change of variables:
\begin{align}
\label{patic-zU}
\begin{aligned}
&\Phi_+(\bra{\alpha^{(\rho)}_{-r,k}},\ket{\alpha^{(\rho)}_{r,k}};z,zy_1,\dots,zy_{r-1})=z^{rk-r}U^{(\rho_+)}_{r-1}(\bm{y}),\\
&\Phi_-(\bra{\alpha^{(\rho)}_{k,-s}},\ket{\alpha^{(\rho)}_{k,s}};z,zy_1,\dots,zy_{s-1})=z^{ks-s}U^{(\rho_-)}_{s-1}(\bm{y}),
\end{aligned}
\end{align}
where 
\begin{align*}
\begin{aligned}
&U^{(\rho_\pm)}_{n-1}(\bm{y})=\prod_{i=1}^{n-1}y^{(1-n)\frac{\rho^2_\pm}{2}}_i(1-y_i)^{\rho^2_\pm}\prod_{1\leq i\neq j\leq n-1}(y_i-y_j)^{\frac{\rho^2_\pm}{2}}.
\end{aligned}
\end{align*}
Thus, using (\ref{Phi-cor-ex}), (\ref{patic-zU}), and Proposition~\ref{TK-thm}, we can apply the twisted cycle $[\Delta_n]$ to (\ref{Phi-cor}).
Then we define the fields
\begin{align*}
&\mathcal{Q}^{(\rho);[r]}_{+}(z)\in {\rm Hom}_{\mathbb{C}}(F^{(\rho)}_{r,k},F^{(\rho)}_{-r,k})[[z,z^{-1}]],\ \ r\in \mathbb{Z}_{\geq 2},\ \ k\in\mathbb{Z},\\
&\mathcal{Q}^{(\rho);[s]}_{-}(z)\in {\rm Hom}_{\mathbb{C}}(F^{(\rho)}_{k,s},F^{(\rho)}_{k,-s})[[z,z^{-1}]],\ \ s\in\mathbb{Z}_{\geq 2},\ \ k\in\mathbb{Z},
\end{align*}
as follows
\begin{equation}
\label{Tsuchiya-Kanie0}
\begin{split}
&\mathcal{Q}^{(\rho);[r]}_{+}(z)=\int_{[\Delta^{(\rho_+)}_{r-1}]}\mathcal{Q}^{(\rho)}_+(z)\mathcal{Q}^{(\rho)}_+(zy_{1})\mathcal{Q}^{(\rho)}_+(zy_2)\cdots \mathcal{Q}^{(\rho)}_+(zy_{r-1})z^{r-1}{\rm d}y_1\cdots{\rm d}y_{r-1},\\
&\mathcal{Q}^{(\rho);[s]}_{-}(z)=\int_{[\Delta^{(\rho_-)}_{s-1}]}\mathcal{Q}^{(\rho)}_-(z)\mathcal{Q}^{(\rho)}_-(zy_{1})\mathcal{Q}^{(\rho)}_-(zy_2)\cdots \mathcal{Q}^{(\rho)}_-(zy_{s-1})z^{s-1}{\rm d}y_1\cdots{\rm d}y_{s-1},
\end{split}
\end{equation}
where we set
\begin{align}
[\Delta^{(\sigma)}_{l-1}]
:=[\Delta_{l-1}\bigl((1-l)\frac{\sigma^2}{2},\sigma^2,\frac{\sigma^2}{2}\bigr)].
\label{eq:twist-rho}
\end{align}
These fields satisfy the following operator product expansions \cite{TW}
\begin{align*}
&T^{(\rho)}(z)\mathcal{Q}^{(\rho);[\bullet]}_{\pm}(w)=\frac{\mathcal{Q}^{(\rho);[\bullet]}_{\pm}(w)}{(z-w)^2}+\frac{\partial_{w}\mathcal{Q}^{(\rho);[\bullet]}_{\pm}(w)}{z-w}+\cdots.
\end{align*}
In particular the following proposition holds.
\begin{prop}[\cite{FF,IK,TK,TW}]
The zero modes
\begin{align*}
&\mathcal{Q}^{(\rho);[r]}_{+}:={\rm Res}_{z=0}\mathcal{Q}^{(\rho);[r]}_{+}(z){\rm d}z \in {\rm Hom}_{\mathbb{C}}(F^{(\rho)}_{r,k},F^{(\rho)}_{-r,k}), \ r\in \mathbb{Z}_{\geq 2},\ \ k\in\mathbb{Z},\\
&\mathcal{Q}^{(\rho);[s]}_{-}:={\rm Res}_{z=0}\mathcal{Q}^{(\rho);[s]}_{-}(z){\rm d}z \in {\rm Hom}_{\mathbb{C}}(F^{(\rho)}_{k,s},F^{(\rho)}_{k,-s}),\ \ s\in\mathbb{Z}_{\geq 2},\ \ k\in\mathbb{Z}
\end{align*}
commute with every Virasoro mode of $\mathcal{F}_{\rho}\mathchar`-{\rm Mod}$. 
\label{com-TK}
\end{prop}
These zero modes are also called {\rm screening operators}.
Hereafter, we set
\begin{align*}
&\mathcal{Q}^{(\rho);[1]}_{+}=\mathcal{Q}^{(\rho)}_{+},
&\mathcal{Q}^{(\rho);[1]}_{-}=\mathcal{Q}^{(\rho)}_{-}.
\end{align*}
The following proposition follows immediately from (\ref{a-opp}) and the definition (\ref{Y-ver}) of the vertex operators.
\begin{prop}
Let $r,s\in \mathbb{Z}_{\geq 1}$ and $k\in \mathbb{Z}$. For any $u_r\in F^{(\rho)}_{r,k}$, $u_s\in F^{(\rho)}_{k,s}$, $\psi^*_r\in F^*_{\alpha^{(\rho)}_{-r,k}}$, and $\psi^*_s\in F^*_{\alpha^{(\rho)}_{k,-s}}$, we have
\begin{align*}
\langle\psi_r,\mathcal{Q}^{(\rho);[r]}_+u_r\rangle_{{F}^{(\rho)}_{-r,k}}&=\langle\mathcal{Q}^{(\rho);[r]}_+\psi_r,u_r\rangle_{{F}^{(\rho)}_{r,k}},\\
\langle\mathcal{Q}^{(\rho);[s]}_-\psi_s,u_s\rangle_{{F}^{(\rho)}_{k,-s}}&=\langle\psi_s,\mathcal{Q}^{(\rho);[s]}_-u_s\rangle_{{F}^{(\rho)}_{k,s}}.
\end{align*}
\label{dual-scop}
\end{prop}
 The following proposition is essentially due to \cite{AdamovicD/MilasA:2008,NT,Nak,TW}.
 \begin{prop}
 \label{prop:qr-qs}
 Let $r,s\in \mathbb{Z}_{\geq 1}$. Then, we have
 $
 \lbrack\mathcal{Q}^{(\rho);[r]}_{+},\mathcal{Q}^{(\rho);[s]}_{-}\rbrack=0.
 $
 \end{prop}
 \begin{proof}
 In what follows, we use the abbreviations $\mathcal{Q}_{\pm}(x)=\mathcal{Q}^{(\rho)}_{\pm}(x)$ and the notation
 \begin{align*}
 \mathcal{Q}_{0}(x):=Y(\ket{\rho_++\rho_-},x).
 \end{align*}
 To prove the commutativity $\lbrack\mathcal{Q}^{(\rho);[r]}_{+},\mathcal{Q}^{(\rho);[s]}_{-}\rbrack=0$, it suffices to show the commutator
 \begin{align}
 \label{vac-fin-tderi}
 \begin{aligned}
 &{\rm Res}_{z=w_1,\dots,w_s}\mathcal{Q}^{(\rho);[r]}_{+}(z)\mathcal{Q}_{-}(w_1)\cdots \mathcal{Q}_{-}(w_s){\rm d}z{\rm d}\bm{w}\\
 &\qquad =\lbrack\mathcal{Q}^{(\rho);[r]}_{+},\mathcal{Q}_{-}(w_1)\cdots \mathcal{Q}_{-}(w_s)\rbrack{\rm d}\bm{w}
 \end{aligned}
 \end{align}
 can be written as a total derivative form with respect to $\bm{w}$, where the $|w_i|$ are taken to be sufficiently small and we denote ${\rm d}\bm{w}={\rm d}w_1\wedge\cdots \wedge {\rm d}w_{s}$.
 Then we consider the composition
 \begin{align}
 \label{eq:q-qs}
\frac{z^{r-1}}{2\pi i}\mathcal{Q}_{+}(z)\mathcal{Q}_{+}(zy_1)\cdots \mathcal{Q}_{+}(zy_{r-1})\mathcal{Q}_{-}(w_1)\cdots \mathcal{Q}_{-}(w_s).
 \end{align}
 Note the operator product expansion
 \begin{align}
 \label{eq:op-nyt}
 \mathcal{Q}_{+}(z)\mathcal{Q}_{-}(w)=\frac{:Y(\ket{\rho_+},z)Y(\ket{\rho_-},w):}{(z-w)^2}.
 \end{align}
 Using (\ref{eq:op-nyt}), and taking residues at $z = w_1, w_2, \dots, w_s$ in (\ref{eq:q-qs}), we obtain
 \begin{align}
 \begin{aligned}
 &\frac{\rho_+}{\rho}\sum_{i=1}^{s}w^{r-1}_i\mathcal{Q}_{-}(w_1)\cdots \mathcal{Q}_{-}(w_{i-1})\bigl(\partial_{w_i}\mathcal{Q}_{0}(w_i)\bigr)\mathcal{Q}_{+}(w_iy_1)\cdots \mathcal{Q}_{+}(w_iy_{r-1})\\
 &\qquad \qquad \qquad \qquad \cdot\mathcal{Q}_{-}(w_{i+1})\cdots \mathcal{Q}_{-}(w_s)\\
 &+(r-1)\sum_{i=1}^{s}w^{r-2}_i\mathcal{Q}_{-}(w_1)\cdots \mathcal{Q}_{-}(w_{i-1})\mathcal{Q}_{0}(w_i)\mathcal{Q}_{+}(w_iy_1)\cdots \mathcal{Q}_{+}(w_iy_{r-1})\\
 &\qquad \qquad \qquad \qquad \cdot\mathcal{Q}_{-}(w_{i+1})\cdots \mathcal{Q}_{-}(w_s)\\
 &+\sum_{i=1}^{s}w^{r-1}_i\mathcal{Q}_{-}(w_1)\cdots \mathcal{Q}_{-}(w_{i-1})\mathcal{Q}_{0}(w_i)\partial_{w_i}\bigl(\mathcal{Q}_{+}(w_iy_1)\cdots \mathcal{Q}_{+}(w_iy_{r-1})\bigr)\\
 &\qquad \qquad \qquad \qquad \cdot\mathcal{Q}_{-}(w_{i+1})\cdots \mathcal{Q}_{-}(w_s)\\
 \end{aligned}
 \label{vac-01}
 \end{align}
 (see \cite[Proposition 2.1]{NT}). Note that
\begin{align*}
\partial_{w_i}\mathcal{Q}_{+}(w_iy_j)=\frac{y_j}{w_i}\partial_{y_j}\mathcal{Q}_{+}(w_iy_j).
\end{align*}
Then, integrating the third sum in (\ref{vac-01}) with respect to $\bm{y}$ and applying the generalized Stokes' theorem (cf.~\cite{AK}), we obtain
\begin{align*}
\begin{aligned}
&\int_{[\Delta^{(\rho_+)}_{r-1}]}\sum_{i=1}^s\sum_{j=1}^{r-1} w^{r-2}_i(-1)^{j+1}y_j{\rm d}_{\bm{y}}\biggl\{\mathcal{Q}_{-}(w_1)\cdots \mathcal{Q}_{-}(w_{i-1})\mathcal{Q}_{0}(w_i)\\
&\qquad \qquad \qquad \cdot \mathcal{Q}_{+}(w_iy_1)\cdots \partial_{y_j}\mathcal{Q}_{+}(w_iy_j)\cdots\mathcal{Q}_{+}(w_iy_{r-1}) \mathcal{Q}_{-}(w_{i+1})\cdots \mathcal{Q}_{-}(w_s)\\
&\qquad \qquad \qquad \qquad \qquad \qquad \cdot  {\rm d}y_1\wedge\cdots\overset{\vee}{{\rm d}y_j}\cdots \wedge{\rm d}y_{r-1}\biggr\}\\
&=-(r-1)\sum_{i=1}^{s}\int_{[\Delta^{(\rho_+)}_{r-1}]}w^{r-2}_i\mathcal{Q}_{-}(w_1)\cdots \mathcal{Q}_{-}(w_{i-1})\mathcal{Q}_{0}(w_i)\\
&\qquad \qquad \qquad \qquad\qquad\cdot\mathcal{Q}_{+}(w_iy_1)\cdots \mathcal{Q}_{+}(w_iy_{r-1})\mathcal{Q}_{-}(w_{i+1})\cdots \mathcal{Q}_{-}(w_s){\rm d}\bm{y},
\end{aligned}
\end{align*}
where we denote ${\rm d}\bm{y}={\rm d}y_1\wedge\cdots \wedge {\rm d}y_{r-1}$.
This cancels with the integral of the second sum in (\ref{vac-01}).
Then the commutator (\ref{vac-fin-tderi}) can be written as
 \begin{align}
 \begin{aligned}
 &\frac{\rho_+}{\rho}\sum_{i=1}^{s}\int_{[\Delta^{(\rho_+)}_{r-1}]}w^{r-1}_i\mathcal{Q}_{-}(w_1)\cdots \mathcal{Q}_{-}(w_{i-1})\bigl(\partial_{w_i}\mathcal{Q}_{0}(w_i)\bigr)\\
 &\qquad \qquad \qquad \qquad \cdot \mathcal{Q}_{+}(w_iy_1)\cdots \mathcal{Q}_{+}(w_iy_{r-1})\mathcal{Q}_{-}(w_{i+1})\cdots \mathcal{Q}_{-}(w_s){\rm d}\bm{y}{\rm d}\bm{w}\\
 &=\frac{\rho_+}{\rho}\int_{[\Delta^{(\rho_+)}_{r-1}]}\sum_{i=1}^{s}(-1)^{i+1}\\
 &\qquad \qquad \cdot{\rm d}_{\bm{w}}\biggl\{w^{r-1}_i\mathcal{Q}_{-}(w_1)\cdots \mathcal{Q}_{-}(w_{i-1})\mathcal{Q}_{0}(w_i)\mathcal{Q}_{+}(w_iy_1)\cdots \mathcal{Q}_{+}(w_iy_{r-1})\\
 &\qquad\qquad \qquad\qquad \qquad \cdot\mathcal{Q}_{-}(w_{i+1})\cdots \mathcal{Q}_{-}(w_s) {\rm d}w_1\wedge\cdots\overset{\vee}{{\rm d}w_i}\cdots \wedge{\rm d}w_{s}\biggr\}{\rm d}\bm{y}\\
 &-(r-1)\sum_{i=1}^{s}\frac{\rho_+}{\rho}\int_{[\Delta^{(\rho_+)}_{r-1}]}w^{r-2}_i\mathcal{Q}_{-}(w_1)\cdots \mathcal{Q}_{-}(w_{i-1})\mathcal{Q}_{0}(w_i)\\
 &\qquad \qquad \qquad \qquad \cdot \mathcal{Q}_{+}(w_iy_1)\cdots \mathcal{Q}_{+}(w_iy_{r-1})\mathcal{Q}_{-}(w_{i+1})\cdots \mathcal{Q}_{-}(w_s){\rm d}\bm{y}{\rm d}\bm{w}\\
 &-\frac{\rho_+}{\rho}\sum_{i=1}^{s}\int_{[\Delta^{(\rho_+)}_{r-1}]}w^{r-1}_i\mathcal{Q}_{-}(w_1)\cdots \mathcal{Q}_{-}(w_{i-1})\mathcal{Q}_{0}(w_i)\\
 &\qquad \qquad \qquad \qquad \cdot \partial_{w_i}\bigl(\mathcal{Q}_{+}(w_iy_1)\cdots \mathcal{Q}_{+}(w_iy_{r-1})\bigr)\mathcal{Q}_{-}(w_{i+1})\cdots \mathcal{Q}_{-}(w_s){\rm d}\bm{y}{\rm d}\bm{w}.
 \end{aligned}
 \label{vac-001}
 \end{align}
  As in the discussion preceding (\ref{vac-001}), by applying the generalized Stokes' theorem to the integration with respect to $\bm{y}$, we see that the third form on the right-hand side of (\ref{vac-001}) cancels with the second one.
 Therefore, the commutator (\ref{vac-fin-tderi}) can be written as as a total derivative form.
 \end{proof}

\subsection{The structure of Fock modules}
In this subsection, we review the $\mathfrak{Vir}$-structure of Fock modules in the case where the parameter $\rho$ corresponds to the Virasoro minimal series.
The structure of these Fock modules was classified in \cite{FF}, and they are also referred to as Feigin-Fuchs modules.
Let $p_+$, $p_-$ be coprime integers satisfying $p_->p_+\geq 1$. 
Following \cite{TW}, we introduce the notation
\begin{align*}
\alpha_+=\sqrt{\frac{2p_-}{p_+}},\qquad\qquad
\alpha_-=-\sqrt{\frac{2p_+}{p_-}},\qquad\qquad
\alpha_0=\alpha_++\alpha_-.
\end{align*}
Note that $\alpha_+,\alpha_-$ are the solutions of (\ref{second-order}) as $\rho=\alpha_0$, and that $c_{\alpha_0}$ equals the minimal central charge 
\begin{align*}
c_{p_+,p_-}:=1-6\frac{(p_+-p_-)^2}{p_+p_-}.
\end{align*} 
Set
\begin{align}
\label{hrs-notpq}
\begin{aligned}
h_{r,s}:&=h^{(\alpha_0)}_{r,s}=\frac{r^2-1}{4}\frac{p_-}{p_+}-\frac{rs-1}{2}+\frac{s^2-1}{4}\frac{p_+}{p_-},\\
h_{r,s;n}:&=h_{r-np_+,s},
\end{aligned}
\end{align}
for $r,s,n\in\mathbb{Z}$. 
Note that 
$
h_{r-np_+,s}=h_{r,s+np_-}.
$
Also, we introduce the notation
\begin{align*}
\begin{aligned}
&\alpha_{r,s}:=\alpha^{(\alpha_0)}_{r,s},
&\alpha_{r,s;n}:=\alpha_{r,s}+\frac{\sqrt{2p_+p_-}}{2}n, \qquad r,s,n\in\mathbb{Z}
\end{aligned}
\end{align*}
(for the notation $\alpha^{(\rho)}_{r,s}$, see (\ref{alpha-rs-rho})), and use simple notation
\begin{align*}
&F_{r,s;n}=F^{(\alpha_0)}_{\alpha_{r,s;n}},
&F_{r,s}&=F^{(\alpha_0)}_{\alpha_{r,s}}
\end{align*}
In the notation of (\ref{F-rho}), $F_{r,s}$ is written as $F^{(\alpha_0)}_{{r,s}}$.
For each $r,s,n\in\mathbb{Z}$, let $L(h_{r,s;n})$ be the simple Virasoro module whose lowest weight is $h_{r,s;n}$ and the central charge 
$c_{p_+,p_-}$
(for the notation $h_{r,s;n}$, see (\ref{hrs-notpq})). 

Before describing the Virasoro structure of Fock modules $F_{r,s;n}$, let us introduce the notion of socle series. 
\begin{dfn}
Let $V$ be a vertex operator algebra or the Virasoro algebra, and let $M$ be a finite length $V$-module. 
Let ${\rm Soc}(M)$ denote the maximal semisimple submodule of $M$, called the {\rm socle}. 
Furthermore, we call the following sequence the {\rm socle} {\rm series} of $M$:
\begin{align*}
{\rm Soc}_1(M)\subsetneq {\rm Soc}_2(M)\subsetneq \cdots \subsetneq {\rm Soc}_n(M)=M
\end{align*}
with ${\rm Soc}_1(M)={\rm Soc}(M)$ and ${\rm Soc}_{i+1}(M)/{\rm Soc}_{i}(M)={\rm Soc}(M/{\rm Soc}_{i}(M))$.
\end{dfn}

The following proposition is due to \cite{FF}.

\begin{prop}[\cite{FF}]
\label{FockSocle}
As the Virasoro modules, there are four cases of socle series for the Fock modules $F_{r,s;n}\in \mathcal{F}_{\alpha_0}\mathchar`-{\rm mod}$:
\begin{enumerate}
\item For each $1\leq r\leq p_+-1,\ 1\leq s\leq p_--1,\ n\in\mathbb{Z}$, we have
\begin{align*}
{\rm Soc}_1(F_{r,s;n})\subsetneq {\rm Soc}_2(F_{r,s;n})\subsetneq {\rm Soc}_3(F_{r,s;n})=F_{r,s;n}
\end{align*} 
such that
\begin{align*}
&{\rm Soc}_1(F_{r,s;n})={\rm Soc}(F_{r,s;n})=\bigoplus_{k\geq 0}L(h_{r,p_--s;|n|+2k+1}),\\
&{\rm Soc}_2(F_{r,s;n})/{\rm Soc}_1(F_{r,s;n})={\rm Soc}(F_{r,s;n}/{\rm Soc}_1(F_{r,s;n}))\\
&\hspace{125pt}=\bigoplus_{k\geq a}L(h_{r,s;|n|+2k})\oplus \bigoplus_{k\geq 1-a}L(h_{p_+-r,p_--s;|n|+2k}),\\
&{\rm Soc}_3(F_{r,s;n})/{\rm Soc}_2(F_{r,s;n})={\rm Soc}(F_{r,s;n}/{\rm Soc}_2(F_{r,s;n}))=\bigoplus_{k\geq 0}L(h_{p_+-r,s;|n|+2k+1}),
\end{align*}
where $a=0$ if $n\geq 0$ and $a=1$ if $n<0$.

\item For each $1\leq s\leq p_--1,\ n\in\mathbb{Z}$, we have
\begin{align*}
{\rm Soc}_1(F_{p_+,s;n})\subsetneq {\rm Soc}_2(F_{p_+,s;n})=F_{p_+,s;n}
\end{align*}
such that
\begin{align*}
&{\rm Soc}_1(F_{p_+,s;n})={\rm Soc}(F_{p_+,s;n})=\bigoplus_{k\geq 0}L(h_{p_+,p_--s;|n|+2k+1}),\\
&{\rm Soc}_2(F_{p_+,s;n})/{\rm Soc}_1(F_{p_+,s;n})=\bigoplus_{k\geq a}L(h_{p_+,s;|n|+2k})
\end{align*}
where $a=0$ if $n\geq 1$ and $a=1$ if $n<1$.

\item For each $1\leq r \leq p_+-1,\ n\in\mathbb{Z}$, we have
\begin{align*}
{\rm Soc}_1(F_{r,p_-;n})\subsetneq {\rm Soc}_2(F_{r,p_-;n})=F_{r,p_-;n}
\end{align*}
such that
\begin{align*}
&{\rm Soc}_1(F_{r,p_-;n})={\rm Soc}(F_{r,p_-;n})=\bigoplus_{k\geq 0}L(h_{r,p_-;|n|+2k}),\\
&{\rm Soc}_2(F_{r,p_-;n})/{\rm Soc}_1(F_{r,p_-;n})=\bigoplus_{k\geq a}L(h_{p_+-r,p_-;|n|+2k-1}),
\end{align*}
where $a=1$ if $n\geq 0$ and $a=0$ if $n<0$.

\item For each $n \in\mathbb{Z}$, the Fock module $F_{p_+,p_-;n}$  is semi-simple as a Virasoro module:
\begin{align*}
F_{p_+,p_-;n}={\rm Soc}(F_{p_+,p_-;n})=\bigoplus_{k\geq 0}L(h_{p_+,p_-;|n|+2k}).
\end{align*}
\end{enumerate}
\end{prop}
In the above proposition, the first family of the Fock modules is referred to as braided type, the second and third as chain type, and the fourth as semisimple type.
Figures~\ref{51fig} and \ref{51fig-2} illustrate the structures of the Fock modules of chain type and braided type.

\begin{figure}[h]
  \centering
\begin{tikzpicture}[scale=1.05]
\node[inner sep=0.7pt] (a1) at (-0.5, -0.75) {$\bullet$};
\node[inner sep=0.7pt] (a2) at (-0.5, -0.75-1) {$\circ$};
\node[inner sep=0.7pt] (a3) at (-0.5, -0.75-2) {$\bullet$};
\node[inner sep=0.7pt] (a4) at (-0.5, -0.75-3) {$\circ$};
\node[inner sep=0.7pt] (a5) at (-0.5, -0.75-3-1) {$\bullet$};
\node[inner sep=0.7pt] (c) at (-0.5, -0.75-3.1-1) {};
\node[inner sep=0.7pt] (c1) at (-0.5, -0.75-3-0.5-1) {};

\draw[arrows = {-Stealth[scale=0.9]}] (a2) to (a1);
\draw[arrows = {-Stealth[scale=0.9]}] (a2) to (a3);
\draw[arrows = {-Stealth[scale=0.9]}] (a4) to (a3);
\draw[arrows = {-Stealth[scale=0.9]}] (a4) to (a5);
\draw[ dotted,thick] (c) to (c1);

\node[inner sep=0.7pt,font=\scriptsize] (t1) at (-0.5+0.1, -0.75+0.5) {${F_{p_+,s;n}}$};
\node[inner sep=0.7pt,font=\scriptsize] (t2) at (-0.95+0.1-0.6-4, -0.75+0.5) {${F_{p_+,s;m}}$};


\node[inner sep=0.7pt,shift={(-5,0)}] (ak1) at (-0.5, -0.75) {$\circ$};
\node[inner sep=0.7pt,shift={(-5,0)}] (ak2) at (-0.5, -0.75-1) {$\bullet$};
\node[inner sep=0.7pt,shift={(-5,0)}] (ak3) at (-0.5, -0.75-2) {$\circ$};
\node[inner sep=0.7pt,shift={(-5,0)}] (ak4) at (-0.5, -0.75-3) {$\bullet$};
\node[inner sep=0.7pt,shift={(-5,0)}] (ak5) at (-0.5, -0.75-3-1) {$\circ$};
\node[inner sep=0.7pt,shift={(-5,0)}] (ck) at (-0.5, -0.75-3.1-1) {};
\node[inner sep=0.7pt,shift={(-5,0)}] (ck1) at (-0.5, -0.75-3-0.5-1) {};

\draw[arrows = {-Stealth[scale=0.9]}] (ak1) to (ak2);
\draw[arrows = {-Stealth[scale=0.9]}] (ak3) to (ak2);
\draw[arrows = {-Stealth[scale=0.9]}] (ak3) to (ak4);
\draw[arrows = {-Stealth[scale=0.9]}] (ak5) to (ak4);
\draw[ dotted,thick] (ck) to (ck1);

\end{tikzpicture}
\caption{\label{51fig}Schematic diagrams of the Fock modules $F_{p_+,s;m}$, $F_{p_+,s;n}$ for $1\leq s<p_-$, $m\geq 1$, and $n\leq 0$. The black circles represent the generating vectors for the socle of the Fock modules, and the white circles correspond the generating vectors of ${\rm Soc}_2/{\rm Soc}_1$. The top circles represent the lowest weight vectors, and the circles are arranged so that their $\mathcal{L}_0$-weights increase from top to bottom.
Each arrow indicates that the target vector is obtained, up to an appropriate quotient, via the action of $\mathfrak{Vir}$ (for further details, see\cite[Section 8.3]{IK}).}
\mbox{}\\
\end{figure}
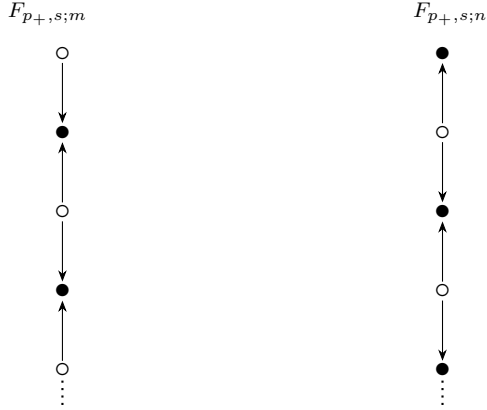

\begin{figure}[h]
  \centering
\begin{tikzpicture}[scale=1.05]
\node[inner sep=0.7pt,font=\scriptsize] (t) at (0, 0.25) {$\square$};
\node[inner sep=0.7pt] (a1) at (-0.5, -0.75) {$\bullet$};
\node[inner sep=0.7pt,font=\scriptsize] (b1) at (0.5, -0.75) {$\triangle$};
\node[inner sep=0.7pt,font=\scriptsize] (a2) at (-0.5, -0.75-1) {$\square$};
\node[inner sep=0.7pt,font=\scriptsize] (b2) at (0.5, -0.75-1) {$\square$};
\node[inner sep=0.7pt] (a3) at (-0.5, -0.75-2) {$\bullet$};
\node[inner sep=0.7pt,font=\scriptsize] (b3) at (0.5, -0.75-2) {$\triangle$};
\node[inner sep=0.7pt,font=\scriptsize] (a4) at (-0.5, -0.75-3) {$\square$};
\node[inner sep=0.7pt,font=\scriptsize] (b4) at (0.5, -0.75-3) {$\square$};
\node[inner sep=0.7pt] (c) at (-0.5, -0.75-3.2) {};
\node[inner sep=0.7pt] (d) at (0.5, -0.75-3.2) {};
\node[inner sep=0.7pt] (c1) at (-0.5, -0.75-3-0.5) {};
\node[inner sep=0.7pt] (d1) at (0.5, -0.75-3-0.5) {};

\draw[arrows = {-Stealth[scale=0.9]}] (t) to (a1);
\draw[arrows = {-Stealth[scale=0.9]}] (b1) to (t);
\draw[arrows = {-Stealth[scale=0.9]}] (b1) to (a2);
\draw[arrows = {-Stealth[scale=0.9]}] (b1) to (b2);
\draw[arrows = {-Stealth[scale=0.9]}] (a2) to (a1);
\draw[arrows = {-Stealth[scale=0.9]}] (b2) to (a1);

\draw[arrows = {-Stealth[scale=0.9]}] (b2) to (a3);
\draw[arrows = {-Stealth[scale=0.9]}] (a2) to (a3);
\draw[arrows = {-Stealth[scale=0.9]}] (b3) to (b2);
\draw[arrows = {-Stealth[scale=0.9]}] (b3) to (a2);

\draw[arrows = {-Stealth[scale=0.9]}] (b4) to (a3);
\draw[arrows = {-Stealth[scale=0.9]}] (a4) to (a3);
\draw[arrows = {-Stealth[scale=0.9]}] (b3) to (a4);
\draw[arrows = {-Stealth[scale=0.9]}] (b3) to (b4);
\draw[arrows = {-Stealth[scale=0.9]}] (b3) to (a4);
\draw[arrows = {-Stealth[scale=0.9]}] (b3) to (b4);
\draw[ dotted,thick] (c) to (c1);
\draw[ dotted,thick] (d) to (d1);

\end{tikzpicture}
\caption{\label{51fig-2}Schematic diagrams of the Fock module $F_{r,s;n}$ $(1\leq r<p_+,\ 1\leq s<p_-,\ n\in \mathbb{Z})$. The black circles represent the generating vectors for the socle of the Fock modules. The white squares and triangles correspond the generating vectors of ${\rm Soc}_2/{\rm Soc}_1$ and ${\rm Soc}_3/{\rm Soc}_2$, respectively. 
The top square represents the lowest weight vector, and the other shapes are arranged so that their $\mathcal{L}_0$-weights increase from top to bottom.
Each arrow indicates that the target vector is obtained, up to an appropriate quotient, via the action of $\mathfrak{Vir}$ (for further details, see\cite[Section 8.3]{IK}).}
\mbox{}\\
\end{figure}

\subsection{Screening currents and Felder complex}
Hereafter we use the following notation
\begin{align}
\label{def-pq-sc}
\begin{aligned}
Q_\pm(z)=\mathcal{Q}^{(\alpha_0)}_\pm(z),\qquad Q^{[r]}_+(z)=\mathcal{Q}^{(\alpha_0);[r]}_+(z),\qquad 
Q^{[s]}_-(z)=\mathcal{Q}^{(\alpha_0);[s]}_-(z).
\end{aligned}
\end{align}
For $1\leq r\leq p_+-1$, $1\leq s\leq p_--1$, we use notation
\begin{align*}
&r^{\vee}(p_+):=p_+-r,
&s^{\vee}(p_-):=p_--s
\end{align*}
and use the abbreviations 
$
r^\vee=r^{\vee}(p_+)
$,
$
s^\vee=s^{\vee}(p_-).
$

For $1\leq r \leq p_+,1\leq s\leq p_-$, and $n\in\mathbb{Z}$, we define the following Virasoro modules in accordance with \cite{TW}:
\begin{enumerate}
\item For $1\leq r<p_+,\ 1\leq s\leq p_-,\ n\in\mathbb{Z}$ 
\begin{align*}
K_{r,s;n;+}&={\rm ker}(Q^{[r]}_+|_{F_{r,s;n}}):F_{r,s;n}\rightarrow F_{r^\vee,s;n+1},\\
X_{r^\vee,s;n+1;+}&={\rm im}(Q^{[r]}_+|_{F_{r,s;n}}):F_{r,s;n}\rightarrow F_{r^\vee,s;n+1}.
\end{align*}
\item For $1\leq r\leq p_+,\ 1\leq s<p_-,\ n\in\mathbb{Z}$ 
\begin{align*}
K_{r,s;n;-}&={\rm ker}(Q^{[s]}_-|_{F_{r,s;n}}):F_{r,s;n}\rightarrow F_{r,s^\vee;n-1},\\
X_{r,s^\vee;n-1;-}&={\rm im}(Q^{[s]}_-|_{F_{r,s;n}}):F_{r,s;n}\rightarrow F_{r,s^\vee;n-1}.
\end{align*}
\end{enumerate}

The following propositions are due to \cite{Felder}.
\begin{prop}[\cite{Felder}]
\label{Felder complex}
The socle series of $K_{r,s;n;\pm}$ and $X_{r,s;n;\pm}$ are given by :
\begin{enumerate}
\item For $1\leq r\leq p_+-1,\ 1\leq s\leq p_--1$, we have
\begin{align*}
&S_1(K_{r,s;n;\pm})={\rm Soc}(K_{r,s;n;\pm})\subsetneq K_{r,s;n;\pm},\\
&S_1(X_{r,s;n;\pm})={\rm Soc}(X_{r,s;n;\pm})\subsetneq X_{r,s;n;\pm}
\end{align*} 
such that
\begin{align*}
n&\geq 0&n&\leq -1\\
S_1(K_{r,s;n;+})&=\bigoplus_{k\geq 1}L(h_{r,s^\vee;n+2k-1}),&S_1(K_{r,s;n;+})&=\bigoplus_{k\geq 1}L(h_{r,s^\vee;-n+2k-1}),\\
K_{r,s;n;+}/S_1&=\bigoplus_{k\geq 1}L(h_{r,s;n+2(k-1)}),&K_{r,s;n;+}/S_1&=\bigoplus_{k\geq 1}L(h_{r,s;-n+2k}),\\
S_1(X_{r,s;n+1;+})&=\bigoplus_{k\geq 1}L(h_{r,s^\vee;n+2k}),&S_1(X_{r,s;n+1;+})&=\bigoplus_{k\geq 1}L(h_{r,s^\vee;-n+2(k-1)}),\\
X_{r,s;n+1;+}/S_1&=\bigoplus_{k\geq 1}L(h_{r,s;n+2k-1}),&X_{r,s;n+1;+}/S_1&=\bigoplus_{k\geq 1}L(h_{r,s;-n+2k-1}).
\end{align*}
\begin{align*}
n&\geq 1&n&\leq 0\\
S_1(K_{r,s;n;-})&=\bigoplus_{k\geq 1}L(h_{r,s^\vee;n+2k-1}),&S_1(K_{r,s;n;-})&=\bigoplus_{k\geq 1}L(h_{r,s^\vee;-n+2k-1}),\\
K_{r,s;n;-}/S_1&=\bigoplus_{k\geq 1}L(h_{r,s;n+2(k-1)}),&K_{r,s;n;-}/S_1&=\bigoplus_{k\geq 1}L(h_{r,s;-n+2k}),\\
S_1(X_{r,s;n+1;-})&=\bigoplus_{k\geq 1}L(h_{r,s^\vee;n+2(k-1)}),&S_1(X_{r,s;n+1;-})&=\bigoplus_{k\geq 1}L(h_{r,s^\vee;-n+2k}),\\
X_{r,s;n+1;-}/S_1&=\bigoplus_{k\geq 1}L(h_{r^\vee,s^\vee;n+2k-1}),&X_{r,s;n+1;-}/S_1&=\bigoplus_{k\geq 1}L(h_{r^\vee,s^\vee;-n+2k-1}).
\end{align*}
\item For $1\leq r\leq p_+-1,\ s=p_-,\ n\in\mathbb{Z}$, we have
\begin{align*}
X_{r,p_-;n}={\rm Soc}(F_{r,p_-;n}).
\end{align*}
\item For $r=p_+,\ 1\leq s\leq p_--1,\ n\in\mathbb{Z}$, we have
\begin{align*}
X_{p_+,s;n}={\rm Soc}(F_{p_+,s;n}).\\
\end{align*}
\end{enumerate}
\end{prop}

\begin{prop}[\cite{Felder}]
\label{Felder complex2}
\mbox{}
\begin{enumerate}
\item 
Let $p_+\geq 2$.
For $1\leq r<p_+,\ 1\leq s<p_-$ and $n\in\mathbb{Z}$ the screening operators $Q^{[r]}_+$ and $Q^{[r^{\vee}]}_+$ define the complex
\begin{align*}
\cdots\xrightarrow{Q^{[r]}_+}F_{r^{\vee},s;n-1}\xrightarrow{Q^{[r^{\vee}]}_+}F_{r,s;n}\xrightarrow{Q^{[r]}_+}F_{r^{\vee},s;n+1}\xrightarrow{Q^{[r^{\vee}]}_+}\cdots.
\end{align*}
This complex is exact everywhere except in $F_{r,s}=F_{r,s;0}$ where the cohomology is given by
\begin{align*}
{\rm ker}Q^{[r]}_+/{\rm im}Q^{[r^\vee]}_+\simeq L(h_{r,s;0}).
\end{align*}
\item
Let $p_+\geq 2$.
For $1\leq r<p_+,\ 1\leq s<p_-$ and $n\in\mathbb{Z}$ the screening operators $Q^{[s]}_-$ and $Q^{[s^{\vee}]}_-$ define the complex
\begin{align}
\label{fel-comp:p>1}
\cdots\xrightarrow{Q^{[s]}_-}F_{r,s^{\vee};n+1}\xrightarrow{Q^{[s^{\vee}]}_-}F_{r,s;n}\xrightarrow{Q^{[s]}_-}F_{r,s^{\vee};n-1}\xrightarrow{Q^{[s^{\vee}]}_-}\cdots.
\end{align}
This complex is exact everywhere except in $F_{r,s}=F_{r,s;0}$ where the cohomology is given by
\begin{align*}
{\rm ker}Q^{[s]}_-/{\rm im}Q^{[s^\vee]}_-\simeq L(h_{r,s;0}).
\end{align*}
\item 
Let $p_+\geq 2$.
For $1\leq r<p_+$ and $n\in\mathbb{Z}$ the screening operators $Q^{[r]}_+$ and $Q^{[r^{\vee}]}_+$ define the complex
\begin{align*}
\cdots\xrightarrow{Q^{[r]}_+}F_{r^{\vee},p_-;n-1}\xrightarrow{Q^{[r^{\vee}]}_+}F_{r,p_-;n}\xrightarrow{Q^{[r]}_+}F_{r^{\vee},p_-;n+1}\xrightarrow{Q^{[r^{\vee}]}_+}\cdots
\end{align*}
and this complex is exact.
\item
For $1\leq s<p_-$ and $n\in\mathbb{Z}$ the screening operators $Q^{[s]}_-$ and $Q^{[s^{\vee}]}_-$ define the complex
\begin{align}
\label{eq:fel-comp}
\cdots\xrightarrow{Q^{[s]}_-}F_{p_+,s^{\vee};n+1}\xrightarrow{Q^{[s^{\vee}]}_-}F_{p_+,s;n}\xrightarrow{Q^{[s]}_-}F_{p_+,s^{\vee};n-1}\xrightarrow{Q^{[s^{\vee}]}_-}\cdots
\end{align}
and this complex is exact.

\end{enumerate}
\end{prop}

\begin{Remark}
\label{rem-subsing}
Figures~\ref{51fig-fel} and \ref{51fig-2-fel} illustrate the structures of the complexes (\ref{eq:fel-comp}) and (\ref{fel-comp:p>1}).
The vectors corresponding to the black and white squares in Figure~\ref{51fig-2-fel} are called {\rm subsingular} {\rm vectors}.
Note that, by Proposition~\ref{Felder complex2}, the subsingular vectors corresponding to the white squares are contained in the image of $Q^{[r^\vee]}_+$.
\end{Remark}

\begin{figure}[h]
  \centering
\begin{tikzpicture}[scale=1.05]
\node[inner sep=0.7pt] (a1) at (-0.5, -0.75) {$\bullet$};
\node[inner sep=0.7pt] (a2) at (-0.5, -0.75-1) {$\circ$};
\node[inner sep=0.7pt] (a3) at (-0.5, -0.75-2) {$\bullet$};
\node[inner sep=0.7pt] (a4) at (-0.5, -0.75-3) {$\circ$};
\node[inner sep=0.7pt] (a5) at (-0.5, -0.75-3-1) {$\bullet$};
\node[inner sep=0.7pt] (c) at (-0.5, -0.75-3.1-1) {};
\node[inner sep=0.7pt] (c1) at (-0.5, -0.75-3-0.5-1) {};

\draw[arrows = {-Stealth[scale=0.9]}] (a2) to (a1);
\draw[arrows = {-Stealth[scale=0.9]}] (a2) to (a3);
\draw[arrows = {-Stealth[scale=0.9]}] (a4) to (a3);
\draw[arrows = {-Stealth[scale=0.9]}] (a4) to (a5);
\draw[ dotted,thick] (c) to (c1);

\node[inner sep=0.7pt,shift={(1,-1.055)}] (aa1) at (-0.5, -0.75) {$\bullet$};
\node[inner sep=0.7pt,shift={(1,-1.055)}] (aa2) at (-0.5, -0.75-1) {$\circ$};
\node[inner sep=0.7pt,shift={(1,-1.055)}] (aa3) at (-0.5, -0.75-2) {$\bullet$};
\node[inner sep=0.7pt,shift={(1,-1.055)}] (aa4) at (-0.5, -0.75-3) {$\circ$};
\node[inner sep=0.7pt,shift={(1,-1.055)}] (cc) at (-0.5, -0.75-3.1) {};
\node[inner sep=0.7pt,shift={(1,-1.055)}] (cc1) at (-0.5, -0.75-3-0.5) {};

\draw[arrows = {-Stealth[scale=0.9]}] (aa2) to (aa1);
\draw[arrows = {-Stealth[scale=0.9]}] (aa2) to (aa3);
\draw[arrows = {-Stealth[scale=0.9]}] (aa4) to (aa3);
\draw[ dotted,thick] (cc) to (cc1);

\draw[->,>=stealth] (-1.25,-0.75)--(-0.75,-0.75);
\draw[->,>=stealth,dashed] (-1.25,-0.75-1)--(-0.75,-0.75-1);
\draw[->,>=stealth] (-1.25,-0.75-2)--(-0.75,-0.75-2);
\draw[->,>=stealth,dashed] (-1.25,-0.75-3)--(-0.75,-0.75-3);
\draw[->,>=stealth] (-1.25,-0.75-4)--(-0.75,-0.75-4);

\draw[->,>=stealth,shift={(-0.95,-1)}] (-1.25,-0.75)--(-0.75,-0.75);
\draw[->,>=stealth,dashed,shift={(-0.95,-1)}] (-1.25,-0.75-1)--(-0.75,-0.75-1);
\draw[->,>=stealth,shift={(-0.95,-1)}] (-1.25,-0.75-2)--(-0.75,-0.75-2);
\draw[->,>=stealth,dashed,shift={(-0.95,-1)}] (-1.25,-0.75-3)--(-0.75,-0.75-3);

\draw[->,>=stealth,shift={(-1.875,-2)}] (-1.25,-0.75)--(-0.75,-0.75);
\draw[->,>=stealth,dashed,shift={(-1.875,-2)}] (-1.25,-0.75-1)--(-0.75,-0.75-1);
\draw[->,>=stealth,shift={(-1.875,-2)}] (-1.25,-0.75-2)--(-0.75,-0.75-2);

\draw[->,>=stealth,shift={(-1.875-0.95,-2-1)}] (-1.25,-0.75)--(-0.75,-0.75);
\draw[->,>=stealth,dashed,shift={(-1.875-0.95,-2-1)}] (-1.25,-0.75-1)--(-0.75,-0.75-1);

\draw[->,>=stealth,shift={(-1.875-0.95-0.95,-2-1-1)}] (-1.25,-0.75)--(-0.75,-0.75);

\draw[->,>=stealth,shift={(-0.95+2-0.05,-1)}] (-1.25,-0.75)--(-0.75,-0.75);
\draw[->,>=stealth,dashed,shift={(-0.95+2-0.05,-1)}] (-1.25,-0.75-1)--(-0.75,-0.75-1);
\draw[->,>=stealth,shift={(-0.95+2-0.05,-1)}] (-1.25,-0.75-2)--(-0.75,-0.75-2);
\draw[->,>=stealth,dashed,shift={(-0.95+2-0.05,-1)}] (-1.25,-0.75-3)--(-0.75,-0.75-3);

\draw[->,>=stealth,shift={(-1.875+4-0.2,-2)}] (-1.25,-0.75)--(-0.75,-0.75);
\draw[->,>=stealth,dashed,shift={(-1.875+4-0.2,-2)}] (-1.25,-0.75-1)--(-0.75,-0.75-1);
\draw[->,>=stealth,shift={(-1.875+4-0.2,-2)}] (-1.25,-0.75-2)--(-0.75,-0.75-2);

\draw[->,>=stealth,shift={(-1.875-0.95+6-0.35,-2-1)}] (-1.25,-0.75)--(-0.75,-0.75);
\draw[->,>=stealth,dashed,shift={(-1.875-0.95+6-0.35,-2-1)}] (-1.25,-0.75-1)--(-0.75,-0.75-1);

\draw[->,>=stealth,shift={(-1.875-0.95-0.95+8-0.35,-2-1-1)}] (-1.25,-0.75)--(-0.75,-0.75);

\node[inner sep=0.7pt,font=\scriptsize] (t1) at (-0.5+0.1, -0.75+0.5) {${F_{p_+,s}}$};
\node[inner sep=0.7pt,font=\scriptsize] (t2) at (-0.95+0.1-0.6, -0.75+0.5) {${F_{p_+,s^\vee;1}}$};
\node[inner sep=0.7pt,font=\scriptsize] (t1) at (-0.5+0.1-2, -0.75+0.5-1) {${F_{p_+,s;2}}$};
\node[inner sep=0.7pt,font=\scriptsize] (t2) at (-0.95+0.1-0.6+2, -0.75+0.5-1) {${F_{p_+,s^\vee;-1}}$};
\node[inner sep=0.7pt,font=\scriptsize] (t1) at (-0.5+0.1-2-2, -0.75+0.5-1-2) {${F_{p_+,s;4}}$};
\node[inner sep=0.7pt,font=\scriptsize] (t2) at (-0.95+0.1-0.6+2+2, -0.75+0.5-1-2) {${F_{p_+,s^\vee;-3}}$};
\node[inner sep=0.7pt,font=\scriptsize] (t1) at (-0.5+0.1-2-1, -0.75+0.5-1-1) {${F_{p_+,s^\vee;3}}$};
\node[inner sep=0.7pt,font=\scriptsize] (t2) at (-0.95+0.1-0.6+2+1, -0.75+0.5-1-1) {${F_{p_+,s;-2}}$};
\node[inner sep=0.7pt,font=\scriptsize] (t1) at (-0.5+0.1-2-1-2, -0.75+0.5-1-1-2) {${F_{p_+,s^\vee;5}}$};
\node[inner sep=0.7pt,font=\scriptsize] (t2) at (-0.95+0.1-0.6+2+1+2, -0.75+0.5-1-1-2) {${F_{p_+,s;-4}}$};

\node[inner sep=0.7pt,shift={(1+1,-1.085-1)}] (aar1) at (-0.5, -0.75) {$\bullet$};
\node[inner sep=0.7pt,shift={(1+1,-1.085-1)}] (aar2) at (-0.5, -0.75-1) {$\circ$};
\node[inner sep=0.7pt,shift={(1+1,-1.085-1)}] (aar3) at (-0.5, -0.75-2) {$\bullet$};
\node[inner sep=0.7pt,shift={(1+1,-1.055)}] (ccr) at (-0.5, -0.75-3.1) {};
\node[inner sep=0.7pt,shift={(1+1,-1.055)}] (ccr1) at (-0.5, -0.75-3-0.5) {};

\draw[arrows = {-Stealth[scale=0.9]}] (aar2) to (aar1);
\draw[arrows = {-Stealth[scale=0.9]}] (aar2) to (aar3);
\draw[ dotted,thick] (ccr) to (ccr1);

\node[inner sep=0.7pt,shift={(1+1+1,-1.085-1-1.05)}] (aarr1) at (-0.5, -0.75) {$\bullet$};
\node[inner sep=0.7pt,shift={(1+1+1,-1.085-1-1.05)}] (aarr2) at (-0.5, -0.75-1) {$\circ$};
\node[inner sep=0.7pt,shift={(1+1+1,-1.055)}] (ccrr) at (-0.5, -0.75-3.1) {};
\node[inner sep=0.7pt,shift={(1+1+1,-1.055)}] (ccrr1) at (-0.5, -0.75-3-0.5) {};

\draw[arrows = {-Stealth[scale=0.9]}] (aarr2) to (aarr1);
\draw[ dotted,thick] (ccrr) to (ccrr1);

\node[inner sep=0.7pt,shift={(1+1+1+1,-1.085-1-1-1.05)}] (aarrr1) at (-0.5, -0.8) {$\bullet$};
\node[inner sep=0.7pt,shift={(1+1+1+1,-1.055)}] (ccrrr) at (-0.5, -0.75-3.1) {};
\node[inner sep=0.7pt,shift={(1+1+1+1,-1.055)}] (ccrrr1) at (-0.5, -0.75-3-0.5) {};
\draw[ dotted,thick] (ccrrr) to (ccrrr1);


\node[inner sep=0.7pt,shift={(-1,0)}] (ak1) at (-0.5, -0.75) {$\circ$};
\node[inner sep=0.7pt,shift={(-1,0)}] (ak2) at (-0.5, -0.75-1) {$\bullet$};
\node[inner sep=0.7pt,shift={(-1,0)}] (ak3) at (-0.5, -0.75-2) {$\circ$};
\node[inner sep=0.7pt,shift={(-1,0)}] (ak4) at (-0.5, -0.75-3) {$\bullet$};
\node[inner sep=0.7pt,shift={(-1,0)}] (ak5) at (-0.5, -0.75-3-1) {$\circ$};
\node[inner sep=0.7pt,shift={(-1,0)}] (ck) at (-0.5, -0.75-3.1-1) {};
\node[inner sep=0.7pt,shift={(-1,0)}] (ck1) at (-0.5, -0.75-3-0.5-1) {};

\draw[arrows = {-Stealth[scale=0.9]}] (ak1) to (ak2);
\draw[arrows = {-Stealth[scale=0.9]}] (ak3) to (ak2);
\draw[arrows = {-Stealth[scale=0.9]}] (ak3) to (ak4);
\draw[arrows = {-Stealth[scale=0.9]}] (ak5) to (ak4);
\draw[ dotted,thick] (ck) to (ck1);

\node[inner sep=0.7pt,shift={(-2,-1.055)}] (aak1) at (-0.5, -0.75) {$\circ$};
\node[inner sep=0.7pt,shift={(-2,-1.055)}] (aak2) at (-0.5, -0.75-1) {$\bullet$};
\node[inner sep=0.7pt,shift={(-2,-1.055)}] (aak3) at (-0.5, -0.75-2) {$\circ$};
\node[inner sep=0.7pt,shift={(-2,-1.055)}] (aak4) at (-0.5, -0.75-3) {$\bullet$};
\node[inner sep=0.7pt,shift={(-2,-1.055)}] (cck) at (-0.5, -0.75-3.1) {};
\node[inner sep=0.7pt,shift={(-2,-1.055)}] (cck1) at (-0.5, -0.75-3-0.5) {};

\draw[arrows = {-Stealth[scale=0.9]}] (aak1) to (aak2);
\draw[arrows = {-Stealth[scale=0.9]}] (aak3) to (aak2);
\draw[arrows = {-Stealth[scale=0.9]}] (aak3) to (aak4);
\draw[ dotted,thick] (cck) to (cck1);

\node[inner sep=0.7pt,shift={(-2-1,-1.085-1)}] (aakk1) at (-0.5, -0.75) {$\circ$};
\node[inner sep=0.7pt,shift={(-2-1,-1.085-1)}] (aakk2) at (-0.5, -0.75-1) {$\bullet$};
\node[inner sep=0.7pt,shift={(-2-1,-1.085-1)}] (aakk3) at (-0.5, -0.75-2) {$\circ$};
\node[inner sep=0.7pt,shift={(-2-1,-1.085)}] (cckk) at (-0.5, -0.75-3.1) {};
\node[inner sep=0.7pt,shift={(-2-1,-1.085)}] (cckk1) at (-0.5, -0.75-3-0.5) {};

\draw[arrows = {-Stealth[scale=0.9]}] (aakk1) to (aakk2);
\draw[arrows = {-Stealth[scale=0.9]}] (aakk3) to (aakk2);
\draw[ dotted,thick] (cckk) to (cckk1);

\node[inner sep=0.7pt,shift={(-2-2,-1.095-2.05)}] (aakkk1) at (-0.5, -0.75) {$\circ$};
\node[inner sep=0.7pt,shift={(-2-2,-1.095-2.05)}] (aakkk2) at (-0.5, -0.75-1) {$\bullet$};
\node[inner sep=0.7pt,shift={(-2-2,-1.095)}] (cckkk) at (-0.5, -0.75-3.1) {};
\node[inner sep=0.7pt,shift={(-2-2,-1.095)}] (cckkk1) at (-0.5, -0.75-3-0.5) {};

\draw[arrows = {-Stealth[scale=0.9]}] (aakkk1) to (aakkk2);
\draw[ dotted,thick] (cckkk) to (cckkk1);

\node[inner sep=0.7pt,shift={(-2-3,-1.135-3.05)}] (aakkkk1) at (-0.5, -0.75) {$\circ$};
\node[inner sep=0.7pt,shift={(-2-3,-1.135)}] (cckkkk) at (-0.5, -0.75-3.1) {};
\node[inner sep=0.7pt,shift={(-2-3,-1.135)}] (cckkkk1) at (-0.5, -0.75-3-0.5) {};
\draw[ dotted,thick] (cckkkk) to (cckkkk1);

\end{tikzpicture}
\caption{\label{51fig-fel}The Felder complex (\ref{eq:fel-comp}). Horizontal arrows represent the action of the screening operators, whereas dashed arrows indicate that the socles are annihilated by the screening operators. The black circles represent the generating vectors of ${\rm ker}Q^{[\bullet]}_-$. The white circles represent the generating vectors of ${\rm coker}Q^{[\bullet]}_-$. }
\mbox{}\\
\end{figure}
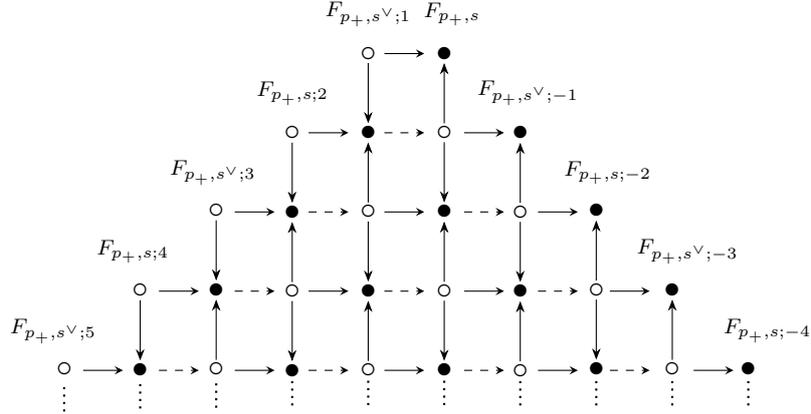

\begin{figure}[h]
  \centering
\begin{tikzpicture}[scale=1.05]
\node[inner sep=0.7pt,font=\scriptsize] (t1) at (0, 0.9) {${F_{r,s}}$};
\node[inner sep=0.7pt,font=\scriptsize] (t) at (0, 0.25) {$\blacksquare$};
\node[inner sep=0.7pt] (a1) at (-0.5, -0.75) {$\bullet$};
\node[inner sep=0.7pt,font=\scriptsize] (b1) at (0.5, -0.75) {$\triangle$};
\node[inner sep=0.7pt,font=\scriptsize] (a2) at (-0.5, -0.75-1) {$\blacksquare$};
\node[inner sep=0.7pt,font=\scriptsize] (b2) at (0.5, -0.75-1) {$\square$};
\node[inner sep=0.7pt] (a3) at (-0.5, -0.75-2) {$\bullet$};
\node[inner sep=0.7pt,font=\scriptsize] (b3) at (0.5, -0.75-2) {$\triangle$};
\node[inner sep=0.7pt,font=\scriptsize] (a4) at (-0.5, -0.75-3) {$\blacksquare$};
\node[inner sep=0.7pt,font=\scriptsize] (b4) at (0.5, -0.75-3) {$\square$};
\node[inner sep=0.7pt] (c) at (-0.5, -0.75-3.2) {};
\node[inner sep=0.7pt] (d) at (0.5, -0.75-3.2) {};
\node[inner sep=0.7pt] (c1) at (-0.5, -0.75-3-0.5) {};
\node[inner sep=0.7pt] (d1) at (0.5, -0.75-3-0.5) {};

\draw[arrows = {-Stealth[scale=0.9]}] (t) to (a1);
\draw[arrows = {-Stealth[scale=0.9]}] (b1) to (t);
\draw[arrows = {-Stealth[scale=0.9]}] (b1) to (a2);
\draw[arrows = {-Stealth[scale=0.9]}] (b1) to (b2);
\draw[arrows = {-Stealth[scale=0.9]}] (a2) to (a1);
\draw[arrows = {-Stealth[scale=0.9]}] (b2) to (a1);

\draw[arrows = {-Stealth[scale=0.9]}] (b2) to (a3);
\draw[arrows = {-Stealth[scale=0.9]}] (a2) to (a3);
\draw[arrows = {-Stealth[scale=0.9]}] (b3) to (b2);
\draw[arrows = {-Stealth[scale=0.9]}] (b3) to (a2);

\draw[arrows = {-Stealth[scale=0.9]}] (b4) to (a3);
\draw[arrows = {-Stealth[scale=0.9]}] (a4) to (a3);
\draw[arrows = {-Stealth[scale=0.9]}] (b3) to (a4);
\draw[arrows = {-Stealth[scale=0.9]}] (b3) to (b4);
\draw[arrows = {-Stealth[scale=0.9]}] (b3) to (a4);
\draw[arrows = {-Stealth[scale=0.9]}] (b3) to (b4);
\draw[ dotted,thick] (c) to (c1);
\draw[ dotted,thick] (d) to (d1);

\node[inner sep=0.7pt,font=\scriptsize] (t1) at (0.5-2.3, 0.9-1) {${F_{r,s^\vee;1}}$};
\node[inner sep=0.7pt,font=\scriptsize] (tt) at (0.5-2, 0.25-1) {$\square$};
\node[inner sep=0.7pt] (aa1) at (-0.5-2, -0.75-1) {$\bullet$};
\node[inner sep=0.7pt,font=\scriptsize] (bb1) at (0.5-2, -0.75-1) {$\triangle$};
\node[inner sep=0.7pt,font=\scriptsize] (aa2) at (-0.5-2, -0.75-1-1) {$\blacksquare$};
\node[inner sep=0.7pt,font=\scriptsize] (bb2) at (0.5-2, -0.75-1-1) {$\square$};
\node[inner sep=0.7pt] (aa3) at (-0.5-2, -0.75-2-1) {$\bullet$};
\node[inner sep=0.7pt,font=\scriptsize] (bb3) at (0.5-2, -0.75-2-1) {$\triangle$};
\node[inner sep=0.7pt] (cc) at (-0.5-2, -0.75-3.2) {};
\node[inner sep=0.7pt] (dd) at (0.5-2, -0.75-3.2) {};
\node[inner sep=0.7pt] (cc1) at (-0.5-2, -0.75-3-0.5) {};
\node[inner sep=0.7pt] (dd1) at (0.5-2, -0.75-3-0.5) {};

\draw[arrows = {-Stealth[scale=0.9]}] (tt) to (aa1);
\draw[arrows = {-Stealth[scale=0.9]}] (bb1) to (tt);
\draw[arrows = {-Stealth[scale=0.9]}] (bb1) to (aa2);
\draw[arrows = {-Stealth[scale=0.9]}] (bb1) to (bb2);
\draw[arrows = {-Stealth[scale=0.9]}] (aa2) to (aa1);
\draw[arrows = {-Stealth[scale=0.9]}] (bb2) to (aa1);

\draw[arrows = {-Stealth[scale=0.9]}] (bb2) to (aa3);
\draw[arrows = {-Stealth[scale=0.9]}] (aa2) to (aa3);
\draw[arrows = {-Stealth[scale=0.9]}] (bb3) to (bb2);
\draw[arrows = {-Stealth[scale=0.9]}] (bb3) to (aa2);

\draw[ dotted,thick] (cc) to (cc1);
\draw[ dotted,thick] (dd) to (dd1);

\node[inner sep=0.7pt,font=\scriptsize] (t1) at (0.5-2-2-0.3, 0.25-1-1+0.7) {${F_{r,s;2}}$};
\node[inner sep=0.7pt,font=\scriptsize] (ttt) at (0.5-2-2, 0.25-1-1) {$\square$};
\node[inner sep=0.7pt] (aaa1) at (-0.5-2-2, -0.75-1-1) {$\bullet$};
\node[inner sep=0.7pt,font=\scriptsize] (bbb1) at (0.5-2-2, -0.75-1-1) {$\triangle$};
\node[inner sep=0.7pt,font=\scriptsize] (aaa2) at (-0.5-2-2, -0.75-1-1-1) {$\blacksquare$};
\node[inner sep=0.7pt,font=\scriptsize] (bbb2) at (0.5-2-2, -0.75-1-1-1) {$\square$};
\node[inner sep=0.7pt] (ccc) at (-0.5-2-2, -0.75-3.2) {};
\node[inner sep=0.7pt] (ddd) at (0.5-2-2, -0.75-3.2) {};
\node[inner sep=0.7pt] (ccc1) at (-0.5-2-2, -0.75-3-0.5) {};
\node[inner sep=0.7pt] (ddd1) at (0.5-2-2, -0.75-3-0.5) {};

\draw[arrows = {-Stealth[scale=0.9]}] (ttt) to (aaa1);
\draw[arrows = {-Stealth[scale=0.9]}] (bbb1) to (ttt);
\draw[arrows = {-Stealth[scale=0.9]}] (bbb1) to (aaa2);
\draw[arrows = {-Stealth[scale=0.9]}] (bbb1) to (bbb2);
\draw[arrows = {-Stealth[scale=0.9]}] (aaa2) to (aaa1);
\draw[arrows = {-Stealth[scale=0.9]}] (bbb2) to (aaa1);

\draw[ dotted,thick] (ccc) to (ccc1);
\draw[ dotted,thick] (ddd) to (ddd1);

\node[inner sep=0.7pt,font=\scriptsize] (t1) at (0.5-2-2-2-0.3, 0.25-1-1-1+0.7) {${F_{r,s^\vee;3}}$};
\node[inner sep=0.7pt,font=\scriptsize] (tttt) at (0.5-2-2-2, 0.25-1-1-1) {$\square$};
\node[inner sep=0.7pt] (aaaa1) at (-0.5-2-2-2, -0.75-1-1-1) {$\bullet$};
\node[inner sep=0.7pt,font=\scriptsize] (bbbb1) at (0.5-2-2-2, -0.75-1-1-1) {$\triangle$};
\node[inner sep=0.7pt] (cccc) at (-0.5-2-2-2, -0.75-3.2) {};
\node[inner sep=0.7pt] (dddd) at (0.5-2-2-2, -0.75-3.2) {};
\node[inner sep=0.7pt] (cccc1) at (-0.5-2-2-2, -0.75-3-0.5) {};
\node[inner sep=0.7pt] (dddd1) at (0.5-2-2-2, -0.75-3-0.5) {};

\draw[arrows = {-Stealth[scale=0.9]}] (tttt) to (aaaa1);
\draw[arrows = {-Stealth[scale=0.9]}] (bbbb1) to (tttt);

\draw[ dotted,thick] (cccc) to (cccc1);
\draw[ dotted,thick] (dddd) to (dddd1);

\node[inner sep=0.7pt,font=\scriptsize] (t1) at (-0.5-2+4+0.3, 0.25-1+0.7) {${F_{r,s^\vee;-1}}$};
\node[inner sep=0.7pt,font=\scriptsize] (ttr) at (-0.5-2+4, 0.25-1) {$\blacksquare$};
\node[inner sep=0.7pt] (aar1) at (-0.5-2+4, -0.75-1) {$\bullet$};
\node[inner sep=0.7pt,font=\scriptsize] (bbr1) at (0.5-2+4, -0.75-1) {$\triangle$};
\node[inner sep=0.7pt,font=\scriptsize] (aar2) at (-0.5-2+4, -0.75-1-1) {$\blacksquare$};
\node[inner sep=0.7pt,font=\scriptsize] (bbr2) at (0.5-2+4, -0.75-1-1) {$\square$};
\node[inner sep=0.7pt] (aar3) at (-0.5-2+4, -0.75-2-1) {$\bullet$};
\node[inner sep=0.7pt,font=\scriptsize] (bbr3) at (0.5-2+4, -0.75-2-1) {$\triangle$};
\node[inner sep=0.7pt] (ccr) at (-0.5-2+4, -0.75-3.2) {};
\node[inner sep=0.7pt] (ddr) at (0.5-2+4, -0.75-3.2) {};
\node[inner sep=0.7pt] (ccr1) at (-0.5-2+4, -0.75-3-0.5) {};
\node[inner sep=0.7pt] (ddr1) at (0.5-2+4, -0.75-3-0.5) {};

\draw[arrows = {-Stealth[scale=0.9]}] (ttr) to (aar1);
\draw[arrows = {-Stealth[scale=0.9]}] (bbr1) to (ttr);
\draw[arrows = {-Stealth[scale=0.9]}] (bbr1) to (aar2);
\draw[arrows = {-Stealth[scale=0.9]}] (bbr1) to (bbr2);
\draw[arrows = {-Stealth[scale=0.9]}] (aar2) to (aar1);
\draw[arrows = {-Stealth[scale=0.9]}] (bbr2) to (aar1);

\draw[arrows = {-Stealth[scale=0.9]}] (bbr2) to (aar3);
\draw[arrows = {-Stealth[scale=0.9]}] (aar2) to (aar3);
\draw[arrows = {-Stealth[scale=0.9]}] (bbr3) to (bbr2);
\draw[arrows = {-Stealth[scale=0.9]}] (bbr3) to (aar2);

\draw[ dotted,thick] (ccr) to (ccr1);
\draw[ dotted,thick] (ddr) to (ddr1);

\node[inner sep=0.7pt,font=\scriptsize] (t1) at (-0.5+4+0.3, 0.25-1-1+0.7) {${F_{r,s;-2}}$};
\node[inner sep=0.7pt,font=\scriptsize] (tttr) at (-0.5-2-2+8, 0.25-1-1) {$\blacksquare$};
\node[inner sep=0.7pt] (aaar1) at (-0.5-2-2+8, -0.75-1-1) {$\bullet$};
\node[inner sep=0.7pt,font=\scriptsize] (bbbr1) at (0.5-2-2+8, -0.75-1-1) {$\triangle$};
\node[inner sep=0.7pt,font=\scriptsize] (aaar2) at (-0.5-2-2+8, -0.75-1-1-1) {$\blacksquare$};
\node[inner sep=0.7pt,font=\scriptsize] (bbbr2) at (0.5-2-2+8, -0.75-1-1-1) {$\square$};
\node[inner sep=0.7pt] (cccr) at (-0.5-2-2+8, -0.75-3.2) {};
\node[inner sep=0.7pt] (dddr) at (0.5-2-2+8, -0.75-3.2) {};
\node[inner sep=0.7pt] (cccr1) at (-0.5-2-2+8, -0.75-3-0.5) {};
\node[inner sep=0.7pt] (dddr1) at (0.5-2-2+8, -0.75-3-0.5) {};

\draw[arrows = {-Stealth[scale=0.9]}] (tttr) to (aaar1);
\draw[arrows = {-Stealth[scale=0.9]}] (bbbr1) to (tttr);
\draw[arrows = {-Stealth[scale=0.9]}] (bbbr1) to (aaar2);
\draw[arrows = {-Stealth[scale=0.9]}] (bbbr1) to (bbbr2);
\draw[arrows = {-Stealth[scale=0.9]}] (aaar2) to (aaar1);
\draw[arrows = {-Stealth[scale=0.9]}] (bbbr2) to (aaar1);

\draw[ dotted,thick] (cccr) to (cccr1);
\draw[ dotted,thick] (dddr) to (dddr1);

\node[inner sep=0.7pt,font=\scriptsize] (t1) at (-0.5+2+4+0.3, 0.25-1-1-1+0.7) {${F_{r,s^\vee;-3}}$};
\node[inner sep=0.7pt,font=\scriptsize] (ttttr) at (-0.5-2-2-2+12, 0.25-1-1-1) {$\blacksquare$};
\node[inner sep=0.7pt] (aaaar1) at (-0.5-2-2-2+12, -0.75-1-1-1) {$\bullet$};
\node[inner sep=0.7pt,font=\scriptsize] (bbbbr1) at (0.5-2-2-2+12, -0.75-1-1-1) {$\triangle$};
\node[inner sep=0.7pt] (ccccr) at (-0.5-2-2+10, -0.75-3.2) {};
\node[inner sep=0.7pt] (ddddr) at (0.5-2-2+10, -0.75-3.2) {};
\node[inner sep=0.7pt] (ccccr1) at (-0.5-2-2+10, -0.75-3-0.5) {};
\node[inner sep=0.7pt] (ddddr1) at (0.5-2-2+10, -0.75-3-0.5) {};

\draw[arrows = {-Stealth[scale=0.9]}] (ttttr) to (aaaar1);
\draw[arrows = {-Stealth[scale=0.9]}] (bbbbr1) to (ttttr);

\draw[ dotted,thick] (ccccr) to (ccccr1);
\draw[ dotted,thick] (ddddr) to (ddddr1);

\draw[->,>=stealth] (0.7, -0.75)--(1.3,-0.75);
\draw[->,>=stealth] (0.7, -0.75-1)--(1.3,-0.75-1);
\draw[->,>=stealth] (0.7, -0.75-1-1)--(1.3,-0.75-1-1);
\draw[->,>=stealth] (0.7, -0.75-1-1-1)--(1.3,-0.75-1-1-1);

\draw[->,>=stealth] (0.7+2, -0.75-1)--(1.3+2,-0.75-1);
\draw[->,>=stealth] (0.7+2, -0.75-1-1)--(1.3+2,-0.75-1-1);
\draw[->,>=stealth] (0.7+2, -0.75-1-1-1)--(1.3+2,-0.75-1-1-1);

\draw[->,>=stealth] (0.7+2+2, -0.75-1-1)--(1.3+2+2,-0.75-1-1);
\draw[->,>=stealth] (0.7+2+2, -0.75-1-1-1)--(1.3+2+2,-0.75-1-1-1);

\draw[->,>=stealth] (0.7-2, -0.75)--(1.3-2,-0.75);
\draw[->,>=stealth] (0.7-2, -0.75-1)--(1.3-2,-0.75-1);
\draw[->,>=stealth] (0.7-2, -0.75-1-1)--(1.3-2,-0.75-1-1);
\draw[->,>=stealth] (0.7-2, -0.75-1-1-1)--(1.3-2,-0.75-1-1-1);

\draw[->,>=stealth] (0.7-2-2, -0.75-1)--(1.3-2-2,-0.75-1);
\draw[->,>=stealth] (0.7-2-2, -0.75-1-1)--(1.3-2-2,-0.75-1-1);
\draw[->,>=stealth] (0.7-2-2, -0.75-1-1-1)--(1.3-2-2,-0.75-1-1-1);

\draw[->,>=stealth] (0.7-2-2-2, -0.75-1-1)--(1.3-2-2-2,-0.75-1-1);
\draw[->,>=stealth] (0.7-2-2-2, -0.75-1-1-1)--(1.3-2-2-2,-0.75-1-1-1);

\end{tikzpicture}
\caption{\label{51fig-2-fel}Schematic diagram of the complex (\ref{fel-comp:p>1}). Horizontal arrows represent the action of the screening operators $Q^{[\bullet]}_-$. The black circles and the black squares represent the generating vectors of ${\rm im}Q^{[\bullet]}_-\cap {\rm im}Q^{[r^\vee]}_+$ and ${\rm ker}Q^{[\bullet]}_-$, respectively. The white squares and triangles represent the generating vectors of ${\rm im}Q^{[r^\vee]}_+$ and ${\rm coker}Q^{[\bullet]}_-$, respectively.
}
\mbox{}\\
\end{figure}
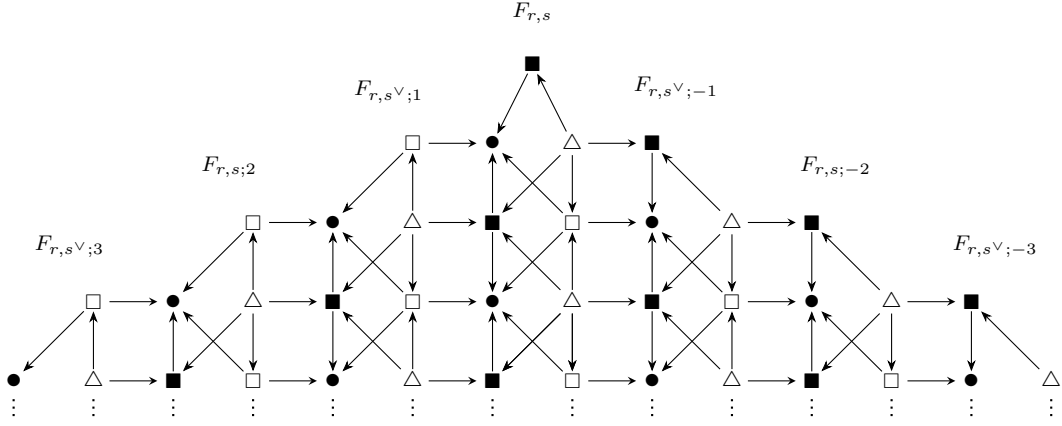

\section{Derivations on $\mathcal{W}_{p_+,p_-}$} 
\label{sec-deri-main}
In this section, we first review the structure of $\mathcal{W}_{p_+,p_-}$ following \cite{AdamovicD/MilasA:2008,AdamovicD/MilasA:2010,AdamovicD/MilasA:2011,TW}, and then introduce the derivations on $\mathcal{W}_{p_+,p_-}$.
\subsection{The triplet $W$-algebras $\mathcal{W}_{p_+,p_-}$}
\label{tripletW}
Let $p_+$, $p_-$ be coprime integers satisfying $p_->p_+\geq 1$. 
\begin{dfn}
The {\rm lattice vertex operator algebra} $\mathcal{V}_{[p_+,p_-]}$ associated with the rank-one lattice $\mathbb{Z}\sqrt{2p_+p_-}$ is given by the tuple 
\begin{align*}
(\mathcal{V}^{+}_{1,1}, \ket{0}, \frac{1}{2}(a^2_{-1}-\alpha_0a_{-2})\ket{0}, Y),
\end{align*}
where the underlying vector space $\mathcal{V}^{+}_{1,1}$ is given by
\begin{align*}
\mathcal{V}^{+}_{1,1}=\bigoplus_{n\in\mathbb{Z}}F_{1,1;2n}=\bigoplus_{n\in\mathbb{Z}}F_{n\sqrt{2p_+p_-}}.
\end{align*}
\end{dfn}
It is known that simple $\mathcal{V}_{[p_+,p_-]}$-modules are given by the $2p_+p_-$ direct sum of Fock modules
\begin{align*}
&\mathcal{V}^{+}_{r,s}:=\bigoplus_{n\in\mathbb{Z}}F_{r,s;2n},
&\mathcal{V}^{-}_{r,s}:=\bigoplus_{n\in\mathbb{Z}}F_{r,s;2n+1},
\end{align*}
for $1\leq r\leq p_+,1\leq s\leq p_-$.

Note that the two screening operators $Q_+$ and $Q_-$ act on $\mathcal{V}^+_{1,1}$. We define the following vector subspace of $\mathcal{V}^+_{1,1}$:
\begin{equation*}
\mathcal{K}_{1,1}
=
\begin{cases}
{\rm ker}Q_-& p_+=1\\
{\rm ker}Q_+\cap{\rm ker}Q_-& p_+\geq 2
\end{cases}
.
\end{equation*}
\begin{dfn}\cite{FGST,FHST}
The {\rm triplet $W$-algebra} is defined by the following vertex operator algebra
\begin{align*}
\mathcal{W}_{p_+,p_-}:=(\mathcal{K}_{1,1},\ket{0},T,Y),
\end{align*}
where the vacuum vector, conformal vector and vertex operator map are those of $\mathcal{V}_{[p_+,p_-]}$.
\end{dfn}
We define the triplet vectors
\begin{align*}
&W^+:=Q^{[p_--1]}_-\ket{\alpha_{1,p_--1;3}},
&W^-&:=Q^{[p_+-1]}_+\ket{\alpha_{p_+-1,1;-3}},\\ 
&W^0:=Q^{[2p_+-1]}_+\ket{\alpha_{p_+-1,1;-3}},
\end{align*}
where in the case $p_+=1$, we set $W^-:=\ket{\alpha_{3,1}}$.
These vectors are non-trivial Virasoro singular vectors with $\mathcal{L}_0$-weight $h_{4p_+-1,1}$.
\begin{prop}[\cite{AdamovicD/MilasA:2008,AdamovicD/MilasA:2010,AdamovicD/MilasA:2011,TW}]
\label{genW}
$\mathcal{W}_{p_+,p_-}$ is strongly  generated by the fields $T(z),Y(W^{\pm},z),Y(W^0,z)$.
\end{prop}

For the case $p_+\geq 2$, we define a subspace $\mathcal{I}_{p_+,p_-}\subset \mathcal{W}_{p_+,p_-}$ as follows
\begin{align*}
\mathcal{I}_{p_+,p_-}=Q^{[p_+-1]}_+(\mathcal{V}^-_{p_+-1,1})\cap Q^{[p_--1]}_-(\mathcal{V}^-_{1,p_--1}).
\end{align*}
For $n\in \mathbb{Z}_{\geq 0}$, we set
\begin{equation*}
\Lambda_{n}(p_+,p_-):=
\begin{cases}
h_{p_+-1,1;-2n-1}&p_+\geq 2\\
h_{p_+,1;-2n}&p_+=1
\end{cases}
\end{equation*}
and for $-n\leq i\leq n$,
\begin{equation*}
w^{(n)}_{i}(p_+,p_-):=
\begin{cases}
Q^{[(n+i+1)p_+-1]}_+\ket{\alpha_{p_+-1,1;-2n-1}} & p_+\geq 2\\
Q^{(n+i)}_+\ket{\alpha_{p_+,1;-2n}} & p_+=1
\end{cases}
.
\end{equation*}
Note that $w^{(1)}_{-1}$, $w^{(1)}_0$, and $w^{(1)}_1$ agree with $W^-$, $W^0$, and $W^+$, respectively, up to scalar multiples.
By Proposition \ref{Felder complex}, $\mathcal{W}_{1,p_-}$ and $\mathcal{I}_{p_+,p_-}$ are isomorphic, as Virasoro modules, to
\begin{align*}
\bigoplus_{n\in \mathbb{Z}_{\geq 0}}(2n+1)L(\Lambda_{n}(p_+,p_-)),
\end{align*}
and $\{w^{(n)}_{i}(p_+,p_-)\}_{i=-n}^n$ correspond to the lowest weight vectors of $(2n+1)L(\Lambda_{n}(p_+,p_-))$.
We set
\begin{align*}
W^{\delta}[n]=\oint_{z=0}Y(W^\delta,z)z^{h_{4p_+-1,1}+n-1}{\rm d}z,\qquad \delta=\pm,0,\ \ n\in \mathbb{Z},
\end{align*}
which is the $n$-th mode of the field $Y(W^\delta,z)$. Hereafter, we simply write $w^{(n)}_{i}=w^{(n)}_{i}(p_+,p_-)$ and $\Lambda_{n}=\Lambda_{n}(p_+,p_-)$.
It is known that $\mathcal{W}_{1,p_-}$ is simple \cite{AdamovicD/MilasA:2008}, and for $p_+\geq 2$, $\mathcal{I}_{p_+,p_-}$ is a simple ideal of $\mathcal{W}_{p_+,p_-}$ \cite{AdamovicD/MilasA:2010,AdamovicD/MilasA:2011,TW}; namely, there exists a short exact sequence
\begin{align*}
0\rightarrow \mathcal{I}_{p_+,p_-}\rightarrow \mathcal{W}_{p_+,p_-}\rightarrow L(0)\rightarrow 0,\qquad p_+\geq 2
\end{align*}
as $\mathcal{W}_{p_+,p_-}$-modules.
The following proposition plays an important role in proving the simplicity of $\mathcal{W}_{1,p}$, $\mathcal{I}_{p_+,p_-}$ and determining the structure of the Zhu algebra $A(\mathcal{W}_{p_+,p_-})$ (for the definition of the Zhu algebra, see \cite{Zh}).
\begin{prop}[\cite{AdamovicD/MilasA:2008,AdamovicD/MilasA:2010,AdamovicD/MilasA:2011,TW}]
The Virasoro lowest weight vectors $w^{(n)}_{i}$ $(n\geq 0,-n\leq i\leq n)$ satisfy
\begin{align}
w^{(n+1)}_{i\pm 1}&\in \mathbb{C}^\times W^\pm[\Lambda_n-\Lambda_{n+1}]w^{(n)}_{i} +\sum_{k=0}^{n}U(\mathfrak{Vir})w^{(k)}_{i\pm 1},\label{wpq-act}
\end{align}
where we set $w^{(n)}_{n+1}=w^{(n)}_{-n-1}=0$ and $U(\mathfrak{Vir})$ is the universal enveloping algebra of $\mathfrak{Vir}$.
\label{sl2action2}
\end{prop}

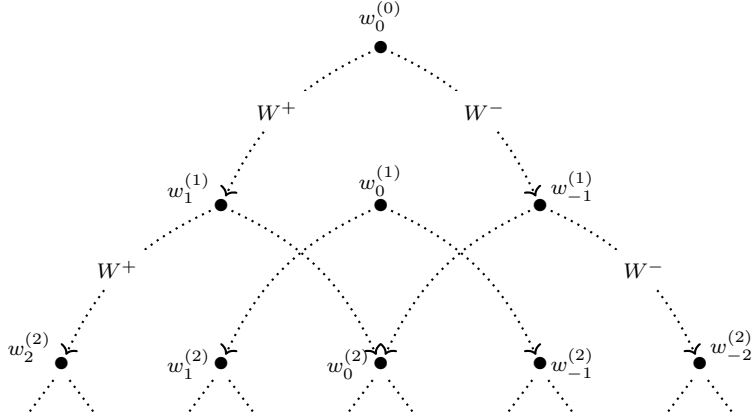
\begin{figure}[h]
  \centering
\begin{tikzpicture}[scale=1.05]

\node[inner sep=0.7pt,font=\scriptsize] (u0) at (0, 0+0.4) {$w^{(0)}_{0}$};
\node[inner sep=0.7pt] (u) at (0, 0) {$\bullet$};

\node[inner sep=0.7pt,font=\scriptsize,above left] (u0) at (0-2-0.1, 0-2) {$w^{(1)}_{1}$};
\node[inner sep=0.7pt,font=\scriptsize,above right] (u0) at (0+2+0.1, 0-2) {$w^{(1)}_{-1}$};
\node[inner sep=0.7pt,font=\scriptsize,above] (u0) at (0, 0-2+0.1) {$w^{(1)}_{0}$};
\node[inner sep=0.7pt] (w1) at (-2, -2) {$\bullet$};
\node[inner sep=0.7pt] (w0) at (0, -2) {$\bullet$};
\node[inner sep=0.7pt] (w-1) at (2, -2) {$\bullet$};

\node[inner sep=0.7pt,font=\scriptsize,above left] (u0) at (0-2-0.1-2, 0-2-2) {$w^{(2)}_{2}$};
\node[inner sep=0.7pt,font=\scriptsize,above right] (u0) at (0+2+0.1+2, 0-2-2) {$w^{(2)}_{-2}$};
\node[inner sep=0.7pt] (y2) at (-4, -4) {$\bullet$};
\node[inner sep=0.7pt] (y1) at (-2, -4) {$\bullet$};
\node[inner sep=0.7pt,left,font=\scriptsize] (t1) at (-2.1, -4) {$w^{(2)}_1$};
\node[inner sep=0.7pt] (y0) at (0, -4) {$\bullet$};
\node[inner sep=0.7pt,left,font=\scriptsize] (t1) at (-0.1, -4) {$w^{(2)}_0$};
\node[inner sep=0.7pt] (y-1) at (2, -4) {$\bullet$};
\node[inner sep=0.7pt,right,font=\scriptsize] (t1) at (2.1, -4) {$w^{(2)}_{-1}$};
\node[inner sep=0.7pt] (y-2) at (4, -4) {$\bullet$};

\draw[->,dotted,thick] (w-1) to[bend right=20]  (y0);
\draw[->,dotted,thick] (w1) to[bend right=-20]  (y0);

\draw[->,dotted,thick] (w0) to[bend right=20]  (y1);
\draw[->,dotted,thick] (w0) to[bend right=-20]  (y-1);

\draw[ dotted,thick] (4.1, -4.2) to (4.4,-4.6);
\draw[ dotted,thick] (3.9, -4.2) to (3.6,-4.6);
\draw[ dotted,thick] (-4.1, -4.2) to (-4.4,-4.6);
\draw[ dotted,thick] (-3.9, -4.2) to (-3.6,-4.6);
\draw[ dotted,thick] (-2.1, -4.2) to (-2.4,-4.6);
\draw[ dotted,thick] (-1.9, -4.2) to (-1.6,-4.6);
\draw[ dotted,thick] (-0.1, -4.2) to (-0.4,-4.6);
\draw[ dotted,thick] (0.1, -4.2) to (0.4,-4.6);
\draw[ dotted,thick] (2.1, -4.2) to (2.4,-4.6);
\draw[ dotted,thick] (1.9, -4.2) to (1.6,-4.6);

\draw[->,dotted,thick] (u) to[bend right=20]  (w1);
\draw (-1-0.3,-1+0.2) node[fill=white,font=\scriptsize]{$W^+$};
\draw[->,dotted,thick] (u) to[bend right=-20]  (w-1);
\draw (1+0.3,-1+0.2) node[fill=white,font=\scriptsize]{$W^-$};

\draw[->,dotted,thick] (w1) to[bend right=20]  (y2);
\draw (-1-0.3-2,-1+0.2-2) node[fill=white,font=\scriptsize]{$W^+$};
\draw[->,dotted,thick] (w-1) to[bend right=-20]  (y-2);
\draw (1+0.3+2,-1+0.2-2) node[fill=white,font=\scriptsize]{$W^-$};
\end{tikzpicture}
\caption{\label{51fig-wpq}A schematic diagram of each action in Proposition~\ref{sl2action2}. The black dots represent the vectors $w^{(n)}_i$, and the arrows indicate a part of (\ref{wpq-act}).
}
\mbox{}\\
\end{figure}

See Figure~\ref{51fig-wpq} for a part of the action described in this proposition.
The derivations acting on $\mathcal{W}_{p_+,p_-}$ play a crucial role in the proof of this proposition. In the following two subsections, we give a detailed construction of these derivations and discuss their properties.

\subsection{Linear operators $\mathbb{G}^{[p_\pm]}_\pm$}
\label{cons-main0}
Recall the screening operators $Q^{[\bullet]}_\pm$ defined by (\ref{def-pq-sc}).
It is known that, when $\rho = \alpha_0$ in $\mathcal{Q}^{[\bullet];(\rho)}_\pm$, the analytic behavior prevents a direct definition of twisted cycles \cite{AK,TK,TW} (see Remark~\ref{rem:aktk}).
To resolve this difficulty, \cite{TW} introduced an $\epsilon$-deformation method, by which the problem of constructing twisted cycles is reduced to analyzing the holomorphicity of the image of the screening operator at $\epsilon=0$. 

To describe this $\epsilon$-deformation method, we introduce some notation.
Following \cite[Section 3]{TW}, we set
\begin{align}
\label{notep1}
\widetilde{\alpha}_+(\epsilon)=\alpha_++\epsilon,\qquad\qquad
\widetilde{\alpha}_-(\epsilon)=-\frac{2}{\widetilde{\alpha}_+(\epsilon)},\qquad\qquad
\widetilde{\alpha}_0(\epsilon)=\widetilde{\alpha}_+(\epsilon)+\widetilde{\alpha}_-(\epsilon),
\end{align}
for $\epsilon \in \mathbb{C}$.
Note that $\widetilde{\alpha}_\pm(0)=\alpha_\pm$ and $\widetilde{\alpha}_0(0)=\alpha_0$. 
The elements $\widetilde{\alpha}_\pm(\epsilon)$ can be regarded as $\epsilon$-deformations of $\alpha_\pm$ that preserve the relation $\alpha_+\alpha_-=-2$.
We define
\begin{align*}
\begin{aligned}
\mathcal{O}_{0;\epsilon}&:=\{\ g(\epsilon)\ |\ g(\epsilon)\ {\rm is}\ {\rm holomorphic}\ {\rm on}\ {\rm some}\ {\rm neighborhood}\ U\ {\rm of}\ 0\}.
\end{aligned}
\end{align*}
For $\gamma(\epsilon)\in \mathcal{O}_{0;\epsilon}$, we denote by $\widetilde{F}^{(\rho)}_{\gamma(\epsilon)}$ the Fock module over $\mathcal{O}_{0;\epsilon}$ 
generated by a vector $\ket{\gamma(\epsilon)}$ satisfying
\begin{align*}
a_0\ket{\gamma(\epsilon)}=\gamma(\epsilon)\ket{\gamma(\epsilon)}, 
\qquad 
a_n\ket{\gamma(\epsilon)}=0 \quad (n>0).
\end{align*}
Here, we regard $\widetilde{F}^{(\rho)}_{\gamma(\epsilon)}$ as a $\pi^{(\rho)}_{\gamma(\epsilon)}(U(\mathfrak{Vir}))$-module in the natural way, for compatibility with the notation introduced in Subsection~\ref{subsc}.
We introduce the notation
\begin{align}
\label{notep2}
\begin{aligned}
\widetilde{\alpha}_{r,s}(\epsilon)&:=\alpha^{(\widetilde{\alpha}_{0}(\epsilon))}_{r,s},&\widetilde{F}_{r,s}&:=\widetilde{F}^{(\widetilde{\alpha}_{0}(\epsilon))}_{\widetilde{\alpha}_{r,s}(\epsilon)},\\
\widetilde{Q}^{[r]}_+(z)&:=\mathcal{Q}^{(\widetilde{\alpha}_0(\epsilon));[r]}_+(z),&\widetilde{Q}^{[s]}_-(z)&:=\mathcal{Q}^{(\widetilde{\alpha}_0(\epsilon));[s]}_-(z),\\
\widetilde{Q}^{[r]}_+&:=\mathcal{Q}^{(\widetilde{\alpha}_0(\epsilon));[r]}_+,&\widetilde{Q}^{[s]}_-&:=\mathcal{Q}^{(\widetilde{\alpha}_0(\epsilon));[s]}_-,
\end{aligned}
\end{align}
and we set $\widetilde{Q}_\pm(z)=\widetilde{Q}^{[1]}_\pm(z)$, $\widetilde{Q}_\pm=\widetilde{Q}^{[1]}_\pm$. 
Note that $\widetilde{\alpha}_{r,s}(0)=\alpha_{r,s}$.
Then we shall call $\widetilde{F}_{r,s}$ an $\epsilon$-deformation or an $\epsilon$-lifting of $F_{r,s}$.
To simplify the notation, we will often omit $(\epsilon)$ from each notation of (\ref{notep1}) and (\ref{notep2}).

Let $p_+-1\geq r\geq 2$, $p_--1\geq s\geq 2$. Since 
\begin{align}
\label{cond-Ahat}
\begin{aligned}
&\Bigl((1-r)\frac{\widetilde{\alpha}_+(\epsilon)^2}{2},\widetilde{\alpha}_+(\epsilon)^2,\frac{\widetilde{\alpha}_+(\epsilon)^2}{2}\Bigr)\notin \widehat{\mathcal{A}}_{r-1},\qquad p_+-1\geq r\geq 2,\\
&\Bigl((1-s)\frac{\widetilde{\alpha}_-(\epsilon)^2}{2},\widetilde{\alpha}_-(\epsilon)^2,\frac{\widetilde{\alpha}_-(\epsilon)^2}{2}\Bigr)\notin \widehat{\mathcal{A}}_{s-1},\qquad p_--1\geq s\geq 2
\end{aligned}
\end{align}
hold for sufficiently small values of $|\epsilon|\geq 0$, we see that the screening operators $\widetilde{Q}^{[r]}_+$ and $\widetilde{Q}^{[s]}_-$
are well-defined on $\widetilde{F}_{r,k}$ and $\widetilde{F}_{k,s}$ $(k\in \mathbb{Z})$, respectively (see Subsection \ref{subsc}). 
Then, from Proposition~\ref{TK-thm} and Theorem~\ref{sus-prop}, we have the following proposition.
\begin{prop}[\cite{TW}]
Let $p_+-1\geq r\geq 1$, $p_--1\geq s\geq 1$, and $k\in \mathbb{Z}$.
Then we have
\begin{align*}
&\widetilde{Q}^{[r]}_+({\widetilde{F}_{r,k}})\subset \widetilde{F}_{-r,k},
&\widetilde{Q}^{[s]}_-({\widetilde{F}_{k,s}})\subset \widetilde{F}_{k,-s}.
\end{align*}
In particular, we can set $\epsilon=0$ on the images of $\widetilde{Q}^{[r]}_+\circ e^{(\widetilde{\alpha}_{r,k}-\alpha_{r,k})\hat{a}}|_{F_{r,k}}$ and $\widetilde{Q}^{[s]}_-\circ e^{(\widetilde{\alpha}_{k,s}-\alpha_{k,s})\hat{a}}|_{F_{k,s}}$. Then, by the definition of $\widetilde{Q}^{[\bullet]}_\pm$ (see (\ref{notep2}) and Proposition~\ref{com-TK}), we have
\begin{align}
\label{eq:def-qq}
\begin{aligned}
&\widetilde{Q}^{[r]}_+\circ e^{(\widetilde{\alpha}_{r,k}-\alpha_{r,k})\hat{a}}|_{F_{r,k}}|_{\epsilon=0}={Q}^{[r]}_+|_{F_{r,k}},
&\widetilde{Q}^{[s]}_-\circ e^{(\widetilde{\alpha}_{k,s}-\alpha_{k,s})\hat{a}}|_{F_{k,s}}|_{\epsilon=0}={Q}^{[s]}_-|_{F_{k,s}}.
\end{aligned}
\end{align}
\label{prop:tw}
\end{prop}
By this proposition, we can regard $\widetilde{Q}^{[n]}_\pm$ as an $\epsilon$-deformation of ${Q}^{[n]}_\pm$.
\begin{Remark}
\begin{enumerate}
\item Proposition \ref{prop:tw} was proved in \cite{TW} using the theory of Jack symmetric polynomials.
It is known that this proposition also holds when $r$ and $s$ are greater than $p_+$ and $p_-$, respectively, and are not multiples of $p_+$ and $p_-$ \cite{TW}.
This fact can also be shown by using \cite[Theorem 1.2]{Su}.
\item 
The equation (\ref{eq:def-qq}) may be viewed as a defining relation for ${Q}^{[n]}_\pm$ in terms of the deformed screening operators $\widetilde{Q}^{[n]}_\pm$ (see also Remarks \ref{del-sus0} and \ref{del-sus2}).
\end{enumerate}
\label{rem:aktk}
\end{Remark}

We introduce the notation
\begin{align}
{\rm e}^{\alpha}=e^{\alpha \hat{a}}
\label{conjugate2}
\end{align}
(see (\ref{conjugate}), for the definition of $\hat{a}$).
Motivated by the idea of~\cite{TW}, we consider the following two operators, which can be regarded as $\epsilon$-deformations of the compositions 
\begin{align}
\label{fel-dd}
&Q^{[r^\vee]}_+\circ Q^{[r]}_+=0,
&Q^{[s^\vee]}_+\circ Q^{[s]}_-=0
\end{align}
(see Proposition \ref{Felder complex2}).
\begin{dfn}
For $p_+>r\geq 1$, $p_->s\geq 1$, and $k\in \mathbb{Z}$, we define linear operators
\begin{align}
\label{Q+Q}
\begin{aligned}
&N^{[r^\vee,r]}_{+,k}(\epsilon):=\widetilde{Q}^{[p_+-r]}_+\circ{\rm e}^{\widetilde{\alpha}_{p_+-r,k+p_-}-\widetilde{\alpha}_{-r,k}}\circ\widetilde{Q}^{[r]}_+\circ{\rm e}^{\widetilde{\alpha}_{r,k}-\alpha_{r,k}},\\
&N^{[s^\vee,s]}_{-,k}(\epsilon):=\widetilde{Q}^{[p_--s]}_-\circ{\rm e}^{\widetilde{\alpha}_{k+p_+,p_--s}-\widetilde{\alpha}_{k,-s}}\circ\widetilde{Q}^{[s]}_-\circ{\rm e}^{\widetilde{\alpha}_{k,s}-\alpha_{k,s}}
\end{aligned}
\end{align}
whose domains are $F_{r,k}$ and $F_{k,s}$, respectively.
\end{dfn}
Note that the presence of the shift elements ${\rm e}^\bullet$ in (\ref{Q+Q}) ensures that the actions of the deformed screening operators $\widetilde{Q}^{[\bullet]}_\pm$ are well-defined:
\begin{align*}
\begin{aligned}
&{\rm e}^{\widetilde{\alpha}_{r,k}-\alpha_{r,k}}F_{r,k}\subset \widetilde{F}_{r,k},\qquad \qquad {\rm e}^{\widetilde{\alpha}_{k,s}-\alpha_{k,s}}F_{k,s}\subset \widetilde{F}_{k,s},\\
&{\rm e}^{\widetilde{\alpha}_{p_+-r,k+p_-}-\widetilde{\alpha}_{-r,k}}\circ\widetilde{Q}^{[r]}_+\circ{\rm e}^{\widetilde{\alpha}_{r,k}-\alpha_{r,k}}F_{r,k}\subset \widetilde{F}_{p_+-r,k+p_-},\\
&{\rm e}^{\widetilde{\alpha}_{k+p_+,p_--s}-\widetilde{\alpha}_{k,-s}}\circ\widetilde{Q}^{[s]}_-\circ{\rm e}^{\widetilde{\alpha}_{k,s}-\alpha_{k,s}}F_{k,s}\subset \widetilde{F}_{k+p_-,p_--s}.
\end{aligned}
\end{align*}
We also note that the functions in the superscript of the shift elements ${\rm e}^\bullet$ are contained in $\epsilon\cdot\mathcal{O}_{0;\epsilon}$.
By Proposition \ref{prop:tw} and (\ref{fel-dd}), we have
\begin{align*}
\begin{aligned}
&\widetilde{Q}^{[p_+-r]}_+\circ{\rm e}^{\widetilde{\alpha}_{p_+-r,k+p_-}-\widetilde{\alpha}_{-r,k}}\circ\widetilde{Q}^{[r]}_+(\widetilde{F}_{r,k})\subset \epsilon\cdot \widetilde{F}_{r-p_+,k+p_-},\\
&\widetilde{Q}^{[p_--s]}_-\circ{\rm e}^{\widetilde{\alpha}_{k+p_+,p_--s}-\widetilde{\alpha}_{k,-s}}\circ\widetilde{Q}^{[s]}_-(\widetilde{F}_{k,s})
\subset \epsilon\cdot \widetilde{F}_{k+p_+,s-p_-}.
\end{aligned}
\end{align*}
In particular, we get
\begin{align}
\label{N+N-}
\begin{aligned}
&{\rm e}^{\alpha_{r,k+2p_-}-\widetilde{\alpha}_{r-p_+,k+p_-}}\circ N^{[r^\vee,r]}_{+,k}(\epsilon)\in {\rm Hom}_{\mathbb{C}}( F_{r,k},\epsilon\cdot\mathcal{O}_{0;\epsilon}\otimes F_{r,k;2}),\\
&{\rm e}^{\alpha_{k+2p_+,s}-\widetilde{\alpha}_{k+p_+,s-p_-}}\circ N^{[s^\vee,s]}_{-,k}(\epsilon)\in {\rm Hom}_{\mathbb{C}}( F_{k,s},\epsilon\cdot\mathcal{O}_{0;\epsilon}\otimes F_{k,s;-2}).
\end{aligned}
\end{align}
\begin{Remark}
\begin{enumerate}
\item 
By definition, the domain of the holomorphic functions in the image of (\ref{N+N-}) can be uniquely restricted.
More precisely, by fixing a sufficiently small $\delta > 0$ so that condition (\ref{cond-Ahat}) holds when $\epsilon=\delta$ for all $2 \leq r \leq p_+ - 1$ and $2 \leq s \leq p_- - 1$, we have
\begin{align*}
\begin{aligned}
&{\rm e}^{\alpha_{r,k+2p_-}-\widetilde{\alpha}_{r-p_+,k+p_-}}\circ N^{[r^\vee,r]}_{+,k}(\epsilon)\in {\rm Hom}_{\mathbb{C}}( F_{r,k},\epsilon\cdot\mathcal{O}_{\epsilon}(\mathbb{D}_\delta)\otimes F_{r,k;2}),\\
&{\rm e}^{\alpha_{k+2p_+,s}-\widetilde{\alpha}_{k+p_+,s-p_-}}\circ N^{[s^\vee,s]}_{-,k}(\epsilon)\in {\rm Hom}_{\mathbb{C}}( F_{k,s},\epsilon\cdot\mathcal{O}_{\epsilon}(\mathbb{D}_\delta)\otimes F_{k,s;-2}),
\end{aligned}
\end{align*}
where $\mathbb{D}_\delta:=\{\epsilon \in \mathbb{C}\ |\ |\epsilon|<\delta\}$ and
\begin{align*}
\begin{aligned}
\mathcal{O}_{\epsilon}(\mathbb{D}_\delta)&:=\{\ g: \mathbb{D}_\delta\rightarrow \mathbb{C}\ |\ g(\epsilon)\ {\rm is}\ {\rm holomorphic}\ {\rm on}\ \mathbb{D}_\delta\}.
\end{aligned}
\end{align*}

\item By dividing 
\begin{align*}
&{\rm e}^{\alpha_{r,k+2p_-}-\widetilde{\alpha}_{r-p_+,k+p_-}}\circ N^{[r^\vee,r]}_{+,k}(\epsilon),
&{\rm e}^{\alpha_{k+2p_+,s}-\widetilde{\alpha}_{k+p_+,s-p_-}}\circ N^{[s^\vee,s]}_{-,k}(\epsilon)
\end{align*}
by the products
\begin{align*}
\begin{aligned}
&S_{r-1}\bigl((1-r)\frac{\widetilde{\alpha}_+(\epsilon)^2}{2},\widetilde{\alpha}_+(\epsilon)^2,\frac{\widetilde{\alpha}_+(\epsilon)^2}{2}\bigr)S_{r^\vee-1}\bigl((1-r^\vee)\frac{\widetilde{\alpha}_+(\epsilon)^2}{2},\widetilde{\alpha}_+(\epsilon)^2,\frac{\widetilde{\alpha}_+(\epsilon)^2}{2}\bigr),\\
&S_{s-1}\bigl((1-s)\frac{\widetilde{\alpha}_-(\epsilon)^2}{2},\widetilde{\alpha}_-(\epsilon)^2,\frac{\widetilde{\alpha}_-(\epsilon)^2}{2}\bigr)S_{s^\vee-1}\bigl((1-s^\vee)\frac{\widetilde{\alpha}_-(\epsilon)^2}{2},\widetilde{\alpha}_-(\epsilon)^2,\frac{\widetilde{\alpha}_-(\epsilon)^2}{2}\bigr),
\end{aligned}
\end{align*}
respectively, the $\mathcal{O}_{0;\epsilon}$ in (\ref{N+N-}) can be replaced by $\mathbb{C}(\epsilon)$.
This is because the image of $\widetilde{Q}^{[\bullet]}_\pm$ can be described in terms of Jack symmetric plynomials \cite{TW}.
\end{enumerate}
\end{Remark}
Note that any $\mathcal{L}_0$-homogeneous vector of $\widetilde{F}^{(\rho)}_{\gamma(\epsilon)}$ can be written uniquely in one of the following forms
\begin{align*}
\begin{aligned}
 \sum_{\lambda\vdash n}h_\lambda(\epsilon)a_{-\lambda}\ket{\gamma(\epsilon)},\ \ \  n>0,\qquad \qquad
h_0(\epsilon)\ket{\gamma(\epsilon)},
\end{aligned}
\end{align*}
where $h_\lambda(\epsilon), h_0(\epsilon)\in \mathcal{O}_{0;\epsilon}$ (for the notaion $a_{-\lambda}$, see (\ref{a-lambda})).
Using this notation, we introduce the following definition.
\begin{dfn}
Let $\gamma(\epsilon)\in \mathcal{O}_{0;\epsilon}$. Let $\{h_\lambda(\epsilon)\}$ be functions in $\mathcal{O}_{0;\epsilon}$ and indexed by partitions $\lambda$ of $n$. Define
\begin{align}
\begin{aligned}
 \lim_{\epsilon\rightarrow 0}\hspace{-1mm}{}^F\Bigl(\sum_{\lambda\vdash n}h_\lambda(\epsilon)a_{-\lambda}\ket{\gamma(\epsilon)}\Bigr)&:=
\sum_{\lambda\vdash n}\lim_{\substack{\epsilon\rightarrow 0\\ \epsilon\neq0}}h_\lambda(\epsilon)a_{-\lambda}\ket{\gamma(0)},\qquad n>0,\\
\lim_{\epsilon\rightarrow 0}\hspace{-1mm}{}^F\Bigl(h_0(\epsilon)\ket{\gamma(\epsilon)}\Bigr)&:=
\lim_{\substack{\epsilon\rightarrow 0\\ \epsilon\neq0}}h_0(\epsilon)\ket{\gamma(0)}.
\end{aligned}
\end{align}
We extend the above definition linearly to arbitrary elements of $\widetilde{F}^{(\rho)}_{\gamma(\epsilon)}$.
\label{notlim}
\end{dfn}
Note that, since $\gamma(\epsilon)\in \mathcal{O}_{0;\epsilon}$, the action of the zero-mode $a_0$ commutes with the above limit.
\begin{Remark}
From Definition~\ref{notlim} and (\ref{conjugate}), the following equality holds for any $\tilde{u}\in \widetilde{F}^{(\rho)}_{\gamma(\epsilon)}$ and $g(\epsilon)\in \mathcal{O}_{0;\epsilon}$:
\begin{align*}
\lim_{\epsilon\rightarrow 0}\hspace{-1mm}{}^F\tilde{u}=\lim_{\epsilon\rightarrow 0}\hspace{-1mm}{}^F{\rm e}^{\epsilon g(\epsilon)}\tilde{u}.
\end{align*}
\label{Rem:lim}
\end{Remark}

Taking (\ref{N+N-}) into account,
we define the following operators.
\begin{dfn}
For $p_+>r\geq 1$, $p_->s\geq 1$, and $k\in \mathbb{Z}$, we define
\begin{align}
\label{G+G-}
\begin{aligned}
G^{[r^\vee,r]}_{+,k}:=\lim_{\epsilon\rightarrow 0}\hspace{-1mm}{}^F \frac{1}{\epsilon}N^{[r^\vee,r]}_{+,k}(\epsilon),\qquad\qquad
G^{[s^\vee,s]}_{-,k}:=\lim_{\epsilon\rightarrow 0}\hspace{-1mm}{}^F \frac{1}{\epsilon}N^{[s^\vee,s]}_{-,k}(\epsilon).
\end{aligned}
\end{align}
\end{dfn}
We see that $G^{[r^\vee,r]}_{+,k}$ and $G^{[s^\vee,s]}_{-,k}$ are uniquely determined, regardless of how the limit is taken,
and that 
$G^{[r^\vee,r]}_{+,k}\in {\rm Hom}_{\mathbb{C}}( F_{r,k}, F_{r,k;2})$
and
$G^{[s^\vee,s]}_{-,k}\in {\rm Hom}_{\mathbb{C}}( F_{k,s}, F_{k,s;-2})$.
In what follows, we show that in the case $r=s=1$, these operators (\ref{G+G-}), when restricted to ${\rm ker}Q_\pm$, are independent of $k$. 
The following lemma is immediate from Proposition~\ref{prop:tw} and Definition~\ref{notlim}.
\begin{lem}
Let $p_+>r\geq 1$, $p_->s\geq 1$, and $k\in \mathbb{Z}$. Then, for any $\tilde{u}\in \widetilde{F}_{r,k}$ and $\tilde{v}\in \widetilde{F}_{k,s}$, we have
\begin{align*}
&{Q}^{[r]}_+\lim_{\epsilon\rightarrow 0}\hspace{-1mm}{}^F\tilde{u}=\lim_{\epsilon\rightarrow 0}\hspace{-1mm}{}^F\widetilde{Q}^{[r]}_+\tilde{u},
&{Q}^{[s]}_-\lim_{\epsilon\rightarrow 0}\hspace{-1mm}{}^F\tilde{v}=\lim_{\epsilon\rightarrow 0}\hspace{-1mm}{}^F\widetilde{Q}^{[s]}_-\tilde{v}.
\end{align*}
\label{lem:tw0}
\end{lem}
Note that by Proposition~\ref{prop:tw},
\begin{align*}
\left.{\rm im}\Bigl(\frac{1}{\epsilon}{\rm e}^{\widetilde{\alpha}_{r^\vee,k+p_-}-\widetilde{\alpha}_{-r,k}}\widetilde{Q}^{[r]}_+{\rm e}^{\widetilde{\alpha}_{r,k}-\alpha_{r,k}}\right|_{F_{r,k}\cap {\rm ker}Q^{[r]}_+}\Bigr)\subset\widetilde{F}_{r^\vee,k+p_-}.
\end{align*}
Then, from Remark~\ref{Rem:lim} and Lemma~\ref{lem:tw0}, we obtain following lemma.
\begin{lem}
\label{lem:tw}
For $p_+>r\geq 1$, $p_->s\geq 1$, and $k\in \mathbb{Z}$, we have
\begin{align*}
{G}^{[r^\vee,r]}_{+,k}&=Q^{[r^\vee]}_+\lim_{\epsilon\rightarrow 0}\hspace{-1mm}{}^F\frac{1}{\epsilon}\widetilde{Q}^{[r]}_+{\rm e}^{\widetilde{\alpha}_{r,k}-\alpha_{r,k}},\qquad {\rm on}\ F_{r,k}\cap {\rm ker}Q^{[r]}_+,\\
{G}^{[s^\vee,s]}_{-,k}&=Q^{[s^\vee]}_-\lim_{\epsilon\rightarrow 0}\hspace{-1mm}{}^F\frac{1}{\epsilon}\widetilde{Q}^{[s]}_-{\rm e}^{\widetilde{\alpha}_{k,s}-\alpha_{k,s}},\qquad {\rm on}\ F_{k,s}\cap {\rm ker}Q^{[s]}_-.
\end{align*}
\end{lem}
Consider the limit
\begin{align*}
\lim_{\epsilon\rightarrow 0}\hspace{-1mm}{}^F\frac{1}{\epsilon}\widetilde{Q}_+{\rm e}^{\widetilde{\alpha}_{1,k}-\alpha_{1,k}}=\lim_{\epsilon\rightarrow 0}\hspace{-1mm}{}^F\frac{1}{\epsilon}\oint_{z=0}\widetilde{Q}_+(z)\frac{{\rm d}z}{2\pi i}{\rm e}^{\widetilde{\alpha}_{1,k}-\alpha_{1,k}},
\end{align*}
which is well-defined on $ F_{1,k}\cap {\rm ker}Q_+$.
Recall that the field $\widetilde{Q}_+(z)$ is defined by
\begin{align*}
\widetilde{Q}_+(z)={\rm e}^{\widetilde{\alpha}_+}z^{\widetilde{\alpha}_+a_0}:\overline{Y}(\ket{\alpha_+},z)\prod_{n\geq 1}{\rm exp}\Bigl(\epsilon\frac{a_{-n}}{n}z^{n}\Bigr)\prod_{n\geq 1}{\rm exp}\Bigl(-\epsilon\frac{a_n}{n}z^{-n}\Bigr):
\end{align*}
(see (\ref{Y-ver})) and the factor $z^{\widetilde{\alpha}_+a_0}$ acts on $\widetilde{F}_{1,k}$ as $z^{k-1}$.
Applying this expression of $\widetilde{Q}_+(z)$ to the above limit, we obtain the following:
\begin{align}
\begin{aligned}
\left.\lim_{\epsilon\rightarrow 0}\hspace{-1mm}{}^F\frac{1}{\epsilon}\widetilde{Q}_+{\rm e}^{\widetilde{\alpha}_{1,k}-\alpha_{1,k}}\right|_{{\rm ker}Q_+}=\oint_{z=0}:{Q}_+(z)\phi_0(z):\frac{{\rm d}z}{2\pi i},
\end{aligned}
\label{wi-1}
\end{align}
where we denote
\begin{align*}
\phi_0(z):=\sum_{n>0}\frac{a_{-n}}{n}z^{n}-\sum_{n>0}\frac{a_n}{n}z^{-n}.
\end{align*}
Note that the field $\phi_0(z)$ agrees with the sum part of the scalar field (\ref{eq:scalar}).
By Lemma~\ref{lem:tw} and (\ref{wi-1}), we obtain
\begin{align}
\label{wi-G+}
\left.{G}^{[p_+-1,1]}_{+,k}\right|_{{\rm ker}Q_+}=Q^{[p_+-1]}_+\oint_{z=0}:{Q}_+(z)\phi_0(z):\frac{{\rm d}z}{2\pi i}.
\end{align}
By a similar argument, we have
\begin{align}
\label{wi-G-}
\left.{G}^{[p_--1,1]}_{-,k}\right|_{{\rm ker}Q_-}=\frac{p_+}{p_-}Q^{[p_--1]}_-\oint_{z=0}:{Q}_-(z)\phi_0(z):\frac{{\rm d}z}{2\pi i},
\end{align}
where the fraction $p_+/p_-$ comes from the expansion
\begin{align*}
\widetilde{\alpha}_-(\epsilon)=\alpha_-+\frac{p_+}{p_-}\epsilon+o(\epsilon).
\end{align*}
Note that the right-hand sides of (\ref{wi-G+}) and (\ref{wi-G-}) are independent of $k$ and well defined on $F_{1,k}$ and $F_{k,1}$ for all $k \in \mathbb{Z}$, respectively.
Thus, we obtain the following theorem.
\begin{thm}
\label{G+G-prop}
Define 
\begin{align*}
&\mathbb{G}^{[p_+]}_+\in {\rm End}_{\mathbb{C}}\Bigl(\bigoplus_{\theta=\pm}\bigoplus_{s=1}^{p_-}\mathcal{V}^{\theta}_{1,s}\Bigr),
&\mathbb{G}^{[p_-]}_-\in {\rm End}_{\mathbb{C}}\Bigl(\bigoplus_{\theta=\pm}\bigoplus_{r=1}^{p_+}\mathcal{V}^{\theta}_{r,1}\Bigr)
\end{align*}
as the right-hand side of (\ref{wi-G+}) and (\ref{wi-G-}), respectively. 
Then we have
\begin{align*}
&\left.\mathbb{G}^{[p_+]}_+\right|_{F_{1,s;n}\cap {\rm ker}Q_+}=\left.G^{[p_+-1,1]}_{+,s+np_-}\right|_{{\rm ker}Q_+},
&\left.\mathbb{G}^{[p_-]}_-\right|_{F_{r,1;n}\cap {\rm ker}Q_-}=\left.G^{[p_--1,1]}_{-,r-np_-}\right|_{{\rm ker}Q_-}.
\end{align*}
\end{thm}

\begin{Remark}
\label{rem:Lim}
In the case of $\mathcal{W}_{2,p_-}$, a derivation of ${\rm End}_{\mathfrak{Vir}}(\mathcal{W}_{2,p_-})$ with the explicit form 
\begin{align}
\label{deri-AM}
G=\sum_{i=1}^{\infty}\frac{1}{i}Q_+[-i]Q_+[i]
\end{align}
was introduced in \cite{AdamovicD/MilasA:2010} (see also \cite{AdamovicD/MilasA:2014}), where $Q_+[n]$ is the $n$-th mode of $Q_+(z)$. Note that the composition $Q_+[0]Q_+[0]=Q_+Q_+$ vanishes due to Proposition \ref{Felder complex2}.
From the relation $Q_+[0]Q_+[0]=0$, the above (\ref{deri-AM}) can be expressed as follows:
\begin{align}
\label{deri-AM2}
G=-{\rm Res}_{x_1}{\rm Res}_{x_2}\log\Bigl(1-\frac{x_2}{x_1}\Bigr)Q_+(x_1)Q_+(x_2).
\end{align}
From the right-hand side of this equation, the derivation (\ref{deri-AM}) can be regarded as the relation $Q_+[0]Q_+[0]=0$ twisted by $\log (1-x_2/x_1)$.
Using the result on the automorphism group in \cite{McRaeR/SopinV:2026} together with the result proved in the next subsection, one can see that $G|_{\mathcal{W}_{2,p_-}}$ and $\mathbb{G}^{[2]}_+|_{\mathcal{W}_{2,p_-}}$ coincide up to a scalar multiple. 
The equivalence of these two linear operators can also be seen from the fact that the operator product expansion of the fields $Q_+(z')$ and $\phi_0(z)$ is 
\begin{align*}
Q_+(z')\phi_0(z)=\alpha_+\log \Bigl(1-\frac{z}{z'}\Bigr)Q_+(z')+:Q_+(z')\phi_0(z):.
\end{align*}
We also expect that, on $\mathcal{W}_{p_+,p_-}$, $\mathbb{G}^{[p_\pm]}_\pm$ coincides, up to scalar multiples, with the following operators
\begin{align*}
\begin{aligned}
&\oint_{z'=0}{\rm d}z'\int_{ [\Delta^{(\alpha_\pm)}_{p_\pm-2}]}{\rm d}y_1\cdots {\rm d}y_{p_\pm-2}\oint_{z=0}{{\rm d}z}\\
&\ \cdot (z')^{p_\pm-2}\log\Bigl(\bigl(1-\frac{z}{z'}\bigr)\prod_{i=1}^{p_\pm-2}\bigl(1-\frac{z}{z'y_i}\bigr)\Bigr)
{Q}_\pm(z'){Q}_\pm(z'y_1)\cdots {Q}_\pm(z'y_{p_\pm-2}) {Q}_\pm(z).
\end{aligned}
\end{align*}
\end{Remark}

\begin{Remark}
\begin{enumerate}
\item From the discussion above Theorem~\ref{G+G-prop}, it is natural to expect that $G^{[r^\vee,r]}_{+,k}$ and $G^{[s^\vee,s]}_{-,k}$ are independent of $k$ also for $r,s\geq 2$. This expectation is further supported by properties such as those described in Theorems~\ref{G-hom} and~\ref{non-Gop} below.
\item In \cite[Subsection 4.5]{TW}, a derivation $E$, called a Frobenius homomorphism, is constructed as a certain limit of the integral
\begin{align*}
\begin{aligned}
&\int_{[\Delta^{\infty}_{p_+-1}]}{\rm d}z'_1\cdots {\rm d}z'_{p_+-1}\oint_{z=0}{\rm d}z\widetilde{Q}_+(z'_1)\cdots \widetilde{Q}_+(z'_{p_+-1})\widetilde{Q}_+(z),
\end{aligned}
\end{align*}
where $[\Delta^{\infty}_{p_+-1}]$ is a regularization of $\Delta^\infty_{p_+-r}=\{\infty >z'_1>z'_2>\cdots >z'_{p_+-1}>1\}$ by taking into account the multivaluedness of 
\begin{align*}
\mathfrak{G}_{p_+-1}\bigl(\bm{z}';\frac{\widetilde{\alpha}_+(\epsilon)^2}{2}\bigr):=\prod_{1\leq i\neq j\leq p_+-1}
(z'_i-z'_j)^{\frac{\widetilde{\alpha}_+(\epsilon)^2}{2}}
\prod_{i=1}^{p_+-1}(z'_i)^{\widetilde{\alpha}_+(\epsilon)^2}.
\end{align*}
At this point, we note that $[\Delta^{\infty}_{p_+-1}]$ does not strictly make sense as a locally finite twisted cycle arising from the multivaluedness of $\mathfrak{G}_{p_+-1}$.
In fact, $\mathfrak{G}_{p_+-1}$ takes a well-defined value on
\begin{align*}
\{(z'_1,z'_2,\dots,z'_{p_+-1})\in \mathbb{C}^{p_+-1}\ |\ z'_{p_+-1}=1,\ z'_i\neq z'_j,\ i\neq j,\ z'_i\neq 0\}
\end{align*}
and does not vanish there.
For example, in the case of $p_+=2$, the point $z'_1=1$ is a boundary of $\Delta^\infty_1$ and $\mathfrak{G}_{p_+-1}$ takes the constant value $1$ on this point.
Consequently, it becomes necessary to re-examine some properties of the Frobenius homomorphisms in \cite{TW} by using our operators $\mathbb{G}^{[p_\pm]}_\pm$.
\end{enumerate}
\label{non-well-tw}
\end{Remark}

\subsection{Properties of the operators $\mathbb{G}^{[p_\pm]}_\pm$}
It is known that the derivation (\ref{deri-AM2}), when restricted to ${\rm ker}Q_+$, is a Virasoro homomorphism \cite[Theorem 3.1]{AdamovicD/MilasA:2010}. This property likewise holds for ${G}^{[r^\vee,r]}_{+,k}$ and ${G}^{[s^\vee,s]}_{-,k}$:
\begin{thm}
\label{G-hom}
The restrictions ${G}^{[r^\vee,r]}_{+,k}|_{{\rm ker}Q^{[r]}_+\cap F_{r,k}}$ and ${G}^{[s^\vee,s]}_{-,k}|_{{\rm ker}Q^{[s]}_-\cap F_{k,s}}$ are Virasoro homomorphisms.
\end{thm}
\begin{proof}
We only prove the case of ${G}^{[r^\vee,r]}_{+,k}$. The case of ${G}^{[s^\vee,s]}_{-,k}$ can be proved in the same way.

We introduce the notation
$
{L}_n:=L^{({\alpha}_0)}_n
$
and
$
\widetilde{L}_n:=L^{(\widetilde{\alpha}_0(\epsilon))}_n,
$
for $n\in \mathbb{Z}$.
For $k\in \mathbb{Z}$, take an arbitrary vector $u\in F_{r,k}\cap {\rm ker}Q^{[r]}_+$.
From Lemma~\ref{lem:tw}, we have
\begin{align*}
\begin{aligned}
{G}^{[r^\vee,r]}_{+,k}L_nu
&={Q}^{[r^\vee]}_+\lim_{\epsilon\rightarrow 0}\hspace{-1mm}{}^F \frac{1}{\epsilon}\widetilde{Q}^{[r]}_+\circ{\rm e}^{\widetilde{\alpha}_{r,k}-\alpha_{r,k}}L_nu.
\end{aligned}
\end{align*}
Note that, for any $n\in \mathbb{Z}$, 
\begin{align*}
{\rm e}^{\widetilde{\alpha}_{r,k}-\alpha_{r,k}}L_nu-\widetilde{L}_n{\rm e}^{\widetilde{\alpha}_{r,k}-\alpha_{r,k}}u\in \epsilon \cdot \widetilde{F}_{r,k}.
\end{align*}
Then, by Lemma~\ref{lem:tw0} and the relation $Q^{[r^\vee]}_+\circ Q^{[r]}_+=0$, we have
\begin{align}
\begin{aligned}
{G}^{[r^\vee,r]}_{+,k}L_nu
&={Q}^{[r^\vee]}_+\lim_{\epsilon\rightarrow 0}\hspace{-1mm}{}^F \frac{1}{\epsilon}\widetilde{Q}^{[r]}_+\widetilde{L}_n\circ{\rm e}^{\widetilde{\alpha}_{r,k}-\alpha_{r,k}}u\\
&\qquad \qquad+{Q}^{[r^\vee]}_+\lim_{\epsilon\rightarrow 0}\hspace{-1mm}{}^F \frac{1}{\epsilon}\widetilde{Q}^{[r]}_+\bigl({\rm e}^{\widetilde{\alpha}_{r,k}-\alpha_{r,k}}L_nu-\widetilde{L}_n{\rm e}^{\widetilde{\alpha}_{r,k}-\alpha_{r,k}}\bigr)u\\
&={Q}^{[r^\vee]}_+\lim_{\epsilon\rightarrow 0}\hspace{-1mm}{}^F \frac{1}{\epsilon}\widetilde{Q}^{[r]}_+\widetilde{L}_n\circ{\rm e}^{\widetilde{\alpha}_{r,k}-\alpha_{r,k}}u\\
&\qquad \qquad+{Q}^{[r^\vee]}_+\circ {Q}^{[r]}_+\lim_{\epsilon\rightarrow 0}\hspace{-1mm}{}^F \frac{1}{\epsilon}\bigl({\rm e}^{\widetilde{\alpha}_{r,k}-\alpha_{r,k}}L_nu-\widetilde{L}_n{\rm e}^{\widetilde{\alpha}_{r,k}-\alpha_{r,k}}\bigr)u\\
&={Q}^{[r^\vee]}_+\lim_{\epsilon\rightarrow 0}\hspace{-1mm}{}^F \widetilde{L}_n\frac{1}{\epsilon}\widetilde{Q}^{[r]}_+\circ{\rm e}^{\widetilde{\alpha}_{r,k}-\alpha_{r,k}}u\\
&=L_n{Q}^{[r^\vee]}_+\lim_{\epsilon\rightarrow 0}\hspace{-1mm}{}^F \frac{1}{\epsilon}\widetilde{Q}^{[r]}_+\circ{\rm e}^{\widetilde{\alpha}_{r,k}-\alpha_{r,k}}u
\end{aligned}
\label{eq:imc-v}
\end{align}
where we used the commutativity $[\widetilde{L}_n,\widetilde{Q}^{[r]}_+]=0$ in the third equality and, for the fourth equality, we used the fact that the zero mode $a_0$ commutes with the limit (for the limit, see Definition~\ref{notlim}).
Therefore, from Lemma~\ref{lem:tw} and (\ref{eq:imc-v}), we obtain $[{G}^{[r^\vee,r]}_{+,k},L_n]=0$ on $F_{r,k}\cap {\rm ker}Q^{[r]}_+$. 
\end{proof}

\begin{thm}
\begin{enumerate}
\item Let $p_+\geq 2$, $p_+>r\geq 1$, $p_-\geq s\geq 1$, and $n\in \mathbb{Z}$. Then $G^{[r^\vee,r]}_{+,s+np_-}$ maps each singular vector in
\begin{align*}
\bigoplus_{l\in \mathbb{Z}_{\geq 0}}F_{r,s;n}[h+l]\cap {\rm ker}Q^{[r]}_+\subset K_{r,s;n;+},
\end{align*}
to a nonzero singular vector in $K_{r,s;n+2;+}$, where we denote
\begin{equation*}
h=
\begin{cases}
h_{r,s;n+2},& {\rm if}\ h_{r,s;n+2}-h_{r,s;n}\geq 0,\\
h_{r,s;n},&{\rm if}\ h_{r,s;n+2}-h_{r,s;n}<0
\end{cases}
.
\end{equation*}
 \item Let $p_->s\geq 1$, $p_+\geq r\geq 1$, and $n\in \mathbb{Z}$. Then $G^{[s^\vee,s]}_{-,r-np_-}$ maps each singular vector in
\begin{align*}
\bigoplus_{l\in \mathbb{Z}_{\geq 0}}F_{r,s;n}[h+l]\cap {\rm ker}Q^{[s]}_-\subset K_{r,s;n;-},
\end{align*}
to a nonzero singular vector in $K_{r,s;n-2;-}$, where we denote
\begin{equation*}
h=
\begin{cases}
h_{r,s;n-2},&{\rm if}\  h_{r,s;n-2}-h_{r,s;n}\geq 0,\\
h_{r,s;n},&{\rm if}\ h_{r,s;n-2}-h_{r,s;n}<0
\end{cases}
.
\end{equation*}
\end{enumerate}
\label{non-Gop}
\end{thm}
In particular, from this theorem, we see that the linear operators $G^{[r^\vee,r]}_{+,k}$ and $G^{[s^\vee,s]}_{-,k}$ are nontrivial.
Before giving the proof, we introduce some notation and a lemma.
Let $\gamma(\epsilon) \in \mathcal{O}_{0;\epsilon}$.
We regard $\widetilde{F}^{(\widetilde{\alpha}_0)}_{\gamma(\epsilon)}$ as a $\mathcal{O}_{0;\epsilon}\otimes\pi^{(\widetilde{\alpha}_0)}_{\gamma(\epsilon)}(U(\mathfrak{Vir}))$-module in the natural way.
Let ${}_\mathcal{O}M(h^{(\widetilde{\alpha}_0)}_{\gamma(\epsilon)},c_{\widetilde{\alpha}_0})$ be the $\mathcal{O}_{0;\epsilon}\otimes U(\mathfrak{Vir})$-Verma module whose lowest $\mathcal{L}_0$-weight and central charge are given by $(h^{(\widetilde{\alpha}_0)}_{\gamma(\epsilon)},c_{\widetilde{\alpha}_0})$.
By the universal property of ${}_\mathcal{O}M(h^{(\widetilde{\alpha}_0)}_{\gamma(\epsilon)},c_{\widetilde{\alpha}_0})$, there exists the $\mathcal{O}_{0;\epsilon}\otimes U(\mathfrak{Vir})$-module map
\begin{align*}
\begin{aligned}
\widetilde{\varGamma}^{(\widetilde{\alpha}_0)}_{\gamma(\epsilon)}:{}_\mathcal{O}M(h^{(\widetilde{\alpha}_0)}_{\gamma(\epsilon)},c_{\widetilde{\alpha}_0})&\rightarrow \widetilde{F}^{(\widetilde{\alpha}_0)}_{\gamma(\epsilon)}\\
\ket{h^{(\widetilde{\alpha}_0)}_{\gamma(\epsilon)}}&\mapsto \ket{\gamma(\epsilon)}.
\end{aligned}
\end{align*}
Let $\mathcal{K}_{0;\epsilon}$ denote the quotient field of $\mathcal{O}_{0;\epsilon}$.
The following lemma follows from the classification results in \cite{FF,IK} (see also \cite[Proposition 4.2]{TW}).
\begin{lem}
The canonical $\mathcal{K}_{0;\epsilon}\otimes U(\mathfrak{Vir})$-module map 
\begin{align*}
\mathcal{K}_{0;\epsilon}\otimes \widetilde{\varGamma}^{(\widetilde{\alpha}_0)}_{\gamma(\epsilon)}: \mathcal{K}_{0;\epsilon}\otimes{}_\mathcal{O}M(h^{(\widetilde{\alpha}_0)}_{\gamma(\epsilon)},c_{\widetilde{\alpha}_0})\rightarrow \mathcal{K}_{0;\epsilon}\otimes\widetilde{F}^{(\widetilde{\alpha}_0)}_{\gamma(\epsilon)}
\end{align*}
is not an isomorphism if and only if $\gamma(\epsilon)=\widetilde{\alpha}_{r,s}(\epsilon)$ for some $r,s\in \mathbb{Z}_{\geq 1}$.
\label{lem-imageQ}
\end{lem}
In particular, from this lemma, the Virasoro modules
\begin{align*}
&\pi^{(\widetilde{\alpha}_0)}(U(\mathfrak{Vir}_-))\widetilde{Q}^{[r]}_+\ket{\widetilde{\alpha}_{r,-k}},
&\pi^{(\widetilde{\alpha}_0)}(U(\mathfrak{Vir}_-))\widetilde{Q}^{[s]}_-\ket{\widetilde{\alpha}_{-k,s}},\ \ r,s{\geq 1},\ k\in \mathbb{Z} 
\end{align*}
have no Virasoro null vectors, where we used the abbreviations $\pi^{(\rho)} =\pi^{(\rho)}_{\bullet}$.
For $u\in \widetilde{F}^{(\widetilde{\alpha}_0)}_{\gamma(\epsilon)}$, write
\begin{align*}
u=h_0(\epsilon)\ket{\gamma(\epsilon)}+\sum_{n\in \mathbb{Z}_{>0}}\sum_{\lambda\vdash n} h_{\lambda}(\epsilon)a_{-\lambda}\ket{\gamma(\epsilon)},
\end{align*}
where $h_0(\epsilon),h_\lambda(\epsilon)\in \mathcal{O}_{0;\epsilon}$.
We define $\partial^F_{\epsilon} u$ by
\begin{align*}
\partial^F_{\epsilon} u =h'_0(\epsilon)\ket{\gamma(\epsilon)}+\sum_{n\in \mathbb{Z}_{>0}}\sum_{\lambda\vdash n}h'_{\lambda}(\epsilon)a_{-\lambda}\ket{\gamma(\epsilon)}.
\end{align*}
Following \cite[Subsection 8.2.3]{IK}, we introduce the notation
\begin{align*}
IK(k-1\rbrack(F^{(\alpha_0)}_{\gamma(0)})&:={\rm ker}\Bigl[F^{(\alpha_0)}_{\gamma(0)}\twoheadrightarrow {\rm im}\Bigl(\bigl(\left.(\partial^F_{\epsilon})^k\circ \widetilde{\varGamma}^{(\widetilde{\alpha}_0)}_{\gamma(\epsilon)}\bigr)\right|_{\epsilon=0}\Bigr)\Bigr].
\end{align*}

\begin{proof}[Proof of Theorem~\ref{non-Gop}]
We prove the cases of $G^{[s^\vee,s]}_{-,r-np_+}$ $(n\in \mathbb{Z}, 1\leq r\leq p_+)$. The other cases can be proved in the same way.
First, we show the case $r=p_+$. In this case, the domain of $G^{[s^\vee,s]}_{-,(1-n)p_+}$ is the chain type $F_{p_+,s;n}$.
By Proposition \ref{Felder complex}, we have
\begin{align*}
F_{p_+,s;n}\cap {\rm ker}Q^{[s]}_-=X_{p_+,s;n}={\rm Soc}(F_{p_+,s;n}),
\end{align*}
and by Proposition \ref{FockSocle}, 
\begin{align*}
&{\rm Soc}_1(F_{p_+,s;n})={\rm Soc}(F_{p_+,s;n})=\bigoplus_{k\geq 0}L(h_{p_+,s^\vee;|n|+2k+1}),\\
&{\rm Soc}_2(F_{p_+,s;n})/{\rm Soc}_1(F_{p_+,s;n})=\bigoplus_{k\geq a}L(h_{p_+,s;|n|+2k})
\end{align*}
where $a=0$ if $n\geq 1$ and $a=1$ if $n<1$.
Let $n\geq 1$. 
Note that in this case, the vector $Q^{[s]}_-\ket{\alpha_{p_+,s;n}}$ gives a nonzero singular vector (see Figures \ref{51fig} and \ref{51fig-fel}).
For $k\geq 0$, let $u_{n,k}$ denote the singular vector of $F_{p_+,s;n}$ with $\mathcal{L}_0$-weight $h_{p_+,s^\vee;n+2k+1}$, which corresponds to the lowest weight vector of the Virasoro submodule $L(h_{p_+,s^\vee;n+2k+1})\subset {\rm Soc}(F_{p_+,s;n})$ (see Figure~\ref{51fig-lem-sing}).
Note that, from Proposition~\ref{prop:tw},
\begin{align}
\label{eq:04}
\widetilde{Q}^{[s]}_-{\rm e}^{\widetilde{\alpha}_{(1-n)p_+,s}-\alpha_{p_+,s;n}}u_{n,k}\in \epsilon\cdot \widetilde{F}_{(1-n)p_+,-s}.
\end{align}
From Proposition~\ref{Felder complex2} and Theorem~\ref{G-hom}, to prove that $G^{[s^\vee,s]}_{-,(1-n)p_+}u_{n,k}$ is a nonzero singular vector of $F_{p_+,s;n-2}$, it suffices to show the nonvanishingness of $G^{[s^\vee,s]}_{-,(1-n)p_+}u_{n,k}$.
By the Jantzen filtration of $F_{p_+,s;n}$ \cite[Sections 8.2-8.3]{IK}, 
\begin{align*}
u_{n,k}\in IK(k\rbrack(F_{p_+,s;n})\setminus IK(k-1\rbrack(F_{p_+,s;n}).
\end{align*}
From this, we have
\begin{align}
\label{ep-k-J}
\begin{aligned}
\epsilon^k \cdot {\rm e}^{\widetilde{\alpha}_{(1-n)p_+,s}-\alpha_{p_+,s;n}}u_{n,k}&\in \bigl(\mathcal{O}_{0;\epsilon}\otimes \pi^{(\widetilde{\alpha}_0)}(U(\mathfrak{Vir}))\bigr)\ket{\widetilde{\alpha}_{(1-n)p_+,s}(\epsilon)},\\
\epsilon^{k} \cdot {\rm e}^{\widetilde{\alpha}_{(1-n)p_+,s}-\alpha_{p_+,s;n}}u_{n,k}&\notin \bigl(\epsilon\cdot\mathcal{O}_{0;\epsilon}\otimes \pi^{(\widetilde{\alpha}_0)}(U(\mathfrak{Vir}))\bigr)\ket{\widetilde{\alpha}_{(1-n)p_+,s}(\epsilon)}.
\end{aligned}
\end{align}
In particular, by Lemma \ref{lem-imageQ}, the vector
\begin{align}
\label{ep-k-J-2}
\epsilon^k \cdot\widetilde{Q}^{[s]}_-{\rm e}^{\widetilde{\alpha}_{(1-n)p_+,s}-\alpha_{p_+,s;n}}u_{n,k}
\in \bigl(\mathcal{O}_{0;\epsilon}\otimes \pi^{(\widetilde{\alpha}_0)}(U(\mathfrak{Vir}))\bigr)\widetilde{Q}^{[s]}_-\ket{\widetilde{\alpha}_{(1-n)p_+,s}(\epsilon)}
\end{align}
is nonzero. Here, we used the commutativity of $\widetilde{Q}^{[s]}_-$ with the action of $\pi^{(\widetilde{\alpha}_0)}(U(\mathfrak{Vir}))$.
Note that ${Q}^{[s]}_-\ket{{\alpha}_{p_+,s;n}}\in \pi^{(\alpha_0)}(U(\mathfrak{Vir}))\ket{{\alpha}_{p_+,s^\vee;n-1}}$ (see Figure \ref{51fig-fel}).
Then, by Proposition~\ref{prop:tw} and the Jantzen filtration of $F_{p_+,s^\vee;n-1}$, we have
\begin{align}
\label{ep-k-J-4}
\begin{aligned}
\widetilde{Q}^{[s]}_-\ket{\widetilde{\alpha}_{(1-n)p_+,s}(\epsilon)}&\in \bigl(\mathcal{O}_{0;\epsilon}\otimes \pi^{(\widetilde{\alpha}_0)}(U(\mathfrak{Vir}))\bigr)\ket{\widetilde{\alpha}_{(1-n)p_+,-s}(\epsilon)},\\
\widetilde{Q}^{[s]}_-\ket{\widetilde{\alpha}_{(1-n)p_+,s}(\epsilon)}&\notin \bigl(\epsilon\cdot \mathcal{O}_{0;\epsilon}\otimes \pi^{(\widetilde{\alpha}_0)}(U(\mathfrak{Vir}))\bigr)\ket{\widetilde{\alpha}_{(1-n)p_+,-s}(\epsilon)}.\\
\end{aligned}
\end{align}
By (\ref{ep-k-J})-(\ref{ep-k-J-4}), we obtain
\begin{align}
\label{ep-k-J-3}
\begin{aligned}
\epsilon^k \cdot \widetilde{Q}^{[s]}_-{\rm e}^{\widetilde{\alpha}_{(1-n)p_+,s}-\alpha_{p_+,s;n}}u_{n,k}&\in \bigl(\mathcal{O}_{0;\epsilon}\otimes \pi^{(\widetilde{\alpha}_0)}(U(\mathfrak{Vir}))\bigr)\ket{\widetilde{\alpha}_{(1-n)p_+,-s}(\epsilon)},\\
\epsilon^{k} \cdot \widetilde{Q}^{[s]}_-{\rm e}^{\widetilde{\alpha}_{(1-n)p_+,s}-\alpha_{p_+,s;n}}u_{n,k}&\notin \bigl(\epsilon\cdot\mathcal{O}_{0;\epsilon}\otimes \pi^{(\widetilde{\alpha}_0)}(U(\mathfrak{Vir}))\bigr)\ket{\widetilde{\alpha}_{(1-n)p_+,-s}(\epsilon)}.
\end{aligned}
\end{align}
For $l\in \mathbb{Z}_{\geq 1}$, let $v_{n-1,l}$ denote a cosingular vector of $F_{p_+,s^\vee;n-1}$ with the $\mathcal{L}_0$-weight $h_{p_+,s^\vee;n+2l-1}$ (see Figure~\ref{51fig-lem-sing}), and let $v_{n-1,0}$ denote the lowest weight vector of $F_{p_+,s^\vee;n-1}$.
Here, by a cosingular vector we mean a nonzero vector that is annihilated by all positive modes of the Virasoro algebra in the quotient $F_{p_+,s^\vee;n-1}/{\rm Soc}$.
Then, by (\ref{eq:04}), (\ref{ep-k-J-3}), and the Jantzen filtration of $F_{p_+,s^\vee;n-1}$ \cite[Sections 8.2-8.3]{IK}, we have
\begin{align*}
\lim_{\epsilon\rightarrow 0}\hspace{-1mm}{}^F\epsilon^{-1}\widetilde{Q}^{[s]}_-{\rm e}^{\widetilde{\alpha}_{(1-n)p_+,s}-\alpha_{p_+,s;n}}u_{n,k}\in IK(k+1\rbrack(F_{p_+,s^\vee;n-1})\setminus IK(k\rbrack(F_{p_+,s^\vee;n-1}).
\end{align*}
From this, we have
\begin{align*}
\lim_{\epsilon\rightarrow 0}\hspace{-1mm}{}^F\epsilon^{-1}\widetilde{Q}^{[s]}_-{\rm e}^{\widetilde{\alpha}_{(1-n)p_+,s}-\alpha_{p_+,s;n}}u_{n,k}\in \mathbb{C}^\times v_{n-1,k+1}+\sum_{l=0}^{k}\pi^{(\alpha_0)}(U(\mathfrak{Vir}))v_{n-1,l}
\end{align*}
(for the definition of the above limit, see Definition~\ref{notlim}).
Thus by Proposition~\ref{Felder complex2} and Lemma~\ref{lem:tw}, $G^{[s^\vee,s]}_{-,(1-n)p_+}u_{n,k}$ does not vanish and hence, is a nonzero singular vector from Theorem~\ref{G-hom}. Thus the desired claim follows for $n\geq 1$.
The case of $n\leq 0$ can be proved in the same way by noting Proposition \ref{dual-scop}.

Next, we show the case $1\leq r<p_+$.
In this case, the domain of $G^{[s^\vee,s]}_{-,r-np_+}$ is the braided type $F_{r,s;n}$.
Let $w$ be a subsingular vector in ${\rm ker}Q^{[s]}_-\cap F_{r,s;n}$.
Here, by a subsingular vector we mean a nonzero vector that is annihilated by all positive modes of the Virasoro algebra in the quotient ${\rm Soc}_2(F_{r,s;n})/{\rm Soc}_1$.
Assume that the $\mathcal{L}_0$-weight of $w$ is greater than or equal to the lowest of $F_{r,s;n-2}$.
Let $w'$ be the subsingular vector in $F_{r,s;n-2}$ with the same $\mathcal{L}_0$-weight as $w$.
By Lemma~\ref{lem:tw}, $G^{[s^\vee,s]}_{-,r-np_+}w$ is contained in a $U(\mathfrak{Vir})$-module generated by some subsingular vectors of ${\rm ker}Q^{[s]}_-\cap F_{r,s;n-2}$.
Then, by an argument analogous to the case $r=p_+$, we see that the image $G^{[s^\vee,s]}_{-,r-np_+}w$ contains a nonzero component in $\mathbb{C}^\times w'$.
Together with Proposition~\ref{Felder complex} and Theorem~\ref{G-hom}, this shows that the assertion also holds for $1\leq r<p_+$.
\end{proof}

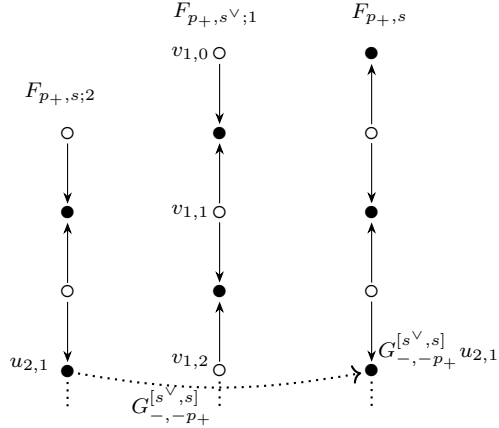
\begin{figure}[htbp]
  \centering
\begin{tikzpicture}[scale=1.05]
\node[inner sep=0.7pt] (a1) at (-0.5, -0.75) {$\bullet$};
\node[inner sep=0.7pt] (a2) at (-0.5, -0.75-1) {$\circ$};
\node[inner sep=0.7pt] (a3) at (-0.5, -0.75-2) {$\bullet$};
\node[inner sep=0.7pt] (a4) at (-0.5, -0.75-3) {$\circ$};
\node[inner sep=0.7pt] (a5) at (-0.5, -0.75-3-1) {$\bullet$};
\node[inner sep=0.7pt] (c) at (-0.5, -0.75-3.1-1) {};
\node[inner sep=0.7pt] (c1) at (-0.5, -0.75-3-0.5-1) {};

\draw[arrows = {-Stealth[scale=0.9]}] (a2) to (a1);
\draw[arrows = {-Stealth[scale=0.9]}] (a2) to (a3);
\draw[arrows = {-Stealth[scale=0.9]}] (a4) to (a3);
\draw[arrows = {-Stealth[scale=0.9]}] (a4) to (a5);
\draw[ dotted,thick] (c) to (c1);

\node[inner sep=0.7pt,font=\scriptsize] (t1) at (-0.5+0.1, -0.75+0.5) {${F_{p_+,s}}$};
\node[inner sep=0.7pt,font=\scriptsize] (t2) at (-0.95+0.1-0.6-1, -0.75+0.5) {${F_{p_+,s^\vee;1}}$};
\node[inner sep=0.7pt,font=\scriptsize] (t1) at (-0.5+0.1-2-2, -0.75+0.5-1) {${F_{p_+,s;2}}$};


\node[inner sep=0.7pt,shift={(-2,0)}] (ak1) at (-0.5, -0.75) {$\circ$};
\node[inner sep=0.7pt,shift={(-2,0)}] (ak2) at (-0.5, -0.75-1) {$\bullet$};
\node[inner sep=0.7pt,shift={(-2,0)}] (ak3) at (-0.5, -0.75-2) {$\circ$};
\node[inner sep=0.7pt,shift={(-2,0)}] (ak4) at (-0.5, -0.75-3) {$\bullet$};
\node[inner sep=0.7pt,shift={(-2,0)}] (ak5) at (-0.5, -0.75-3-1) {$\circ$};
\node[inner sep=0.7pt,shift={(-2,0)}] (ck) at (-0.5, -0.75-3.1-1) {};
\node[inner sep=0.7pt,shift={(-2,0)}] (ck1) at (-0.5, -0.75-3-0.5-1) {};

\draw[arrows = {-Stealth[scale=0.9]}] (ak1) to (ak2);
\draw[arrows = {-Stealth[scale=0.9]}] (ak3) to (ak2);
\draw[arrows = {-Stealth[scale=0.9]}] (ak3) to (ak4);
\draw[arrows = {-Stealth[scale=0.9]}] (ak5) to (ak4);
\draw[ dotted,thick] (ck) to (ck1);

\node[inner sep=0.7pt,shift={(-4,-1.055)}] (aak1) at (-0.5, -0.75) {$\circ$};
\node[inner sep=0.7pt,shift={(-4,-1.055)}] (aak2) at (-0.5, -0.75-1) {$\bullet$};
\node[inner sep=0.7pt,shift={(-4,-1.055)}] (aak3) at (-0.5, -0.75-2) {$\circ$};
\node[inner sep=0.7pt,shift={(-4,-1.055)}] (aak4) at (-0.5, -0.75-3) {$\bullet$};
\node[inner sep=0.7pt,shift={(-4,-1.055)}] (cck) at (-0.5, -0.75-3.1) {};
\node[inner sep=0.7pt,shift={(-4,-1.055)}] (cck1) at (-0.5, -0.75-3-0.5) {};

\draw[arrows = {-Stealth[scale=0.9]}] (aak1) to (aak2);
\draw[arrows = {-Stealth[scale=0.9]}] (aak3) to (aak2);
\draw[arrows = {-Stealth[scale=0.9]}] (aak3) to (aak4);
\draw[ dotted,thick] (cck) to (cck1);

\draw[->,dotted,thick] (aak4) to[bend right=10]  (a5);
\node[inner sep=0.7pt,font=\scriptsize] (t1) at (-0.5-2.5, -0.75-3-1-0.45) {$G^{[s^\vee,s]}_{-,-p_+}$};
\node[inner sep=0.7pt,font=\scriptsize,above left] (vs) at (-0.5-4, -0.75-3-1.055) {$u_{2,1}$};
\node[inner sep=0.7pt,font=\scriptsize,above right] (vs) at  (-0.5+0.05, -0.75-3-1) {$G^{[s^\vee,s]}_{-,-p_+}u_{2,1}$};
\node[inner sep=0.7pt,font=\scriptsize, left] (vs) at  (-0.5-2, -0.75) {$v_{1,0}$};
\node[inner sep=0.7pt,font=\scriptsize, left] (vs) at  (-0.5-2, -0.75-2) {$v_{1,1}$};
\node[inner sep=0.7pt,font=\scriptsize, above left] (vs) at  (-0.5-2, -0.75-4) {$v_{1,2}$};
\end{tikzpicture}
\caption{\label{51fig-lem-sing}The singular vector $u_{n,k}$ and the cosingular vectors $v_{n-1,l}$ in the case of $(n,k)=(2,1)$.}
\mbox{}\\
\end{figure}

Recall that a derivation on a vertex (super)algebra $V$ is an even linear map $D:V\rightarrow V$, such that $D(a_nb)=(Da)_nb+a_nDb$ for all $a,b\in V$ and $n\in \mathbb{Z}$ (cf.~\cite{AdamovicD/MilasA:2010}). In terms of vertex operators, this formula is equivalent to the following:
\begin{align*}
DY(a,w)b=Y(Da)b+Y(a,w)Db. 
\end{align*} 
\begin{thm}
\label{G-deri}
For any $v_1,v_2\in \mathcal{W}_{p_+,p_-}$, we have
\begin{align*}
\mathbb{G}^{[p_\pm]}_{\pm}Y(v_1,w)v_2=Y(\mathbb{G}^{[p_\pm]}_{\pm}v_1,w)v_2+Y(v_1,w)\mathbb{G}^{[p_\pm]}_{\pm}v_2.
\end{align*}
That is, $\mathbb{G}^{[p_+]}_{+}$ and $\mathbb{G}^{[p_-]}_{-}$ act on $\mathcal{W}_{p_+,p_-}$ as derivations.
\end{thm}
\begin{proof}
We prove the case of $\mathbb{G}^{[p_+]}_{+}$. The second case can be proved in the same way.
Let the $a_0$-weights of $v_1$ and $v_2$ be $\alpha_{1,1;2m}$ and $\alpha_{1,1;2n}$, respectively.
We denote
\begin{align*}
\widetilde{v}_1&={\rm e}^{\widetilde{\alpha}_{1,1+2mp_-}-{\alpha}_{1,1;2m}}v_1,&
\widetilde{v}_2&={\rm e}^{\widetilde{\alpha}_{1,1+2np_-}-{\alpha}_{1,1;2n}}v_2,\\
\bm{e}_1&={\rm e}^{\widetilde{\alpha}_{p_+-1,1+(2m+2n+1)p_-}-\widetilde{\alpha}_{-1,1+2(m+n)p_-}},&
\bm{e}_2&={\rm e}^{{\alpha}_{1,1;2(m+n+1)}-\widetilde{\alpha}_{1-p_+,1+(2m+2n+1)p_-}}.
\end{align*}
Note that, for any vector $\widetilde{v}^\vee_3\in \widetilde{F}^\vee_{1,1+2(m+n)p_-}$,
\begin{align}
\label{eq:O-v12}
( \widetilde{v}^\vee_3,{\rm e}^{\widetilde{\alpha}_{1,1+2(m+n)p_-}-{\alpha}_{1,1;2(m+n)}}Y(v_1,w)v_2)-( \widetilde{v}^\vee_3,Y(\widetilde{v}_1,w)\widetilde{v}_2)\in \epsilon \cdot\mathcal{O}_{0;\epsilon},
\end{align}
with $w$ fixed, where we define the inner product $(\ \cdot\ ,\ \cdot \ )_{\widetilde{F}_{1,1+2(m+n)p_-}}$ on $\widetilde{F}_{1,1+2(m+n)p_-}$ by naturally extending the inner product (\ref{inn-pro}) over $\mathbb{C}$.
Fix any $\psi\in {F}^*_{1,1+2(m+n)p_-}$.
Then, we have
\begin{align}
\label{eq:deri-deform}
\begin{aligned}
&\langle \psi,\mathbb{G}^{[p_+]}_{+}Y(v_1,w)v_2\rangle \\
&=\lim_{\epsilon\rightarrow 0}\epsilon^{-1}\langle \psi,\bm{e}_2\widetilde{Q}^{[p_+-1]}_+\bm{e}_1\widetilde{Q}_+{\rm e}^{\widetilde{\alpha}_{1,1+2(m+n)p_-}-{\alpha}_{1,1;2(m+n)}}Y(v_1,w)v_2\rangle\\
&=\lim_{\epsilon\rightarrow 0}\epsilon^{-1}\langle \psi,\bm{e}_2\widetilde{Q}^{[p_+-1]}_+\bm{e}_1\widetilde{Q}_+Y(\widetilde{v}_1,w)\widetilde{v}_2\rangle\\
&=\lim_{\epsilon\rightarrow 0}\epsilon^{-1}\Bigl(\langle \psi,\bm{e}_2\widetilde{Q}^{[p_+-1]}_+\bm{e}_1Y(\widetilde{Q}_+\widetilde{v}_1,w)\widetilde{v}_2\rangle+\langle \psi,\bm{e}_2\widetilde{Q}^{[p_+-1]}_+\bm{e}_1Y(\widetilde{v}_1,w)\widetilde{Q}_+\widetilde{v}_2\rangle\Bigr),
\end{aligned}
\end{align}
where in the first equality, we used Theorem~\ref{G+G-prop}, in the second equality, we used Proposition \ref{prop:tw}, (\ref{eq:O-v12}), and the relation $Q^{[r^\vee]}_+\circ Q^{[r]}_+ =0$, and in the last equality, we used the relation
\begin{align*}
\widetilde{Q}_+Y(\widetilde{v}_1,w)\widetilde{v}_2=Y(\widetilde{Q}_+\widetilde{v}_1,w)\widetilde{v}_2+Y(\widetilde{v}_1,w)\widetilde{Q}_+\widetilde{v}_2,
\end{align*}
which follows from the fact that $\widetilde{Q}_+(z)$ is local with $Y(\widetilde{v}_1,w)$ and $Y(\widetilde{v}_2,w)$.
We denote
\begin{align*}
\widehat{Q}_+:=\oint_{z=0}:{Q}_+(z)\phi_0(z):\frac{{\rm d}z}{2\pi i}.
\end{align*}
Then $\mathbb{G}^{[p_+]}_{+}=Q^{[p_+-1]}_+\circ\widehat{Q}_+$, and we have
\begin{align*}
\lim_{\epsilon\rightarrow 0}\hspace{-1mm}{}^F\epsilon^{-1}\widetilde{Q}_+\widetilde{v}_i=\widehat{Q}_+v_i,\qquad i=1,2.
\end{align*}
As in Lemma~\ref{lem:tw}, we can show that
\begin{align*}
&\lim_{\epsilon\rightarrow 0}\epsilon^{-1}\langle \psi,\bm{e}_2\widetilde{Q}^{[p_+-1]}_+\bm{e}_1Y(\widetilde{Q}_+\widetilde{v}_1,w)\widetilde{v}_2\rangle=\langle \psi,{Q}^{[p_+-1]}_+Y(\lim_{\epsilon\rightarrow 0}\hspace{-1mm}{}^F\epsilon^{-1}\widetilde{Q}_+\widetilde{v}_1,w){v}_2\rangle,\\
&\lim_{\epsilon\rightarrow 0}\epsilon^{-1}\langle \psi,\bm{e}_2\widetilde{Q}^{[p_+-1]}_+\bm{e}_1Y(\widetilde{v}_1,w)\widetilde{Q}_+\widetilde{v}_2\rangle=\langle \psi,{Q}^{[p_+-1]}_+Y({v}_1,w)\lim_{\epsilon\rightarrow 0}\hspace{-1mm}{}^F\epsilon^{-1}\widetilde{Q}_+\widetilde{v}_2\rangle.
\end{align*}
Thus, from (\ref{eq:deri-deform}), we have
\begin{align}
\begin{aligned}
\mathbb{G}^{[p_+]}_{+}Y(v_1,w)v_2
&={Q}^{[p_+-1]}_+Y(\widehat{Q}_+v_1,w){v}_2+{Q}^{[p_+-1]}_+Y({v}_1,w)\widehat{Q}_+v_2\\
&={Q}^{[p_+-1]}_+e^{wL_{-1}}Y(v_2,-w)\widehat{Q}_+v_1 +{Q}^{[p_+-1]}_+Y({v}_1,w)\widehat{Q}_+v_2\\
&=e^{wL_{-1}}Y(v_2,-w){Q}^{[p_+-1]}_+\widehat{Q}_+v_1 +Y({v}_1,w){Q}^{[p_+-1]}_+\widehat{Q}_+v_2\\
&=Y({Q}^{[p_+-1]}_+\widehat{Q}_+v_1,w){v}_2+Y({v}_1,w){Q}^{[p_+-1]}_+\widehat{Q}_+v_2,
\end{aligned}
\label{eq:g-v1v2}
\end{align}
where in the third equality, we used the fact that ${Q}^{[p_+-1]}_+$ is a $\mathcal{W}_{p_+,p_-}$-homomorphism \cite{NT,TW,Nak}. Hence, we obtain the desired identity.
\end{proof}

\begin{Remark}
It is clear that $:{Q}_+(z)\phi_0(z):$ and $Y(v,w)$ $(v\in \mathcal{W}_{p_+,p_-})$ are mutually local. 
Therefore, the argument preceding~(\ref{eq:g-v1v2}) in the proof of Theorem~\ref{G-deri} could in principle be omitted. 
On the other hand, it is natural to expect that the general operators $G^{[r^\vee,r]}_{+,k}$ and $G^{[s^\vee,s]}_{-,k}$ satisfy properties similar to those discussed in \cite[Theorem~4.17]{TW}. 
Since the first part of the proof of Theorem~\ref{G-deri} is likely to be important for establishing such properties, we include it for completeness.
\label{rem:g-hat}
\end{Remark}

\section{Application to the symmetry of $\mathcal{W}_{p_+,p_-}$}
\label{symWpq}
In this section, we show that the Lie algebra $\mathfrak{sl}_2(\mathbb{C})$ acts on $\mathcal{W}_{p_+,p_-}$ using the derivations constructed in the previous section.
\subsection{The $\mathfrak{sl}_2$-symmetry of $\mathcal{W}_{1,p}$}
\label{hidden-1}
Let $(p_+,p_-)=(1,p)$.
In what follows, we use the following notation
\begin{align*}
&
\mathbb{E}
=
Q_+,
&
\mathbb{F}
=
\mathbb{G}^{[p]}_-|_{\mathcal{W}_{1,p}}
.
\end{align*}
We use the same notation for the restrictions of $\mathbb{E}$ and $\mathbb{F}$ to $\mathcal{W}_{1,p}$.

It is known that $\mathcal{W}_{1,p}$ has automorphism group, and its structure was determined by \cite{AdamovicD/LinX/MilasA:2013}.
\begin{thm}[\cite{AdamovicD/LinX/MilasA:2013}]
\label{thm-ALM}
The full automorphism group of $\mathcal{W}_{1,p}$ is $PSL_2(\mathbb{C})$.
\end{thm}
The key points in the proof of this theorem are to show the existence of a homomorphism from \(  {\rm Aut}(\mathcal{W}_{1,p}) \) to \( PSL_2(\mathbb{C})\), and to show that \( PSL_2(\mathbb{C}) \) is a subgroup of the full automorphism group \( {\rm Aut}(\mathcal{W}_{1,p}) \).
The former essentially follows from the structure of the Zhu-algebra ${A}(\mathcal{W}_{1,p})$ \cite{AdamovicD/MilasA:2008}, while the latter follows from the fact that \( \mathfrak{sl}_2 \) acts on \( \mathcal{W}_{1,p} \) by derivations, a hidden symmetry due to \cite{AdamovicD/LinX/MilasA:2013}.
Although the next theorem follows immediately from Theorems~\ref{thm-ALM} and \ref{G-deri}, we include it here to show that it can also be derived by our method.
\begin{thm}
\label{thm:mainW1p}
Define $h\in {\rm End}_{\mathfrak{Vir}}(\mathcal{W}_{1,p})$ as follows:
\begin{align*}
h=-\frac{2a_0}{p\alpha_-}.
\end{align*}
Then there exists a nonzero constant $c_{\mathbb{E},\mathbb{F}}$ such that $\mathbb{E}$, $h$, and $\mathbb{F}$ satisfy the relation
\begin{align}
\label{eq:efh}
[\mathbb{E},\mathbb{F}]=c_{\mathbb{E},\mathbb{F}}h.
\end{align}
\end{thm}
Note that, by definition, $\mathbb{E}$, $\mathbb{F}$, and $h$ satisfy $[h,\mathbb{E}]=2\mathbb{E}$ and $[h,\mathbb{F}]=-2\mathbb{F}$. Together with (\ref{eq:efh}), we see that $\mathbb{E}$, $h$, and $c^{-1}_{\mathbb{E},\mathbb{F}}\mathbb{F}$ form an $\mathfrak{sl}_2$-triple. 
In particular, the Lie algebra $\mathfrak{sl}_2(\mathbb{C})$ acts on $\mathcal{W}_{1,p}$ by derivations.
By integrating this $\mathfrak{sl}_2(\mathbb{C})$-action, one can recover the result of \cite{AdamovicD/LinX/MilasA:2013} that the full automorphism group ${\rm Aut}(\mathcal{W}_{1,p})$ contains $PSL_2(\mathbb{C})$ as a subgroup.
In what follows, we show that Theorem~\ref{thm:mainW1p} follows directly from the definitions of the derivations $\mathbb{E}$ and $\mathbb{F}$ and Proposition \ref{sl2action2}.
\begin{proof}[Proof of Theorem \ref{thm:mainW1p}]
Noting Theorem~\ref{non-Gop}, we define constants $c_{\mathbb{E},i}$ and $c_{\mathbb{F},i}$ $(i=1,2)$ by
\begin{align*}
\begin{aligned}
\mathbb{E}W^-&=c_{\mathbb{E},1}W^0,&\mathbb{E}W^0&=c_{\mathbb{E},2}W^+,\\
\mathbb{F}W^+&=c_{\mathbb{F},1}W^0,&\mathbb{F}W^0&=c_{\mathbb{F},2}W^-.
\end{aligned}
\end{align*}
Then we have
\begin{align*}
\begin{aligned}
\lbrack\mathbb{E},{\mathbb{F}}\rbrack W^-&=-\mathbb{F}\cdot \mathbb{E}W^-=-c_{\mathbb{F},2}c_{\mathbb{E},1}W^-,\\
\lbrack\mathbb{E},{\mathbb{F}}\rbrack W^+&=\mathbb{E}\cdot \mathbb{F}W^+=c_{\mathbb{F},1}c_{\mathbb{E},2}W^+,\\
\lbrack\mathbb{E},{\mathbb{F}}\rbrack W^0&=(c_{\mathbb{E},1}c_{\mathbb{F},2}-c_{\mathbb{E},2}c_{\mathbb{F},1})W^0.
\end{aligned}
\end{align*}
Here, assuming $c_{\mathbb{E},1}c_{\mathbb{F},2}=c_{\mathbb{E},2}c_{\mathbb{F},1}$, and setting $c_{\mathbb{E},\mathbb{F}}=\frac{1}{2}c_{\mathbb{F},2}c_{\mathbb{E},1}$, we see that $W^+$, $W^0$, and $W^-$ form the three-dimensional irreducible representation of $\mathfrak{sl}_2(\mathbb{C})=\mathbb{C}\mathbb{E}\oplus \mathbb{C}h\oplus \mathbb{C}c^{-1}_{\mathbb{E},\mathbb{F}}\mathbb{F}$.
Combining this assumption with Proposition~\ref{sl2action2} and Theorem~\ref{G-deri}, we see that $\{w^{(n)}_k\}_{k=-n}^n$ realizes the $(2n+1)$-dimensional irreducible representation of $\mathfrak{sl}_2(\mathbb{C})=\mathbb{C}\mathbb{E}\oplus \mathbb{C}h\oplus \mathbb{C}c^{-1}_{\mathbb{E},\mathbb{F}}\mathbb{F}$, and from Theorem~\ref{G-hom}, (\ref{eq:efh}) holds on $\mathcal{W}_{1,p}$.
Therefore, to prove (\ref{eq:efh}), it suffices to show that 
\begin{align}
\label{cru-ef-rel}
[\mathbb{E},\mathbb{F}]|_{\mathcal{W}_{1,p}\cap F_{1,1}}=0.
\end{align}
Fix $\phi\in F_{1,1}\cap \mathcal{W}_{1,p}$.
Note that $F_{1,1}\subset \widetilde{F}_{1,1}$, and that the screening operators $\widetilde{Q}_\pm$ and the compositions $\widetilde{Q}_\pm \widetilde{Q}_\mp$ are well-defined on $\widetilde{F}_{1,1}$.
By Lemma \ref{lem:tw} and Proposition~\ref{prop:qr-qs}, we have
\begin{align}
\label{eq:19191}
\begin{aligned}
\mathbb{E}\cdot \mathbb{F}\phi &=Q_+Q^{[p-1]}_-\lim_{\epsilon\rightarrow 0}\hspace{-1mm}{}^F\frac{1}{\epsilon}\widetilde{Q}_-\phi\\
&=Q^{[p-1]}_-Q_+\lim_{\epsilon\rightarrow 0}\hspace{-1mm}{}^F\frac{1}{\epsilon}\widetilde{Q}_-\phi.
\end{aligned}
\end{align}
Note that $\widetilde{\alpha}_+\widetilde{\alpha}_-=-2$ (see (\ref{notep1})). Then from Proposition~\ref{prop:qr-qs}, it follows that $[\widetilde{Q}_+,\widetilde{Q}_-]=0$ on $\widetilde{F}_{1,1}$.
Using this commutativity and Lemma~\ref{lem:tw0}, we have
\begin{align}
\label{eq:19192}
\begin{aligned}
Q^{[p-1]}_-Q_+\lim_{\epsilon\rightarrow 0}\hspace{-1mm}{}^F\frac{1}{\epsilon}\widetilde{Q}_-\phi
&=Q^{[p-1]}_-\lim_{\epsilon\rightarrow 0}\hspace{-1mm}{}^F\frac{1}{\epsilon}\widetilde{Q}_+\widetilde{Q}_-\phi\\
&=Q^{[p-1]}_-\lim_{\epsilon\rightarrow 0}\hspace{-1mm}{}^F\frac{1}{\epsilon}\widetilde{Q}_-\widetilde{Q}_+\phi.
\end{aligned}
\end{align}
Note that
\begin{align*}
\widetilde{Q}_+\phi-{\rm e}^{\widetilde{\alpha}_{-1,1}-\alpha_{-1,1}}Q_+\phi\in \epsilon\cdot \widetilde{F}_{-1,1}.
\end{align*}
Then, by using the relation $Q^{[p-1]}_-\circ Q_-=0$ and Lemmas~\ref{lem:tw0}-\ref{lem:tw}, we have
\begin{align}
\label{eq:19193}
\begin{aligned}
Q^{[p-1]}_-\lim_{\epsilon\rightarrow 0}\hspace{-1mm}{}^F\frac{1}{\epsilon}\widetilde{Q}_-\widetilde{Q}_+\phi
&=Q^{[p-1]}_-\lim_{\epsilon\rightarrow 0}\hspace{-1mm}{}^F\frac{1}{\epsilon}\widetilde{Q}_-{\rm e}^{\widetilde{\alpha}_{-1,1}-\alpha_{-1,1}}Q_+\phi\\
&\quad\quad +Q^{[p-1]}_-\lim_{\epsilon\rightarrow 0}\hspace{-1mm}{}^F\frac{1}{\epsilon}\widetilde{Q}_-\bigl(\widetilde{Q}_+\phi-{\rm e}^{\widetilde{\alpha}_{-1,1}-\alpha_{-1,1}}Q_+\phi\bigr)\\
&=\mathbb{F}\cdot \mathbb{E}\phi+Q^{[p-1]}_-\circ{Q}_-\lim_{\epsilon\rightarrow 0}\hspace{-1mm}{}^F\frac{1}{\epsilon}\bigl(\widetilde{Q}_+\phi-{\rm e}^{\widetilde{\alpha}_{-1,1}-\alpha_{-1,1}}Q_+\phi\bigr)\\
&=\mathbb{F}\cdot \mathbb{E}\phi.
\end{aligned}
\end{align}
Therefore, from (\ref{eq:19191})-(\ref{eq:19193}), we get (\ref{cru-ef-rel}).
\end{proof}

\subsection{The $\mathfrak{sl}_2$-symmetry of $\mathcal{W}_{p_+,p_-}$ $(p_+\geq2)$}
\label{hidden-2}
Throughout this subsection, we assume that $p_+\geq 2$.
Although these notation overlaps with those of the previous subsection, we introduce the following notation
\begin{align*}
&
\mathbb{E}
=
\left.\mathbb{G}^{[p_+]}_+\right|_{\mathcal{W}_{p_+,p_-}},
&
\mathbb{F}
=
\left.\mathbb{G}^{[p_-]}_-\right|_{\mathcal{W}_{p_+,p_-}}
.
\end{align*}
We use the same notation for the restrictions of $\mathbb{E}$ and $\mathbb{F}$ to $\mathcal{W}_{p_+,p_-}$.

As in $\mathcal{W}_{1,p}$, it is known that $\mathcal{W}_{p_+,p_-}$ has automorphism group, and its structure was determined by \cite{McRaeR/SopinV:2026} using the theory of commutative algebras in braided tensor categories.
\begin{thm}[\cite{McRaeR/SopinV:2026}]
\label{thm-MS}
The full automorphism group of $\mathcal{W}_{p_+,p_-}$ is $PSL_2(\mathbb{C})$.
\end{thm}
Although the next theorem follows immediately from Theorems~\ref{thm-MS} and~\ref{G-deri}, we include it here to show that it can also be derived by our method.
\begin{thm}
\label{thm:mainWpq}
Define $h\in {\rm End}_{\mathfrak{Vir}}(\mathcal{W}_{p_+,p_-})$ as follows:
\begin{align*}
h=-\frac{2a_0}{p_-\alpha_-}.
\end{align*}
Then there exists a nonzero constant $c_{\mathbb{E},\mathbb{F}}$ such that $\mathbb{E}$, $h$, and $\mathbb{F}$ satisfy the relation
\begin{align}
\label{eq:efh-pq}
[\mathbb{E},\mathbb{F}]=c_{\mathbb{E},\mathbb{F}}h.
\end{align}
\end{thm}
Note that $\mathbb{E}$, $\mathbb{F}$, and $h$ satisfy $[h,\mathbb{E}]=2\mathbb{E}$ and $[h,\mathbb{F}]=-2\mathbb{F}$. Together with (\ref{eq:efh-pq}), we see that $\mathbb{E}$, $h$, and $c^{-1}_{\mathbb{E},\mathbb{F}}\mathbb{F}$ form an $\mathfrak{sl}_2$-triple. 
In particular, the Lie algebra $\mathfrak{sl}_2(\mathbb{C})$ acts on $\mathcal{W}_{p_+,p_-}$ by derivations.
As in the case of \({\rm Aut}(\mathcal{W}_{1,p})\), we recover Theorem~\ref{thm-MS} using the above $\mathfrak{sl}_2$-derivation properties and the structure of the Zhu-algebra $A(\mathcal{W}_{p_+,p_-})$ determined in \cite{AdamovicD/MilasA:2010,AdamovicD/MilasA:2011,TW}.
Before the proof of Theorem~\ref{thm:mainWpq}, we introduce the following lemma.
\begin{lem}
\label{lem:sc-cruc}
Let $p_+>r\geq 1$, $p_->s\geq 1$, and $n\in \mathbb{Z}$. 
\begin{enumerate}
\item For any $\phi\in F_{1,s;-n}\cap {\rm ker}Q_+$, we have
\begin{align}
\label{d+-}
\begin{aligned}
&\widetilde{Q}^{[p_+-1]}_+\circ {\rm e}^{\widetilde{\alpha}_{p_+-1,-(n-1)p_--s}-\widetilde{\alpha}_{-1,-np_--s}}\\
&\qquad \qquad \qquad \circ \widetilde{Q}_+\circ {\rm e}^{\widetilde{\alpha}_{1,-np_--s}-\widetilde{\alpha}_{1+np_+,-s}}\circ \widetilde{Q}^{[s]}_-\circ {\rm e}^{\widetilde{\alpha}_{1+np_+,s}-\alpha_{1+np_+,s}} \phi\\
&- \widetilde{Q}^{[p_+-1]}_+\circ{\rm e}^{\widetilde{\alpha}_{p_+-1,-s-(n-1)p_-}-\widetilde{\alpha}_{-1+np_+,-s}}\\
&\qquad \qquad\qquad \circ\widetilde{Q}^{[s]}_-\circ {\rm e}^{\widetilde{\alpha}_{-1+np_+,s}-\widetilde{\alpha}_{-1,s-np_-}} \circ \widetilde{Q}_+\circ {\rm e}^{\widetilde{\alpha}_{1,s-np_-}-\alpha_{1,s-np_-}} \phi\\
&\in \epsilon^2\cdot \widetilde{F}_{1-p_+,-s-(n-1)p_-}.
\end{aligned}
\end{align}
\item  For any $\phi\in F_{r,1;-n}\cap {\rm ker}Q_-$, we have
\begin{align*}
\begin{aligned}
&\widetilde{Q}^{[p_--1]}_-\circ {\rm e}^{\widetilde{\alpha}_{-r+(n+1)p_+,p_--1}-\widetilde{\alpha}_{np_+-r,-1}}\\
&\qquad \qquad \qquad \circ\widetilde{Q}_-\circ {\rm e}^{\widetilde{\alpha}_{np_+-r,1}-\widetilde{\alpha}_{-r,1-np_-}}\circ \widetilde{Q}^{[r]}_+\circ {\rm e}^{\widetilde{\alpha}_{r,1-np_-}-\alpha_{r,1-np_-}} \phi\\
&- \widetilde{Q}^{[p_--1]}_-\circ {\rm e}^{\widetilde{\alpha}_{-r+(n+1)p_+,p_--1}-\widetilde{\alpha}_{-r,-np_--1}}\\
&\qquad \qquad\qquad \circ\widetilde{Q}^{[r]}_+\circ {\rm e}^{\widetilde{\alpha}_{r,-np_--1}-\widetilde{\alpha}_{r+np_+,-1}}\circ \widetilde{Q}_-\circ {\rm e}^{\widetilde{\alpha}_{r+np_+,1}-\alpha_{r+np_+,1}} \phi\\
&\in \epsilon^2\cdot \widetilde{F}_{-r+(n+1)p_+,1-p_-}.
\end{aligned}
\end{align*}
\end{enumerate}
\end{lem}
Note that the above compositions are well-defined by the definitions of the screening operators and the shifting elements ${\rm e}^{\bullet}$ (see (\ref{notep2}) and (\ref{conjugate2})).
The left-hand side of (\ref{d+-}) can be regarded as an $\epsilon$-deformation of $Q^{[p_+-1]}_+[Q_+,Q^{[s]}_-]\phi$.
The statement (\ref{d+-}) shows that, under the restriction to $F_{1,s;-n}\cap {\rm ker}Q_+$, the image of the $\epsilon$-deformation of $Q^{[p_+-1]}_+[Q_+,Q^{[s]}_-]$ vanishes to order $\epsilon^{2}$.
\begin{proof}[Proof of Lemma \ref{lem:sc-cruc}]
We only prove the first case. The second case can be proved in the same way.
Note that 
$Q_+\circ Q^{[s]}_-=0$ and $Q^{[s]}_-\circ Q_+ =0$ on ${\rm Soc}_2(F_{1,s;-n})$.
Then, noting Proposition~\ref{prop:tw}, we can define linear operators
\begin{align*}
\begin{aligned}
T_1:&=\lim_{\epsilon\rightarrow 0}\hspace{-1mm}{}^F\frac{1}{\epsilon}\widetilde{Q}_+\circ {\rm e}^{\widetilde{\alpha}_{1,-np_--s}-\widetilde{\alpha}_{1+np_+,-s}}\circ \widetilde{Q}^{[s]}_-\circ {\rm e}^{\widetilde{\alpha}_{1+np_+,s}-\alpha_{1+np_+,s}},\\
T_2:&=\lim_{\epsilon\rightarrow 0}\hspace{-1mm}{}^F\frac{1}{\epsilon}\widetilde{Q}^{[s]}_-\circ {\rm e}^{\widetilde{\alpha}_{-1+np_+,s}-\widetilde{\alpha}_{-1,s-np_-}} \circ \widetilde{Q}_+\circ {\rm e}^{\widetilde{\alpha}_{1,s-np_-}-\alpha_{1,s-np_-}}
\end{aligned}
\end{align*}
which are well-defined on ${\rm Soc}_2(F_{1,s;-n})$. 
By Proposition~\ref{prop:tw}, in order to prove (\ref{d+-}), it is suffices to show that 
\begin{align}
\label{final-t1t2}
Q^{[p_+-1]}_+\circ T_1|_{F_{1,s;-n}\cap {\rm ker}Q_+}=Q^{[p_+-1]}_+\circ T_2|_{F_{1,s;-n}\cap {\rm ker}Q_+}.
\end{align}
We divide the proof into four steps.

\begin{enumerate}

\item 
We introduce the notation 
\begin{equation}
\label{gamma-twist-def}
\Gamma^{(\rho_\pm)}_{n}
=
\begin{cases}
(2\pi i)^{-1}C_{0}\times [\Delta^{(\rho_\pm)}_{n-1}],&n\geq 2\\ 
(2\pi i)^{-1}C_{0},&n=1
\end{cases}
\end{equation}
where $C_0$ is the positively oriented circle centered at $0$ and $[\Delta^{(\rho_\pm)}_{n-1}]$ is the twisted cycle defined by (\ref{eq:twist-rho}). 
Let us show the equation
\begin{align}
\label{lem:sc-cruc-st}
\begin{aligned}
&T_1-T_2=\frac{-2np_+}{\alpha_++\alpha_-}\sum_{i=1}^s\int_{\Gamma^{(\alpha_-)}_{s}} \frac{1}{x_i}Q_-(x_1)\cdots Y(\ket{\alpha_0},x_i)\cdots Q_-(x_s){\rm d}\bm{x}
\end{aligned}
\end{align}
where ${\rm d}\bm{x}={\rm d}x_1\wedge \cdots \wedge{\rm d}x_s$, and we apply the product (\ref{gamma-twist-def}) by setting $x_1=z$, $x_2=zy_1,\dots,x_s=zy_{s-1}$ in accordance with (\ref{Tsuchiya-Kanie0}).
Let $\widetilde{T}_1$ and $\widetilde{T}_2$ denote the expressions before taking the limits of $T_1$ and $T_2$, respectively.
Using the operator product expansion
\begin{align*}
\widetilde{Q}_+(z)\widetilde{Q}_-(x)=\frac{1}{(z-x)^{2}}:\widetilde{Q}_+(z)\widetilde{Q}_-(x):
\end{align*}
(cf.~\cite[Proposition 2.1]{NT}), we have
\begin{align}
\label{lem:sc-cruc-2-0}
\begin{aligned}
\widetilde{T}_1
&=\frac{1}{\epsilon}\int _{C_z}\int_{\Gamma^{(\widetilde{\alpha}_-)}_{s}}\widetilde{Q}_+(z)\cdot {\rm e}^{\widetilde{\alpha}_{1,-np_--s}-\widetilde{\alpha}_{1+np_+,-s}}\\
&\qquad \qquad \qquad \cdot \widetilde{Q}_-(x_1)\cdots \widetilde{Q}_-(x_s)\cdot {\rm e}^{\widetilde{\alpha}_{1+np_+,s}-\alpha_{1+np_+,s}} {\rm d}\bm{x}\frac{{\rm d}z}{2\pi i}\\
&=\frac{1}{\epsilon}\int _{C_z}\int_{\Gamma^{(\widetilde{\alpha}_-)}_{s}}{\rm e}^{\widetilde{\alpha}_{1,-np_--s}-\widetilde{\alpha}_{1+np_+,-s}}\cdot \bigl(z^{\frac{np_+}{2}\widetilde{\alpha}^2_+-np_-}\widetilde{Q}_+(z) \bigr)\\
&\qquad \qquad \qquad \cdot \widetilde{Q}_-(x_1)\cdots \widetilde{Q}_-(x_s)\cdot {\rm e}^{\widetilde{\alpha}_{1+np_+,s}-\alpha_{1+np_+,s}} {\rm d}\bm{x}\frac{{\rm d}z}{2\pi i}\\
&=\frac{1}{\epsilon}\int_{\Gamma^{(\widetilde{\alpha}_-)}_{s}}\oint _{{z=0}}{\rm e}^{\widetilde{\alpha}_{1,-np_--s}-\widetilde{\alpha}_{1+np_+,-s}}\cdot \widetilde{Q}_-(x_1)\cdots \widetilde{Q}_-(x_s)\\
&\qquad \qquad \qquad \cdot \bigl(z^{\frac{np_+}{2}\widetilde{\alpha}^2_+-np_-}\widetilde{Q}_+(z) \bigr)\cdot {\rm e}^{\widetilde{\alpha}_{1+np_+,s}-\alpha_{1+np_+,s}} \frac{{\rm d}z}{2\pi i}{\rm d}\bm{x}\\
&\qquad +\frac{1}{\epsilon}\Psi_1(\epsilon)+\frac{1}{\epsilon}\Psi_2(\epsilon),
\end{aligned}
\end{align}
where $C_z$ denotes the positively oriented circle centered at $z=0$ with sufficiently large radius, and we denote
\begin{align*}
\begin{aligned}
\Psi_1(\epsilon)&:=\sum_{i=1}^s\int_{\Gamma^{(\widetilde{\alpha}_-)}_{s}}{\rm e}^{\widetilde{\alpha}_{1,-np_--s}-\widetilde{\alpha}_{1+np_+,-s}}\\
&\qquad \qquad\qquad \cdot (\partial_{x_i}x_i^{\frac{np_+}{2}\widetilde{\alpha}^2_+-np_-})\widetilde{Q}_-(x_1)\cdots Y(\ket{\widetilde{\alpha}_0},x_i)\cdots \widetilde{Q}_-(x_s)\\
&\qquad \qquad\qquad\qquad \qquad\cdot {\rm e}^{\widetilde{\alpha}_{1+np_+,s}-\alpha_{1+np_+,s}}{\rm d}\bm{x},\\
\Psi_2(\epsilon)&:=\frac{\widetilde{\alpha}_+}{\widetilde{\alpha}_0}\sum_{i=1}^s\int_{\Gamma^{(\widetilde{\alpha}_-)}_{s}}{\rm e}^{\widetilde{\alpha}_{1,-np_--s}-\widetilde{\alpha}_{1+np_+,-s}}\\
&\qquad \qquad\qquad \cdot x_i^{\frac{np_+}{2}\widetilde{\alpha}^2_+-np_-}\widetilde{Q}_-(x_1)\cdots \partial_{x_i}Y(\ket{\widetilde{\alpha}_0},x_i)\cdots \widetilde{Q}_-(x_s)\\
&\qquad \qquad\qquad \qquad \qquad \cdot {\rm e}^{\widetilde{\alpha}_{1+np_+,s}-\alpha_{1+np_+,s}}{\rm d}\bm{x}.
\end{aligned}
\end{align*}
Note that
\begin{align*} 
&z^{\frac{np_+}{2}\widetilde{\alpha}^2_+-np_-}\widetilde{Q}_+(z) {\rm e}^{\widetilde{\alpha}_{1+np_+,s}-\alpha_{1+np_+,s}}\\
&={\rm e}^{\widetilde{\alpha}_{-1+np_+,s}-\widetilde{\alpha}_{-1,s-np_-}}\widetilde{Q}_+(z) {\rm e}^{\widetilde{\alpha}_{1,s-np_-}-\alpha_{1,s-np_-}},
\end{align*}
and that ${\rm e}^{\widetilde{\alpha}_{1,-np_--s}-\widetilde{\alpha}_{1+np_+,-s}}={\rm e}^{\widetilde{\alpha}_{-1,-np_--s}-\widetilde{\alpha}_{-1+np_+,-s}}$.
Then, from (\ref{lem:sc-cruc-2-0}), we have
\begin{align}
\label{lem:sc-cruc-2}
\widetilde{T}_1-{\rm e}^{\widetilde{\alpha}_{-1,-np_--s}-\widetilde{\alpha}_{-1+np_+,-s}}\circ \widetilde{T}_2=\frac{1}{\epsilon}\Psi_1(\epsilon)+\frac{1}{\epsilon}\Psi_2(\epsilon).
\end{align}
Noting the expansion $\frac{np_+}{2}\widetilde{\alpha}^2_+-np_-=np_+\alpha_+\epsilon+o(\epsilon)$ and using Theorem~\ref{sus-prop}, we see that
$
{\rm im}\Psi_1(\epsilon)\subset \epsilon \cdot\widetilde{F}_{-1,-np_--s}.
$
Then, we have
\begin{align}
\label{lem:sc-cruc-22}
\left.\frac{1}{\epsilon}\Psi_1(\epsilon)\right|_{\epsilon=0}={np_+\alpha_+}\sum_{i=1}^s\int_{\Gamma^{(\alpha_-)}_{s}} \frac{1}{x_i}Q_-(x_1)\cdots Y(\ket{\alpha_0},x_i)\cdots Q_-(x_s){\rm d}\bm{x}.
\end{align}
The operator $\Psi_2(\epsilon)$ can be rewritten as follows 
\begin{align*}
\begin{aligned}
\Psi_2(\epsilon)&=\frac{\widetilde{\alpha}_+}{\widetilde{\alpha}_0}\sum_{i=1}^s\int_{\Gamma^{(\widetilde{\alpha}_-)}_{s}}{\rm e}^{\widetilde{\alpha}_{1,-np_--s}-\widetilde{\alpha}_{1+np_+,-s}}\\
&\quad \cdot x_i^{\frac{np_+}{2}\widetilde{\alpha}^2_+-np_-}{\rm d}_{\bm{x}}\biggl\{(-1)^{i+1}\cdot \widetilde{Q}_-(x_1) \cdots Y(\ket{\widetilde{\alpha}_0},x_i)\bigr)\cdots  \widetilde{Q}_-(x_s)\\
&\qquad \qquad\qquad \cdot  {\rm e}^{\widetilde{\alpha}_{1+np_+,s}-\alpha_{1+np_+,s}} {\rm d}x_1\wedge \cdots\wedge\overset{\vee}{{\rm d}x_i}\wedge \cdots \wedge {\rm d}x_{s}\biggr\}\\
&=-\frac{\widetilde{\alpha}_+}{\widetilde{\alpha}_0}\sum_{i=1}^s\int_{\Gamma^{(\widetilde{\alpha}_-)}_{s}}{\rm e}^{\widetilde{\alpha}_{1,-np_--s}-\widetilde{\alpha}_{1+np_+,-s}}\\
&\qquad \qquad\qquad \cdot (\partial_{x_i}x_i^{\frac{np_+}{2}\widetilde{\alpha}^2_+-np_-})\widetilde{Q}_-(x_1)\cdots Y(\ket{\widetilde{\alpha}_0},x_i)\cdots \widetilde{Q}_-(x_s)\\
&\qquad \qquad\qquad\qquad\qquad\cdot {\rm e}^{\widetilde{\alpha}_{1+np_+,s}-\alpha_{1+np_+,s}}{\rm d}\bm{x}\\
&=-\frac{\widetilde{\alpha}_+}{\widetilde{\alpha}_0}\Psi_1(\epsilon),
\end{aligned}
\end{align*}
where ${\rm d}_{\bm{x}}$ is the total derivative of $x_1,\dots, x_s$, and we used the generalized Stokes' theorem (cf.~\cite{AK}). 
Then, from (\ref{lem:sc-cruc-2}) and (\ref{lem:sc-cruc-22}), we get (\ref{lem:sc-cruc-st}).

\item 
Recall that the cycle $\Gamma^{(\alpha_-)}_{s}$ is defined as the product of the circle $C_0$ in the $z$-variable and the cycle $[\Delta^{(\alpha_-)}_{s-1}]$ in the $\bm{y}$-variables
 (see (\ref{gamma-twist-def})).
 In what follows, let $r_0$ denote the radius of $C_0$.
 We define the linear operators $R_1\in {\rm Hom}_{\mathbb{C}}(F_{1,s;-n},F_{-1,-s;-n})$, $R_{2,r_0}\in {\rm Hom}_{\mathbb{C}}(F_{1,s;-n},F_{1,-s;-n})$, and the field $P(z)\in {\rm Hom}_{\mathbb{C}}(F_{1,s;-n},F_{-1,-s;-n})[[z,z^{-1}]]$ as follows
\begin{align}
 R_1&:=\sum_{i=1}^s\int_{\Gamma^{(\alpha_-)}_s}\frac{1}{x_i} Q_-(x_1)\cdots Y(\ket{\alpha_0},x_i)\cdots Q_-(x_s){\rm d}\bm{x},\nonumber\\
 R_{2,r_0}&:=\left.\partial_{\epsilon}\int_{C_0\times [\Delta^{(\alpha_-;\epsilon)}_{s-1}]}x^\epsilon_1x^\epsilon_2\cdots x^\epsilon_s Q_-(x_1)Q_-(x_2) \cdots Q_-(x_s)\frac{{\rm d}\bm{x}}{2\pi i}\right|_{\epsilon=0},\label{r20-eqlog}\\
 P(z)&:=\int_{[\Delta^{(\alpha_-)}_{s-1}]} Y(\ket{\alpha_0},z)Q_-(zy_1)\cdots Q_-(zy_{s-1})z^{s-1}{\rm d}\bm{y}\nonumber,
\end{align}
where we choose the branch cut of $\log z$ (a factor in the integrand of $R_{2,r_0}$) along the negative real axis under $x_1=z,x_2=zy_1,\dots,x_s=zy_{s-1}$, and we use the nation
\begin{align}
\label{eq:sig-del}
[\Delta^{(\sigma;t)}_{s-1}]:=\Bigl\lbrack\Delta_{s-1}\Bigl((1-s)\frac{{\sigma}^2}{2}+{t},{\sigma}^2,\frac{{\sigma}^2}{2}\Bigr)\Bigr\rbrack.
\end{align}
We see that $[\Delta^{(\sigma;0)}_{s-1}]=[\Delta^{(\sigma)}_{s-1}]$ (for the notation $[\Delta^{(\sigma)}_{s-1}]$, see (\ref{eq:twist-rho})).
Note that the right-hand side of (\ref{r20-eqlog}) is well-defined. In fact, 
for sufficiently small $|\epsilon|\geq 0$, the parameters defining the twisted cycle $[\Delta^{(\alpha_-;\epsilon)}_{s-1}]$ satisfy
\begin{align*}
\Bigl((1-s)\frac{{\alpha}^2_-}{2}+{\epsilon},{\alpha}^2_-,\frac{{\alpha}^2_-}{2}\Bigr)\notin \widehat{\mathcal{A}}_{s-1}.
\end{align*}
Thus, from Theorem \ref{sus-prop}, we see that the correlation functions
\begin{align*}
&z^{s-1+s\epsilon}\int_{[\Delta^{(\alpha_-;\epsilon)}_{s-1}]}y^\epsilon_1\cdots y^\epsilon_{s-1}\langle u', Q_-(z)Q_-(zy_1) \cdots Q_-(zy_{s-1})u\rangle{{\rm d}\bm{y}},\\
&\qquad (u\in F_{1,s;-n}, u'\in F^*_{1,-s;-n})
\end{align*}
are holomorphic at $\epsilon=0$.
For simplicity, we write 
\begin{align*}
R_{2,r_0}=\int_{\Gamma^{(\alpha_-)}_{s}}\log( x_1x_2\cdots x_s) Q_-(x_1)Q_-(x_2) \cdots Q_-(x_s){\rm d}\bm{x},
\end{align*}
which does not affect the following argument.
From (\ref{lem:sc-cruc-st}), the operator $R_1$ agrees with $T_1-T_2$ up to a constant multiple. Therefore, to prove (\ref{final-t1t2}), it suffices to show 
\begin{align}
\label{final-t1t2t3}
Q^{[p_+-1]}_+\circ R_1|_{F_{1,s;-n}\cap {\rm ker}Q_+}=0.
\end{align}
Below we show that, on $F_{1,s;-n}\cap {\rm ker}Q_+$, the zero mode of $Q^{[p_+-1]}_+P(z)$ acts as a scalar multiple of the left-hand side of $F_{1,s;-n}\cap {\rm ker}Q_+$, while all the other modes act trivially.
As in \cite[Proposition 2.1]{NT},
by using the operator product expansion 
\begin{align*}
{Q}_+(z){Q}_-(x)=\frac{1}{(z-x)^{2}}:{Q}_+(z){Q}_-(x):
\end{align*}
and the generalized Stokes' theorem,
we have
\begin{align*}
\begin{aligned}
&[Q_+,R_{2,r_0}]\\
&=\sum_{j=1}^s\int_{\Gamma^{(\alpha_-)}_{s}} \log( x_1x_2\cdots x_s)Q_-(x_1)\cdots :\partial_{x_j}Q_+(x_j)Q_-(x_j):\cdots Q_-(x_s){\rm d}\bm{x}\\
&=\frac{\alpha_+}{\alpha_0}\sum_{j=1}^s\int_{\Gamma^{(\alpha_-)}_{s}}\log\Bigl(\prod_{i=1}^sx_i\Bigr) Q_-(x_1)\cdots \partial_{x_j}Y(\ket{\alpha_0},x_j)\cdots Q_-(x_s){\rm d}\bm{x}\\
&=\frac{\alpha_+}{\alpha_0}\sum_{j=1}^s\int_{\Gamma^{(\alpha_-)}_{s}}\log\Bigl(\prod_{i=1}^sx_i\Bigr)\\
&\qquad  \cdot {\rm d}_{\bm{x}}\biggl\{(-1)^{j+1}Q_-(x_1)\cdots Y(\ket{\alpha_0},x_j)\cdots Q_-(x_s) {\rm d}x_1\wedge\cdots\overset{\vee}{{\rm d}x_j}\cdots \wedge{\rm d}x_{s}\biggr\}\\
&=\frac{\alpha_+}{\alpha_0}\int_{\Gamma^{(\alpha_-)}_{s}}\biggl\{s\log z+\log\Bigl(\prod_{i=1}^{s-1}y_i\Bigr)\biggr\}\\
&\qquad \cdot\biggl\lbrack z^{s-1}\partial_{z}\Bigl(Y(\ket{\alpha_0},z)Q_-(zy_1)\cdots Q_-(zy_{s-1})\Bigr){\rm d}z{\rm d}y_1\wedge\cdots \wedge{\rm d}y_{s-1}\\
&\qquad  -\sum_{j=1}^{s-1} z^{s-2}(-1)^{j+1}y_j{\rm d}_{\bm{y}}\biggl\{Y(\ket{\alpha_0},z)Q_-(zy_1)\cdots Q_-(zy_j)\cdots Q_-(zy_{s-1})\\
&\qquad \qquad \qquad \qquad \qquad \cdot  {\rm d}z{\rm d}y_1\wedge\cdots\overset{\vee}{{\rm d}y_j}\cdots \wedge{\rm d}y_{s-1}\biggr\}\\
&\qquad  +\sum_{j=1}^{s-1} z^{s-2}{\rm d}_{\bm{y}}\biggl\{(-1)^{j+1}Q_-(z)Q_-(zy_1)\cdots Y(\ket{\alpha_0},zy_j)\cdots Q_-(zy_{s-1})\\
&\qquad \qquad \qquad \qquad \qquad \cdot  {\rm d}z{\rm d}y_1\wedge\cdots\overset{\vee}{{\rm d}y_j}\cdots \wedge{\rm d}y_{s-1}\biggr\}\biggr\rbrack\\
&=\frac{\alpha_+}{\alpha_0}\Bigl(-R_1+sP(r_0)\Bigr).
\end{aligned}
\end{align*}
Then we get
\begin{align}
\label{lem:eq-x-10}
P(z)=\frac{\alpha_0}{\alpha_+s}[Q_+,R_{2,z}]+\frac{1}{s}R_1
\end{align}
for $z\in \mathbb{R}_{>0}$. Applying $Q^{[p_+-1]}_+$ to (\ref{lem:eq-x-10}) and using $Q^{[p_+-1]}_+\circ Q_+=0$, we get 
\begin{align}
\label{lem:eq-x-cr}
Q^{[p_+-1]}_+P(z)|_{F_{1,s;-n}\cap {\rm ker}Q_+}=s^{-1}Q^{[p_+-1]}_+\circ R_1|_{F_{1,s;-n}\cap {\rm ker}Q_+}.
\end{align}
This implies that all non-zero modes of $Q^{[p_+-1]}_+P(z)$ act trivially on $F_{1,s;-n}\cap {\rm ker}Q_+$. 

\item 
Set
\begin{align*}
\kappa_1&:=\alpha_+-\widetilde{\alpha}_++s(\alpha_--\widetilde{\alpha}_-),\\
\kappa_2&:=\widetilde{\alpha}_-(\alpha_{1,s;-n}-\widetilde{\alpha}_{1+np_+,s}+\kappa_1),\\
\kappa_3&:=\widetilde{\alpha}_+(\widetilde{\alpha}_{1,s-np_-}-{\alpha}_{1,s;-n}-\kappa_1)-s\kappa_2.
\end{align*}
Note that $\kappa_1,\kappa_2,\kappa_3\in \epsilon\cdot \mathcal{O}_{0;\epsilon}$.
We introduce the operator
\begin{align*}
\widetilde{P}(z)&:=\int_{[\Delta^{(\widetilde{\alpha}_-;\kappa_2)}_{s-1}]} Y(\ket{\widetilde{\alpha}_0},z)\widetilde{Q}_-(zy_1)\cdots \widetilde{Q}_-(zy_{s-1})z^{s-1}{\rm d}\bm{y},
\end{align*}
where we used the notation (\ref{eq:sig-del}) as $(\sigma,t)=(\widetilde{\alpha}_-,\kappa_2)$.
We see that from (\ref{eq:opeVV}), $\widetilde{P}(z)$ is well-defined on ${\rm e}^{\kappa_1}F_{1,s;-n}$, and that $z^{\kappa_3}\widetilde{P}(z)$ is a single-valued operator on ${\rm e}^{\kappa_1}F_{1,s;-n}$.
Note that, due to the shift by ${\rm e}^{\kappa_1}$, the compositions
\begin{align*}
&z^{\kappa_3}Q^{[p_+-1]}_{+}\widetilde{P}(z),
&z^{\kappa_3}Q^{[p_+-1]}_{+}\widetilde{P}_l(z) 
\end{align*}
are well-defined on ${\rm e}^{\kappa_1}F_{1,s;-n}$.
In this step, we rewrite the linear operator (\ref{lem:eq-x-cr}) in terms of $\widetilde{P}(z)$.

In what follows, we consider $\widetilde{P}(z)$ on ${\rm e}^{\kappa_1}F_{1,s;-n}$.
Note that for sufficiently small $|\epsilon|\geq 0$, the parameters defining the twisted cycle $[\Delta^{(\widetilde{\alpha}_-,\kappa_2)}_{s-1}]$ satisfy
\begin{align*}
\Bigl((1-s)\frac{\widetilde{\alpha}^2_-}{2}+{\kappa_2},\widetilde{\alpha}^2_-,\frac{\widetilde{\alpha}^2_-}{2}\Bigr)\notin \widehat{\mathcal{A}}_{s-1}.
\end{align*}
Then, from Theorem~\ref{sus-prop}, we have
\begin{align*}
\begin{aligned}
&\widetilde{P}(z)\in z^{-\kappa_3}\cdot {\rm Hom}({\rm e}^{\kappa_1}F_{1,s;-n}, \mathcal{O}_{0;\epsilon}\otimes F_{p_+-1,-s;-n+1})[[ z^{\pm 1}]].
\end{aligned}
\end{align*}
Thus, by the definition of $\widetilde{P}(z)$, we obtain
\begin{align}
\label{eq:ptil}
\left.\oint_{z=0}z^{\kappa_3}z^{l}\widetilde{P}(z){\rm e}^{\kappa_1} u\frac{{\rm d}z}{2\pi i}\right|_{\epsilon=0}=\oint_{z=0}z^lP(z)u\frac{{\rm d}z}{2\pi i},
\ \ u\in F_{1,s;-n},\ l\in \mathbb{Z}.
\end{align}
We denote
\begin{align*}
K&:={\rm e}^{\kappa_1}(F_{1,s;-n}\cap {\rm ker}Q_+).
\end{align*}
Note that from (\ref{lem:eq-x-cr}) and (\ref{eq:ptil}),
\begin{align}
\label{eq:ptil-2}
{\rm im}\Bigl(Q^{[p_+-1]}_+\left.\oint_{z=0}z^{\kappa_3}z^{l-1}\widetilde{P}(z)\frac{{\rm d}z}{2\pi i}\right|_{K}\Bigr)\subset \epsilon\cdot \mathcal{O}_{0;\epsilon}\otimes F_{1,-s;-n+2},\ \ l\neq 0.
\end{align}
Consider the commutators
$
\lbrack \widetilde{P}(z), L^{(\alpha_0)}_l\rbrack\ (l\in \mathbb{Z})
$ on ${\rm e}^{\kappa_1}F_{1,s;-n}$,  where we regard ${\rm e}^{\kappa_1}F_{1,s;-n}$ as a $\pi^{(\alpha_0)}(\mathfrak{Vir})$-module.
Note that the conformal weights of $Y(\ket{\widetilde{\alpha}_0},x)$ and $\widetilde{Q}_-(x)$ with respect to $T^{(\alpha_0)}$ are $h^{({\alpha}_0)}_{\widetilde{\alpha}_0}$ and $h^{({\alpha}_0)}_{\widetilde{\alpha}_-}$, respectively.
Then, we have the relations
\begin{align*}
\lbrack L^{(\alpha_0)}_l, Y(\ket{\widetilde{\alpha}_0},x)\rbrack&=x^l\bigl(x\partial_{x}+h^{({\alpha}_0)}_{\widetilde{\alpha}_0}(l+1)\bigr)Y(\ket{\widetilde{\alpha}_0},x),\\
\lbrack L^{(\alpha_0)}_l, \widetilde{Q}_-(x)\rbrack&=x^l\bigl(x\partial_{x}+h^{({\alpha}_0)}_{\widetilde{\alpha}_-}(l+1)\bigr)\widetilde{Q}_-(x).
\end{align*}
Using the above relations and the generalized Stokes' theorem, we have
\begin{align*}
\begin{aligned}
&\lbrack L^{({\alpha}_0)}_l,\widetilde{P}(z)\rbrack\\
&=z^{s-1}\int_{[\Delta^{(\widetilde{\alpha}_-;\kappa_2)}_{s-1}]} z^{l}\bigl\{\bigl(z\partial_{z}+h^{({\alpha}_0)}_{\widetilde{\alpha}_0}(l+1)\bigr)Y(\ket{\widetilde{\alpha}_0},z)\bigr\}\\
&\qquad \qquad \qquad \qquad \qquad \qquad \qquad\cdot \widetilde{Q}_-(zy_1)\cdots \widetilde{Q}_-(zy_{s-1}){\rm d}\bm{y}\\
&+z^{s-1}\sum_{i=1}^{s-1}\int_{[\Delta^{(\widetilde{\alpha}_-;\kappa_2)}_{s-1}]} Y(\ket{\widetilde{\alpha}_0},z)\widetilde{Q}_-(zy_1)\\
&\qquad \cdots z^ly^l_i\bigl\{\bigl(y_i\partial_{y_i}+h^{({\alpha}_0)}_{\widetilde{\alpha}_-}(l+1)\bigr)\widetilde{Q}_-(zy_i)\bigr\}\cdots \widetilde{Q}_-(zy_{s-1}){\rm d}\bm{y}\\
&=z^{s-1}\int_{[\Delta^{(\widetilde{\alpha}_-;\kappa_2)}_{s-1}]} \bigl(z^{l+1}\partial_{z}Y(\ket{\widetilde{\alpha}_0},z)\bigr)\widetilde{Q}_-(zy_1)\cdots \widetilde{Q}_-(zy_{s-1}){\rm d}\bm{y}\\
&\qquad +h^{({\alpha}_0)}_{\widetilde{\alpha}_0}(l+1)z^{s+l-1}\int_{[\Delta^{(\widetilde{\alpha}_-;\kappa_2)}_{s-1}]} Y(\ket{\widetilde{\alpha}_0},z)\bigr)\widetilde{Q}_-(zy_1)\cdots \widetilde{Q}_-(zy_{s-1}){\rm d}\bm{y}\\
&\qquad+z^{s-1}\sum_{i=1}^{s-1}\int_{[\Delta^{(\widetilde{\alpha}_-;\kappa_2)}_{s-1}]} Y(\ket{\widetilde{\alpha}_0},z)\widetilde{Q}_-(zy_1)\\
&\qquad \qquad\cdots z^ly^l_i\bigl\{(h^{({\alpha}_0)}_{\widetilde{\alpha}_-}-1)(l+1)\widetilde{Q}_-(zy_i)\bigr\}\cdots \widetilde{Q}_-(zy_{s-1}){\rm d}\bm{y}.
\end{aligned}
\end{align*}
Using the relation
\begin{align*}
\partial_{z}\widetilde{Q}_{-}(zy_j)=\frac{y_j}{z}\partial_{y_j}\widetilde{Q}_{-}(zy_j).
\end{align*}
and the generalized Stokes' theorem, we have
\begin{align*}
\begin{aligned}
&z^{s-1}\int_{[\Delta^{(\widetilde{\alpha}_-;\kappa_2)}_{s-1}]} \bigl(z^{l+1}\partial_{z}Y(\ket{\widetilde{\alpha}_0},z)\bigr)\widetilde{Q}_-(zy_1)\cdots \widetilde{Q}_-(zy_{s-1}){\rm d}\bm{y}=z^{l+1}\partial_{z}\widetilde{P}(z).
\end{aligned}
\end{align*}
Then, we obtain
\begin{align}
\begin{aligned}
\lbrack L^{({\alpha}_0)}_l,\widetilde{P}(z)\rbrack=&z^{l+1}\partial_{z}\widetilde{P}(z)+(l+1)z^l\bigl\{h^{({\alpha}_0)}_{\widetilde{\alpha}_0} \widetilde{P}(z)+(h^{({\alpha}_0)}_{\widetilde{\alpha}_-}-1)\widetilde{P}_{l}(z)\bigr\},
\end{aligned}
\label{lem:eq-ss}
\end{align}
where we define the field $\widetilde{P}_l(z)$ by
\begin{align*}
z^{s-1}\sum_{i=1}^{s-1}\int_{[\Delta^{(\widetilde{\alpha}_-;\kappa_2)}_{s-1}]} y^l_iY(\ket{\widetilde{\alpha}_0},z)\widetilde{Q}_-(zy_1)\cdots \widetilde{Q}_-(zy_i)\cdots \widetilde{Q}_-(zy_{s-1}){\rm d}\bm{y}.
\end{align*}
Note that 
\begin{align}
\label{eq:ep-ex}
&h^{({\alpha}_0)}_{\widetilde{\alpha}_0}=\frac{1}{2}\alpha_0\Bigl(1+\frac{p_+}{p_-}\Bigr)\epsilon+o(\epsilon),
&h^{({\alpha}_0)}_{\widetilde{\alpha}_-}-1=\frac{1}{2}\alpha_-\Bigl(1+\frac{p_+}{p_-}\Bigr)\epsilon+o(\epsilon).
\end{align}
By definition, we have $\widetilde{P}_0(z)=(s-1)\widetilde{P}(z)$, and the $m$-th mode of $z^{\kappa_3}\widetilde{P}_l(z)$ coincides with that of $z^{\kappa_3}\widetilde{P}(z)$ up to a scalar in $\mathcal{O}_{0;\epsilon}$.
We set
\begin{align*}
\begin{aligned}
\mathcal{X}_l(\epsilon)&:=\frac{1}{\epsilon}Q^{[p_+-1]}_+\oint_{z=0}z^{\kappa_3}\biggl\lbrack L^{(\alpha_0)}_{l},\frac{\widetilde{P}(z)}{z}\biggr\rbrack\frac{{\rm d}z}{2\pi i},\\
\mathcal{Y}_l(\epsilon)&:=\frac{1}{\epsilon}Q^{[p_+-1]}_+\oint_{z=0}z^{\kappa_3}z^{l}\partial_{z}\widetilde{P}(z)\frac{{\rm d}z}{2\pi i},\\
\mathcal{Z}_l(\epsilon)&:=\mathcal{X}_l(\epsilon)-\mathcal{Y}_l(\epsilon)
\end{aligned}
\end{align*}
for $l\in \mathbb{Z}$.
By (\ref{lem:eq-x-cr}), (\ref{eq:ptil-2}), (\ref{lem:eq-ss}), and (\ref{eq:ep-ex}), we see that the images of $\mathcal{X}_l(\epsilon)|_{K}$ and $\mathcal{Y}_l(\epsilon)|_{K}$ are contained in $\mathcal{O}_{0;\epsilon}\otimes F_{1,-s;-n+2}$.
Then, we define the linear operators 
\begin{align*}
\mathcal{X}_l,\mathcal{Y}_l,\mathcal{Z}_l\in {\rm Hom}(F_{1,s;-n}\cap {\rm ker}Q_+,F_{1,-s;-n+2}),\qquad  l\in \mathbb{Z}
\end{align*}
as follows
\begin{align*}
\mathcal{X}_l:=\left.\left.\mathcal{X}_l(\epsilon)\right|_{K}\right|_{\epsilon=0},\qquad
\mathcal{Y}_l:=\left.\left.\mathcal{Y}_l(\epsilon)\right|_{K}\right|_{\epsilon=0},\qquad 
\mathcal{Z}_l:=\mathcal{X}_l-\mathcal{Y}_l.
\end{align*}
Using (\ref{lem:eq-x-cr}), (\ref{eq:ptil-2}), and (\ref{lem:eq-ss}), we have
\begin{align}
\label{deri-theta2}
\begin{aligned}
&\mathcal{Z}_l=
\begin{cases}
s^{-1}C \cdot Q^{[p_+-1]}_+\circ R_1|_{F_{1,s;-n}\cap {\rm ker}Q_+}&l=0\\
0&l\neq 0,
\end{cases}
\end{aligned}
\end{align}
where we denote
\begin{align*}
C=\left.\frac{d}{d{\epsilon}}\bigl\{h^{({\alpha}_0)}_{\widetilde{\alpha}_0}+(s-1)(h^{({\alpha}_0)}_{\widetilde{\alpha}_-}-1)\bigr\}\right|_{\epsilon=0}.
\end{align*}
We see that $C\neq 0$. In fact, from (\ref{eq:ep-ex}),
\begin{align*}
C
&=\frac{1}{2}\Bigl(1+\frac{p_+}{p_-}\Bigr)\alpha_-\bigl(-\frac{p_-}{p_+}+s\bigr).
\end{align*}
\item 
Consider $2l\mathcal{X}_0(\epsilon)|_{K}$ for $l\in \mathbb{Z}\setminus\{0\}$.
Using the commutation relation
\begin{align*}
\lbrack L^{(\alpha_0)}_l, L^{(\alpha_0)}_{-l}\rbrack=2lL^{(\alpha_0)}_0+\frac{c_{\alpha_0}}{12}(l^3-l),\qquad l\neq 0,
\end{align*}
we have
\begin{align*}
&2l\mathcal{X}_0(\epsilon)|_{K}\\
&=[\mathcal{X}_l(\epsilon),L^{(\alpha_0)}_{-l}]|_{K}+[L^{(\alpha_0)}_{l}, \mathcal{X}_{-l}(\epsilon)]|_{K}\\
&=[\mathcal{Y}_l(\epsilon),L^{(\alpha_0)}_{-l}]|_{K}+[L^{(\alpha_0)}_{l}, \mathcal{Y}_{-l}(\epsilon)]|_{K}+ [\mathcal{Z}_l(\epsilon),L^{(\alpha_0)}_{-l}]|_{K}+[L^{(\alpha_0)}_{l}, \mathcal{Z}_{-l}(\epsilon)]|_{K}.
\end{align*}
From the definition of $\mathcal{Z}_l(\epsilon)$, (\ref{eq:ptil-2}), and (\ref{deri-theta2}), we see that
\begin{align*}
[L^{(\alpha_0)}_{l}, \mathcal{Z}_{-l}(\epsilon)]|_{K}|_{\epsilon=0}=0,\qquad l\neq 0.
\end{align*}
Then, setting $\epsilon=0$ and using (\ref{deri-theta2}) for $2l\mathcal{X}_0(\epsilon)|_{K}$, we obtain
\begin{align}
\label{eq:y0z0}
2l(\mathcal{Y}_0+\mathcal{Z}_0)=[\mathcal{Y}_l(\epsilon),L^{(\alpha_0)}_{-l}]|_{K}|_{\epsilon=0}+[L^{(\alpha_0)}_{l}, \mathcal{Y}_{-l}(\epsilon)]|_{K}|_{\epsilon=0},\qquad l\neq 0.
\end{align}
Using (\ref{lem:eq-ss}), we have
\begin{align*}
&[L^{(\alpha_0)}_{l}, \mathcal{Y}_{-l}(\epsilon)]|_{K}\\
&=\frac{1}{\epsilon}Q^{[p_+-1]}_+\left.\oint_{z=0}z^{\kappa_3}z^{-l}\partial_{z}\bigl(z^{l+1}\partial_z\widetilde{P}(z)\bigr)\frac{{\rm d}z}{2\pi i}\right|_{K}\\
&+(l+1)\frac{1}{\epsilon}Q^{[p_+-1]}_+\left.\oint_{z=0}z^{\kappa_3}z^{-l}\partial_{z}\Bigl(z^{l}\bigl\{h^{({\alpha}_0)}_{\widetilde{\alpha}_0} \widetilde{P}(z)+(h^{({\alpha}_0)}_{\widetilde{\alpha}_-}-1)\widetilde{P}_{l}(z)\bigr\}\Bigr)\frac{{\rm d}z}{2\pi i}\right|_{K}
\end{align*}
From (\ref{eq:ptil-2}) and (\ref{eq:ep-ex}), we see that the image of the second operator is contained in $\epsilon\cdot \mathcal{O}_{0;\epsilon}\otimes F_{1,-s;-n+2}$. Then, we have
\begin{align*}
&[L^{(\alpha_0)}_{l}, \mathcal{Y}_{-l}(\epsilon)]|_{K}|_{\epsilon=0}\\
&=(l+1)\mathcal{Y}_{0}+\frac{1}{\epsilon}Q^{[p_+-1]}_+\left.\left.\oint_{z=0}z^{\kappa_3}z\partial^2_z\widetilde{P}(z)\frac{{\rm d}z}{2\pi i}\right|_{K}\right|_{\epsilon=0},\ \ l\neq 0.
\end{align*}
Using this relation, we obtain
\begin{align*}
2l\mathcal{Y}_0=[\mathcal{Y}_l(\epsilon),L^{(\alpha_0)}_{-l}]|_{K}|_{\epsilon=0}+[L^{(\alpha_0)}_{l}, \mathcal{Y}_{-l}(\epsilon)]|_{K}|_{\epsilon=0},\qquad l\neq 0.
\end{align*}
Then, from (\ref{eq:y0z0}), we obtain $\mathcal{Z}_0=0$. 
Therefore, since $C\neq 0$, we get (\ref{final-t1t2t3}) from (\ref{deri-theta2}).
\end{enumerate}
\end{proof}

\begin{proof}[Proof of Theorem \ref{thm:mainWpq}]
As in Theorem \ref{thm:mainW1p}, from  Proposition~\ref{sl2action2} and Theorems~\ref{G-hom} and~\ref{non-Gop}, it suffices to show that 
\begin{align}
\label{cru-ef-rel-pq}
[\mathbb{E},\mathbb{F}]|_{\mathcal{W}_{p_+,p_-}\cap F_{1,1}}=0.
\end{align}
Note that $F_{1,1}\subset \widetilde{F}_{1,1}$, and that the screening operators $\widetilde{Q}_\pm$ and the compositions $\widetilde{Q}_\pm \widetilde{Q}_\mp$ are well-defined on $\widetilde{F}_{1,1}$. 
Fix a Virasoro singular vector $\phi$ in $ F_{1,1}$.
Since $[\mathbb{E},\mathbb{F}]$ acts trivially on the singular vector of the $\mathcal{L}_0$-weight $(p_+-1)(p_--1)$, we may assume that the weight of $\phi$ is at least $h_{4p_+-1,1}$.
We set 
\begin{align}
v_\pm:=\lim_{\epsilon\rightarrow 0}\hspace{-1mm}{}^F\frac{1}{\epsilon}\widetilde{Q}_\pm\phi\in F_{\mp1,\pm 1}
\end{align}
In what follows, we use the notation
\begin{align*}
\bm{e}_1&={\rm e}^{\widetilde{\alpha}_{1+p_+,p_--1}-{\alpha}_{1,-1}}, &\bm{e}_2&={\rm e}^{\widetilde{\alpha}_{1,1-2p_-}-\widetilde{\alpha}_{1+p_+,1-p_-}}, &\bm{e}_3&={\rm e}^{\widetilde{\alpha}_{p_+-1,1-p_-}-\widetilde{\alpha}_{-1,1-2p_-}},\\
\bm{e}'_1&={\rm e}^{\widetilde{\alpha}_{p_+-1,1+p_+}-{\alpha}_{-1,1}}, &\bm{e}'_2&={\rm e}^{\widetilde{\alpha}_{1-2p_+,1}-\widetilde{\alpha}_{1-p_+,1+p_-}}, &\bm{e}'_3&={\rm e}^{\widetilde{\alpha}_{1-p_+,p_--1}-\widetilde{\alpha}_{1-2p_+,-1}}.
\end{align*}
Taking into account the argument for (\ref{eq:19193}), from Lemmas~\ref{lem:tw0}-\ref{lem:tw} and the relation $Q^{[p_+-1]}_+\circ Q_+=0$, we have
\begin{align}
\label{2026feq}
\begin{aligned}
\mathbb{E}\cdot \mathbb{F}\phi
&=\lim_{\epsilon\rightarrow 0}\hspace{-1mm}{}^F\frac{1}{\epsilon}\widetilde{Q}^{[p_+-1]}_+\circ \bm{e}_3\circ\widetilde{Q}_+\circ \bm{e}_2\circ {\rm e}^{\widetilde{\alpha}_{1+p_+,1-p_-}-\alpha_{p_++1,1-p_-}}\circ{Q}^{[p_--1]}_-v_-\\
&=\lim_{\epsilon\rightarrow 0}\hspace{-1mm}{}^F\frac{1}{\epsilon}\widetilde{Q}^{[p_+-1]}_+\circ \bm{e}_3\circ\widetilde{Q}_+\circ \bm{e}_2\circ \widetilde{Q}^{[p_--1]}_-\circ \bm{e}_1v_-\\
&\quad +\lim_{\epsilon\rightarrow 0}\hspace{-1mm}{}^F\frac{1}{\epsilon}\widetilde{Q}^{[p_+-1]}_+\circ \bm{e}_3\circ\widetilde{Q}_+\circ \bm{e}_2\\
&\qquad \qquad \cdot \bigl( {\rm e}^{\widetilde{\alpha}_{1+p_+,1-p_-}-\alpha_{p_++1,1-p_-}}\circ{Q}^{[p_--1]}_--\widetilde{Q}^{[p_--1]}_-\circ \bm{e}_1\bigr)v_-\\
&=\lim_{\epsilon\rightarrow 0}\hspace{-1mm}{}^F\frac{1}{\epsilon}\widetilde{Q}^{[p_+-1]}_+\circ \bm{e}_3\circ\widetilde{Q}_+\circ \bm{e}_2\circ \widetilde{Q}^{[p_--1]}_-\circ \bm{e}_1v_-\\
&\quad +{Q}^{[p_+-1]}_+\circ {Q}_+\lim_{\epsilon\rightarrow 0}\hspace{-1mm}{}^F\frac{1}{\epsilon}\bigl( {\rm e}^{\widetilde{\alpha}_{1+p_+,1-p_-}-\alpha_{p_++1,1-p_-}}\circ{Q}^{[p_--1]}_--\widetilde{Q}^{[p_--1]}_-\circ \bm{e}_1\bigr)v_-\\
&=\lim_{\epsilon\rightarrow 0}\hspace{-1mm}{}^F\frac{1}{\epsilon}\widetilde{Q}^{[p_+-1]}_+\circ \bm{e}_3\circ\widetilde{Q}_+\circ \bm{e}_2\circ \widetilde{Q}^{[p_--1]}_-\circ \bm{e}_1v_-.
\end{aligned}
\end{align}
In the same way, we have
\begin{align}
\label{2026feq-22}
\begin{aligned}
\mathbb{F}\cdot \mathbb{E}\phi
&=\lim_{\epsilon\rightarrow 0}\hspace{-1mm}{}^F\frac{1}{\epsilon}\widetilde{Q}^{[p_--1]}_-\circ \bm{e}'_3\circ\widetilde{Q}_-\circ \bm{e}'_2\circ \widetilde{Q}^{[p_+-1]}_+\circ \bm{e}'_1v_+.
\end{aligned}
\end{align}
By Proposition~\ref{Felder complex2}, Lemma~\ref{lem:tw}, and Theorem~\ref{non-Gop}, the limit $v_\pm$ can be written as the following form
\begin{align}
\label{eq:vpm}
v_\pm=\phi_{{\rm ker}Q_\mp}+\psi_{{\rm im}Q_\pm},
\end{align}
where $\phi_{{\rm ker}Q_\mp}$ is a subsingular vector in ${\rm ker}Q_\mp$ such that the image $Q^{[p_\pm-1]}_\pm\phi_{{\rm ker}Q_\mp}$ is a singular vector (see also Remark~\ref{rem-subsing}), and $\psi_{{\rm im}Q_\pm}$ is a vector contained in ${\rm im}Q_\pm$.
Since $\psi_{{\rm im}Q_\pm}\in {\rm ker}Q^{[p_\pm-1]}_\pm$ by Propositions~\ref{Felder complex2}, we have
\begin{align}
\label{eq:psi-lim}
\begin{aligned}
\widetilde{Q}^{[p_--1]}_-\circ \bm{e}_1\psi_{{\rm im}Q_-}&\in \epsilon\cdot\widetilde{F}_{1+p_+,1-p_-},\\
\widetilde{Q}^{[p_+-1]}_+\circ \bm{e}'_1\psi_{{\rm im}Q_+}&\in \epsilon\cdot\widetilde{F}_{1-p_+,1+p_-}.
\end{aligned}
\end{align}
Then, by using Lemma~\ref{lem:tw0}, (\ref{eq:vpm}), (\ref{eq:psi-lim}), and the relations $Q^{[p_\pm-1]}_\pm\circ Q_\pm=0$, the two limits (\ref{2026feq})-(\ref{2026feq-22}) can be written as follows
\begin{align}
\label{2026feq0}
\begin{aligned}
\mathbb{E}\cdot \mathbb{F}\phi
&=\lim_{\epsilon\rightarrow 0}\hspace{-1mm}{}^F\frac{1}{\epsilon}\widetilde{Q}^{[p_+-1]}_+\circ \bm{e}_3\circ\widetilde{Q}_+\circ \bm{e}_2\circ \widetilde{Q}^{[p_--1]}_-\circ \bm{e}_1\phi_{{\rm ker}Q_+},\\
\mathbb{F}\cdot \mathbb{E}\phi
&=\lim_{\epsilon\rightarrow 0}\hspace{-1mm}{}^F\frac{1}{\epsilon}\widetilde{Q}^{[p_--1]}_-\circ \bm{e}'_3\circ\widetilde{Q}_-\circ \bm{e}'_2\circ \widetilde{Q}^{[p_+-1]}_+\circ \bm{e}'_1\phi_{{\rm ker}Q_-}.
\end{aligned}
\end{align}
By applying Lemma \ref{lem:sc-cruc} to the expressions before taking the limit in (\ref{2026feq0}), we have
\begin{align}
\label{eq:cru-lim1}
\begin{aligned}
&\widetilde{Q}^{[p_+-1]}_+\circ \bm{e}_3\circ\widetilde{Q}_+\circ \bm{e}_2\circ \widetilde{Q}^{[p_--1]}_-\circ \bm{e}_1 \phi_{{\rm ker}Q_+}\\
&\qquad -\widetilde{Q}^{[p_+-1]}_+\widetilde{Q}^{[p_--1]}_-\circ {\rm e}^{\widetilde{\alpha}_{p_+-1,p_--1}-\widetilde{\alpha}_{-1,-1}}\circ \widetilde{Q}_+ {\rm e}^{\widetilde{\alpha}_{1,-1}-\alpha_{1,-1}}\phi_{{\rm ker}Q_+}\\
&\in \epsilon^2\cdot \widetilde{F}_{1-p_+,1-p_-}
\end{aligned}
\end{align}
and
\begin{align}
\label{eq:cru-lim2}
\begin{aligned}
&\widetilde{Q}^{[p_--1]}_-\circ \bm{e}'_3\circ\widetilde{Q}_-\circ \bm{e}'_2\circ \widetilde{Q}^{[p_+-1]}_+\circ \bm{e}'_1\phi_{{\rm ker}Q_-}\\
&\qquad -\widetilde{Q}^{[p_--1]}_-\widetilde{Q}^{[p_+-1]}_+\circ {\rm e}^{\widetilde{\alpha}_{p_+-1,p_--1}-\widetilde{\alpha}_{-1,-1}}\circ \widetilde{Q}_- {\rm e}^{\widetilde{\alpha}_{-1,1}-\alpha_{-1,1}}\phi_{{\rm ker}Q_-}\\
&\in \epsilon^2\cdot \widetilde{F}_{1-p_+,1-p_-}.
\end{aligned}
\end{align}
Note that 
\begin{align}
\label{eq:vec-2142}
\widetilde{Q}_\pm {\rm e}^{\widetilde{\alpha}_{\pm1,\mp1}-\alpha_{\pm1,\mp1}}\phi_{{\rm ker}Q_\pm}\in \epsilon\cdot \widetilde{F}_{\mp1,\mp 1}. 
\end{align}
Then, using Lemma~\ref{lem:tw0} and (\ref{2026feq0})-(\ref{eq:vec-2142}), we obtain
\begin{align}
\begin{aligned}
\mathbb{E}\cdot \mathbb{F}\phi&={Q}^{[p_+-1]}_+{Q}^{[p_--1]}_-\lim_{\epsilon\rightarrow 0}\hspace{-1mm}{}^F\frac{1}{\epsilon} \widetilde{Q}_+ {\rm e}^{\widetilde{\alpha}_{1,-1}-\alpha_{1,-1}}\phi_{{\rm ker}Q_+},\\
\mathbb{F}\cdot \mathbb{E}\phi&={Q}^{[p_--1]}_-{Q}^{[p_+-1]}_+\lim_{\epsilon\rightarrow 0}\hspace{-1mm}{}^F\frac{1}{\epsilon} \widetilde{Q}_- {\rm e}^{\widetilde{\alpha}_{-1,1}-\alpha_{-1,1}}\phi_{{\rm ker}Q_-}\\
&={Q}^{[p_+-1]}_+{Q}^{[p_--1]}_-\lim_{\epsilon\rightarrow 0}\hspace{-1mm}{}^F\frac{1}{\epsilon} \widetilde{Q}_- {\rm e}^{\widetilde{\alpha}_{-1,1}-\alpha_{-1,1}}\phi_{{\rm ker}Q_-},
\end{aligned}
\label{eq:ef-vec900}
\end{align}
where we used the commutativity $[{Q}^{[p_--1]}_-,{Q}^{[p_+-1]}_+]=0$ which follows from Proposition~\ref{prop:qr-qs}.

To show that $\mathbb{E}\cdot \mathbb{F}\phi$ and $\mathbb{F}\cdot \mathbb{E}\phi$ are equal, we consider the following vector contained in $\widetilde{F}_{1-p_+,1-p_-}$:
 \begin{align}
 \label{vec-900}
\begin{aligned}
\widetilde{Q}^{[p_+-1]}_+\widetilde{Q}^{[p_--1]}_-\circ {\rm e}^{\widetilde{\alpha}_{p_+-1,p_--1}-\widetilde{\alpha}_{-1,-1}}\circ  \widetilde{Q}_-\widetilde{Q}_+(\phi-\epsilon\cdot\psi^{{\rm pre}}_{{\rm im}Q_+}),
\end{aligned}
\end{align}
where we fixed a vector $\psi^{{\rm pre}}_{{\rm im}Q_+}\in F_{1,1}$ satisfying 
\begin{align*}
Q_+\psi^{{\rm pre}}_{{\rm im}Q_+}=\psi_{{\rm im}Q_+}.
\end{align*}
Note that, from Proposition~\ref{prop:tw} and the relation $Q^{[p_--1]}_-\circ Q_-=0$, the vector (\ref{vec-900}) is contained in $\epsilon^2\cdot\widetilde{F}_{1-p_+,1-p_-}$. 
By the definition of $\psi^{{\rm pre}}_{{\rm im}Q_+}$, we have
\begin{align*}
\widetilde{Q}_+\psi^{{\rm pre}}_{{\rm im}Q_+}-{\rm e}^{\widetilde{\alpha}_{-1,1}-{\alpha}_{-1,1}}\psi_{{\rm im}Q_+}\in \epsilon\cdot\widetilde{F}_{-1,1}.
\end{align*}
Then, from Proposition~\ref{prop:tw} and $Q^{[p_--1]}_-\circ Q_-=0$, we have
\begin{align}
\label{vec-000}
\begin{aligned}
&\widetilde{Q}^{[p_--1]}_-\circ {\rm e}^{\widetilde{\alpha}_{p_+-1,p_--1}-\widetilde{\alpha}_{-1,-1}}\circ  \widetilde{Q}_-\bigl(\widetilde{Q}_+\psi^{{\rm pre}}_{{\rm im}Q_+}-{\rm e}^{\widetilde{\alpha}_{-1,1}-{\alpha}_{-1,1}}\psi_{{\rm im}Q_+}\bigr)\\
&\qquad \in \epsilon^2\cdot\widetilde{F}_{p_+-1,1-p_-}
\end{aligned}
\end{align}
By the definition of $v_+$ and Proposition~\ref{prop:tw}, we have
\begin{align*}
\widetilde{Q}_+\phi-\epsilon\cdot {\rm e}^{\widetilde{\alpha}_{-1,1}-\alpha_{-1,1}}(\phi_{{\rm ker}Q_-}+\psi_{{\rm im}Q_+})\in \epsilon^2\cdot\widetilde{F}_{-1,1}.
\end{align*}
Then, from Proposition~\ref{prop:tw} and $Q^{[p_--1]}_-\circ Q_-=0$, we have
\begin{align}
\label{vec-00+}
\begin{aligned}
&\widetilde{Q}^{[p_--1]}_-\circ {\rm e}^{\widetilde{\alpha}_{p_+-1,p_--1}-\widetilde{\alpha}_{-1,-1}}\circ  \widetilde{Q}_-\bigl(\widetilde{Q}_+\phi-\epsilon\cdot {\rm e}^{\widetilde{\alpha}_{-1,1}-\alpha_{-1,1}}(\phi_{{\rm ker}Q_-}+\psi_{{\rm im}Q_+})\bigr)\\
&\qquad \in \epsilon^3\cdot\widetilde{F}_{p_+-1,1-p_-}.
\end{aligned}
\end{align}
Then, using (\ref{vec-000}), (\ref{vec-00+}), we have
\begin{align}
\label{vec-00+s}
\begin{aligned}
&\widetilde{Q}^{[p_--1]}_-\circ {\rm e}^{\widetilde{\alpha}_{p_+-1,p_--1}-\widetilde{\alpha}_{-1,-1}}\circ  \widetilde{Q}_-\bigl(\widetilde{Q}_+\phi-\epsilon\cdot\widetilde{Q}_+\psi^{{\rm pre}}_{{\rm im}Q_+}-\epsilon\cdot {\rm e}^{\widetilde{\alpha}_{-1,1}-\alpha_{-1,1}}\phi_{{\rm ker}Q_-}\bigr)\\
&\qquad \in \epsilon^3\cdot\widetilde{F}_{p_+-1,1-p_-}.
\end{aligned}
\end{align}
Thus, using (\ref{eq:ef-vec900}), (\ref{vec-00+s}), and Lemma~\ref{lem:tw0}, we obtain
\begin{align}
\label{eq:lim-X}
\lim_{\epsilon\rightarrow 0}\hspace{-1mm}{}^F\frac{1}{\epsilon^2}\widetilde{Q}^{[p_+-1]}_+\widetilde{Q}^{[p_--1]}_-\circ {\rm e}^{\widetilde{\alpha}_{p_+-1,p_--1}-\widetilde{\alpha}_{-1,-1}}\circ  \widetilde{Q}_-\widetilde{Q}_+(\phi-\epsilon\cdot\psi^{{\rm pre}}_{{\rm im}Q_+})=\mathbb{F}\cdot \mathbb{E}\phi.
\end{align}
It remains to show that the limit of the left-hand side is $\mathbb{E}\cdot\mathbb{F}\phi$.
Note that 
\begin{align*}
Q_-\psi_{{\rm im}Q_+}=Q_+\psi_{{\rm im}Q_-}.
\end{align*}
In fact, using $[\widetilde{Q}_+,\widetilde{Q}_-]=0$ and Lemma~\ref{lem:tw0} , we have
\begin{align*}
\begin{aligned}
Q_-\lim_{\epsilon\rightarrow 0}\hspace{-1mm}{}^F\frac{1}{\epsilon}\widetilde{Q}_+ \phi=\lim_{\epsilon\rightarrow 0}\hspace{-1mm}{}^F\frac{1}{\epsilon}\widetilde{Q}_-\widetilde{Q}_+ \phi=\lim_{\epsilon\rightarrow 0}\hspace{-1mm}{}^F\frac{1}{\epsilon}\widetilde{Q}_+\widetilde{Q}_- \phi=Q_+\lim_{\epsilon\rightarrow 0}\hspace{-1mm}{}^F\frac{1}{\epsilon}\widetilde{Q}_- \phi.
\end{aligned}
\end{align*}
Then, using $[{Q}_+,{Q}_-]=0$, we have 
\begin{align*}
Q_-(\psi^{{\rm pre}}_{{\rm im}Q_+})-\psi_{{\rm im}Q_-}\in {\rm ker}Q_+.
\end{align*}
Thus, since $\psi_{{\rm im}Q_-}\in {\rm ker}Q^{[p_--1]}_-$, we obtain
\begin{align}
\label{vec-im}
 Q_-(\psi^{{\rm pre}}_{{\rm im}Q_+})-\psi_{{\rm im}Q_-}\in {\rm ker}Q_+\cap {\rm ker}Q^{[p_--1]}_-.
\end{align}
By (\ref{vec-im}) and Proposition~\ref{prop:tw}, we have
\begin{align}
\label{eq:q-pre}
\begin{aligned}
&\widetilde{Q}_-\psi^{{\rm pre}}_{{\rm im}Q_+}-{\rm e}^{\widetilde{\alpha}_{1,-1}-{\alpha}_{1,-1}}\psi_{{\rm im}Q_-}\\
&\qquad \in {\rm e}^{\widetilde{\alpha}_{1,-1}-{\alpha}_{1,-1}}{\rm ker}Q_+\cap {\rm ker}Q^{[p_--1]}_-\cap F_{1,-1}+\epsilon\cdot\widetilde{F}_{1,-1}.
\end{aligned}
\end{align}
From Propositions~\ref{prop:tw} and Theorems~\ref{G-hom}-\ref{non-Gop}, we have
\begin{align*}
\begin{aligned}
&\widetilde{Q}^{[p_+-1]}_+\circ {\rm e}^{\widetilde{\alpha}_{p_+-1,p_--1}-\widetilde{\alpha}_{-1,-1}}\circ  \widetilde{Q}_+\bigl( {\rm e}^{\widetilde{\alpha}_{1,-1}-{\alpha}_{1,-1}}{\rm ker}Q_+\cap {\rm ker}Q^{[p_--1]}_-\cap F_{1,-1}\bigr)\\
&\qquad \in \epsilon\cdot {\rm e}^{\widetilde{\alpha}_{1-p_+,p_+-1}-{\alpha}_{1-p_+,p_+-1}} {\rm ker}Q^{[p_--1]}_-\cap F_{1-p_+,p_+-1}+\epsilon^2\cdot\widetilde{F}_{1-p_+,p_+-1}.
\end{aligned}
\end{align*}
Then, by using this relation, (\ref{eq:q-pre}), and $Q^{[p_+-1]}_+\circ Q_+=0$, we have
\begin{align}
\label{vec-000-}
\begin{aligned}
&\widetilde{Q}^{[p_+-1]}_+\circ {\rm e}^{\widetilde{\alpha}_{p_+-1,p_--1}-\widetilde{\alpha}_{-1,-1}}\circ  \widetilde{Q}_+\bigl(\widetilde{Q}_-\psi^{{\rm pre}}_{{\rm im}Q_+}-{\rm e}^{\widetilde{\alpha}_{1,-1}-{\alpha}_{1,-1}}\psi_{{\rm im}Q_-}\bigr)\\
& \in \epsilon\cdot {\rm e}^{\widetilde{\alpha}_{1-p_+,p_+-1}-{\alpha}_{1-p_+,p_+-1}} {\rm ker}Q^{[p_--1]}_-\cap F_{1-p_+,p_+-1}+\epsilon^2\cdot\widetilde{F}_{1-p_+,p_+-1}.
\end{aligned}
\end{align}
As in (\ref{vec-00+}), we can show that
\begin{align}
\label{vec-00-}
\begin{aligned}
&\widetilde{Q}^{[p_+-1]}_+\circ {\rm e}^{\widetilde{\alpha}_{p_+-1,p_--1}-\widetilde{\alpha}_{-1,-1}}\circ  \widetilde{Q}_+\bigl(\widetilde{Q}_-\phi-\epsilon\cdot {\rm e}^{\widetilde{\alpha}_{1,-1}-\alpha_{1,-1}}(\phi_{{\rm ker}Q_+}+\psi_{{\rm im}Q_-})\bigr)\\
&\qquad \in \epsilon^3\cdot\widetilde{F}_{1-p_+,p_+-1}.
\end{aligned}
\end{align}
Then, using (\ref{vec-000-})-(\ref{vec-00-}), we have
\begin{align}
\label{vec-00-s}
\begin{aligned}
&\widetilde{Q}^{[p_+-1]}_+\circ {\rm e}^{\widetilde{\alpha}_{p_+-1,p_--1}-\widetilde{\alpha}_{-1,-1}}\circ  \widetilde{Q}_+\bigl(\widetilde{Q}_-\phi-\epsilon\cdot\widetilde{Q}_-\psi^{{\rm pre}}_{{\rm im}Q_+}-\epsilon\cdot {\rm e}^{\widetilde{\alpha}_{1,-1}-\alpha_{1,-1}}\phi_{{\rm ker}Q_+}\bigr)\\
&\qquad \in \epsilon^2\cdot {\rm e}^{\widetilde{\alpha}_{1-p_+,p_+-1}-{\alpha}_{1-p_+,p_+-1}} {\rm ker}Q^{[p_--1]}_-\cap F_{1-p_+,p_+-1}+\epsilon^3\cdot\widetilde{F}_{1-p_+,p_--1}.
\end{aligned}
\end{align}
Thus, using Lemma~\ref{lem:tw0}, (\ref{eq:ef-vec900}), and (\ref{vec-00-s}), we obtain
\begin{align}
\label{eq:lim-X-2}
\lim_{\epsilon\rightarrow 0}\hspace{-1mm}{}^F\frac{1}{\epsilon^2}\widetilde{Q}^{[p_--1]}_-\widetilde{Q}^{[p_+-1]}_+\circ {\rm e}^{\widetilde{\alpha}_{p_+-1,p_--1}-\widetilde{\alpha}_{-1,-1}}\circ  \widetilde{Q}_+\widetilde{Q}_-(\phi-\epsilon\cdot\psi^{{\rm pre}}_{{\rm im}Q_+})=\mathbb{E}\cdot \mathbb{F}\phi.
\end{align}
Then, applying the commutativities $[\widetilde{Q}_+,\widetilde{Q}_-]=0$, $[\widetilde{Q}^{[p_+-1]}_+,\widetilde{Q}^{[p_--1]}_-]=0$ to (\ref{eq:lim-X-2}) and comparing (\ref{eq:lim-X}) with (\ref{eq:lim-X-2}), we obtain $[\mathbb{E},\mathbb{F}]\phi=0$. Therefore, from Theorem~\ref{G-hom}, we get (\ref{cru-ef-rel-pq}).
\end{proof}

\section{Application to $\mathcal{SW}(m)$} 
\label{tri-swm}
In this section, we study derivations acting on the triplet superalgebra $\mathcal{SW}(m)$ by the same method as in the previous section.

\subsection{Free field realization of the Neveu-Scwarz algebra}
\label{FreeNS}
The $N=1$ Neveu-Schwarz algebra is the Lie superalgebra
\begin{equation*}
\mathfrak{ns}=\bigoplus_{n\in\mathbb{Z}}\mathbb{C} \mathcal{L}_n\oplus \bigoplus_{r\in\frac{1}{2}+\mathbb{Z}}\mathbb{C} \mathcal{G}_r\oplus \bigoplus\mathbb{C} C
\end{equation*}
with the relations $(k,l\in\mathbb{Z},\ r,s\in\mathbb{Z}+\frac{1}{2})$:
\begin{align*}
&[\mathcal{L}_k,\mathcal{L}_l]=(k-l)\mathcal{L}_{k+l}+\delta_{k+l,0}\frac{k^3-k}{12}C,\\
&[\mathcal{L}_k,\mathcal{G}_r]=(\frac{1}{2}k-r)\mathcal{G}_{k+r},\\
&\{\mathcal{G}_r,\mathcal{G}_s\}=2\mathcal{L}_{r+s}+\frac{1}{3}(r^2-\frac{1}{4})\delta_{r+s,0}C,\\
&[\mathcal{L}_k,C]=0,\ \ [\mathcal{G}_r,C]=0,
\end{align*}
where $\{,\}$ is the anti-commutator.
We identify $C$ with a scalar multiple of the identity, $C=c\cdot {\rm id}$, when acting on modules and refer to the number $c\in\mathbb{C}$ as the central charge. 
In this subsection, we review the free field realization of the Neveu-Scwarz algebra in accordance with the papers \cite{BMRW,IK2}.

The Neveu-Schwarz fermion algebra $\mathfrak{f}$ is the Lie superalgebra
\begin{align*}
\mathfrak{f}=\bigoplus_{r\in\mathbb{Z}+\frac{1}{2}}\mathbb{C} b_r\oplus\mathbb{C}\bold{1}
\end{align*}
with anti-commutation relations
\(
\{b_r,b_s\}=\delta_{r+s,0},
\)
\(\{b_r,\bold{1}\}=0.
\)
The Neveu-Schwarz fermion algebra $\mathfrak{f}$ has the triangular decomposition
\begin{align*}
\mathfrak{f}^{\pm}=\bigoplus_{r>0}\mathbb{C} b_r,\ \ \ \mathfrak{f}^0=\mathbb{C}\bold{1}.
\end{align*}
Let $\mathbb{C}{\mid}{\rm NS}\rangle$ be the one dimensional representation of $\mathfrak{f}^{\geq}=\mathfrak{f}^+\oplus\mathfrak{f}^0$ defined by
\begin{align*}
\bold{1}{\mid}{\rm NS}\rangle={\mid}{\rm NS}\rangle,\ \ \ \mathfrak{f}^+{\mid}{\rm NS}\rangle=0.\\
\end{align*}
The left {\rm Neveu-Schwarz fermionic Fock module} ${F}^{\mathfrak{f}}$ is defined by
\begin{eqnarray*}
{F}^{\mathfrak{f}}:={\rm Ind}^{\mathfrak{f}}_{\mathfrak{f}^{\geq}}\mathbb{C}{\mid}{\rm NS}\rangle.
\end{eqnarray*}
The right Neveu-Schwarz fermionic Fock module $F^{\mathfrak{f}\vee}$ is defined by
\begin{align*}
F^{\mathfrak{f}\vee}={\rm Ind}_{\mathfrak{f}^\leq}^{\mathfrak{f}}\mathbb{C}\langle{\rm NS}{\mid},
\end{align*} 
where $\mathbb{C}\langle{\rm NS}{\mid}$ is the one dimensional $(\mathfrak{f}^-\oplus \mathfrak{f}^0)$-module defined by
\begin{align*}
&\langle{\rm NS}{\mid}b_n=0,\qquad n< 0,
&\bm{1}\langle{\rm NS}{\mid}=\langle{\rm NS}{\mid}. 
\end{align*}
We see that the two Fock modules $F^{\mathfrak{f}},F^{\mathfrak{f}\vee}$ are equipped with an inner product
\begin{align}
\label{dual-f}
(\ \cdot\ ,\ \cdot\ )_{F^{\mathfrak{f}}}:\ F^{\mathfrak{f}\vee}\times F^{\mathfrak{f}}\rightarrow \mathbb{C}
\end{align}
defined by 
\begin{align*}
&(\bra{{\rm NS}},\ket{{\rm NS}})_{F^{\mathfrak{f}}}=1,\\
&(\bra{{\rm NS}}u_1,u_2\ket{{\rm NS}})_{F^{\mathfrak{f}}}=(\bra{{\rm NS}}u_1u_2,\ket{{\rm NS}})_{F^{\mathfrak{f}}}=(\bra{{\rm NS}},u_1u_2\ket{{\rm NS}})_{F^{\mathfrak{f}}}
\end{align*}
for $u_1,u_2\in U(\mathfrak{f})$, where $U(\mathfrak{f})$ is the universal enveloping superalgebra of $\mathfrak{f}$.
In what follows, we use the notation $\bra{{\rm NS}}u_1u_2\ket{{\rm NS}}=(\bra{{\rm NS}},u_1u_2\ket{{\rm NS}})_{F^{\mathfrak{f}}}$.

Let $b(z)=\sum_{n\in\mathbb{Z}+\frac{1}{2}}b_nz^{-n-\frac{1}{2}}$. Then this field satisfies the operator product expansion
\begin{align}
\label{eq:opebb}
b(z)b(w)=\frac{1}{z-w}+\cdots.
\end{align}
We define the following energy-momentum tensor
\begin{align*}
T^{(\mathfrak{f})}(z)=\frac{1}{2}:\partial b(z)b(z):=\sum_{n\in\mathbb{Z}}L^{(\mathfrak{f})}_nz^{-n-2}.
\end{align*}
The modes $\{L^{(\mathfrak{f})}_n\}_{n\in\mathbb{Z}}$ generate the Virasoso algebra with the central charge fixed to $\frac{1}{2}$.
By the energy-momentum tensor $T^{(\mathfrak{f})}(z)$, the Neveu-Schwarz fermionic Fock module ${F}^{\mathfrak{f}}$ becomes a Virasoro module with
\begin{eqnarray*}
L^{(\mathfrak{f})}_0{\mid}{\rm NS}\rangle=0,\ \ \ \ \ C{\mid}{\rm NS}\rangle=\frac{1}{2}{\mid}{\rm NS}\rangle.\\
\end{eqnarray*}

For $\tau\in \mathbb{C}$, we introduce an even field and an odd field
\begin{align*}
\begin{aligned}
&T^{(\tau;\mathfrak{ns})}(z)=T^{(\tau)}(z)\otimes\bold{1}+\bold{1}\otimes T^{(\mathfrak{f})}(z),\\
&G^{(\tau;\mathfrak{ns})}(z)=a(z)\otimes b(z)+\tau\bold{1}\otimes \partial b(z)
\end{aligned}
\end{align*}
acting on the products $F_{\beta}\otimes F^{\mathfrak{f}}$, $\beta\in \mathbb{C}$.
We see that
$T^{(\tau;\mathfrak{ns})}(z)$ and $G^{(\tau;\mathfrak{ns})}(z)$ satisfy the operator product expansions
\begin{align*}
\begin{aligned}
&T^{(\tau;\mathfrak{ns})}(z)T^{(\tau;\mathfrak{ns})}(w)=\frac{c^{\mathfrak{ns}}_{\tau}/2}{(z-w)^4}+\frac{2T^{(\tau;\mathfrak{ns})}(w)}{(z-w)^3}+\frac{\partial T^{(\tau;\mathfrak{ns})}(w)}{z-w}+\cdots,\\
&T^{(\tau;\mathfrak{ns})}(z)G^{(\tau;\mathfrak{ns})}(w)=\frac{\frac{3}{2}G^{(\tau;\mathfrak{ns})}(w)}{(z-w)^2}+\frac{\partial G^{(\tau;\mathfrak{ns})}(w)}{z-w}+\cdots,\\
&G^{(\tau;\mathfrak{ns})}(z)G^{(\tau;\mathfrak{ns})}(w)=\frac{2c^{\mathfrak{ns}}_{\tau}/{3}}{(z-w)^3}+\frac{2T^{(\tau;\mathfrak{ns})}(w)}{z-w}+\cdots,
\end{aligned}
\end{align*}
where $c^{\mathfrak{ns}}_{\tau}:=\frac{3}{2}-3\tau^2$.  Then, for the Fourier mode expansions of fields
\begin{align*}
T^{(\tau;\mathfrak{ns})}(z)=\sum_{n\in\mathbb{Z}}L^{(\tau;\mathfrak{ns})}_nz^{-n-2},\ \ \ G^{(\tau;\mathfrak{ns})}(w)=\sum_{r\in\mathbb{Z}+\frac{1}{2}}G^{(\tau;\mathfrak{ns})}_rz^{-r-\frac{3}{2}},
\end{align*}
the modes $\{L^{(\tau;\mathfrak{ns})}_n\}$ and $\{G^{(\tau;\mathfrak{ns})}_r\}$ define the following commutation and anti-commutation relations
\begin{align*}
&[L^{(\tau;\mathfrak{ns})}_k,L^{(\tau;\mathfrak{ns})}_l]=(k-l)L^{(\tau)}_{k+l}+\delta_{k+l,0}\frac{k^3-k}{12}c^{\mathfrak{ns}}_{\tau},\\
&[L^{(\tau;\mathfrak{ns})}_k,G^{(\tau;\mathfrak{ns})}_r]=(\frac{1}{2}k-r)G^{(\tau;\mathfrak{ns})}_{k+r},\\
&\{G^{(\tau;\mathfrak{ns})}_r,G^{(\tau;\mathfrak{ns})}_s\}=2L^{(\tau;\mathfrak{ns})}_{r+s}+\frac{1}{3}(r^2-\frac{1}{4})\delta_{r+s,0}c^{\mathfrak{ns}}_{\tau}.
\end{align*}
Thus the modes of the fields $T^{(\tau;\mathfrak{ns})}(z)$ and $G^{(\tau;\mathfrak{ns})}(z)$ generate the Neveu-Schwarz algebra with the central charge fixed to $c^{\mathfrak{ns}}_{\tau}=\frac{3}{2}-3\tau^2$.
Then we have the following proposition.
\begin{prop}
Define the following two vectors in $F_0\otimes F^{\mathfrak{f}}$
\begin{align*}
&T^{(\tau;\mathfrak{ns})}=\frac{1}{2}(a^2_{-1}+\tau a_{-2}+b_{-\frac{1}{2}}b_{-\frac{3}{2}}){\mid}0\rangle,
&G^{(\tau;\mathfrak{ns})}=(a_{-1}b_{-\frac{1}{2}}+\tau b_{-\frac{3}{2}}){\mid}0\rangle.
\end{align*}
Then the product $F_0\otimes F^{\mathfrak{f}}$ carries the structure of an $N=1$ Neveu-Schwarz vertex operator superalgebra, with
\begin{align*}
\begin{aligned}
Y({\mid}0\rangle,z)&={\rm id},&Y(a_{-1}{\mid}0\rangle,z)&=a(z),&Y(b_{-\frac{1}{2}}{\mid}0\rangle,z)&=b(z),\\
Y(G^{(\tau;\mathfrak{ns})},z)&=G^{(\tau;\mathfrak{ns})}(z),&Y(T^{(\tau;\mathfrak{ns})},z)&=T^{(\tau;\mathfrak{ns})}(z).
\end{aligned}
\end{align*}
\end{prop}
We denote by $\mathcal{F}^{\mathfrak{ns}}_{\tau}$ the above vertex operator superalgebra. Set
\begin{align*}
F^{\mathfrak{ns}}_{\beta}:={F}_{\beta}\otimes {F}^{\mathfrak{f}},\ \beta\in \mathbb{C}.
\end{align*}
From the above argument, this product admits the structure of a $\mathcal{F}^{\mathfrak{ns}}_{\tau}$-module.
We call this product $\mathfrak{ns}$-{\rm Fock module}.
When we regard $F^{\mathfrak{ns}}_{\beta}$ as a $\mathcal{F}^{\mathfrak{ns}}_{\tau}$-module, we write $F^{\mathfrak{ns}(\tau)}_{\beta}$.
The dual Fock space of $F^{\mathfrak{ns}}_{\beta}$ is defined by $F^{\mathfrak{ns}\vee}_{\beta}:=F^\vee_\beta\otimes F^{\mathfrak{f}}$. 
We see that the two Fock modules $F^{\mathfrak{ns}}_{\beta},F^{\mathfrak{ns}\vee}_{\beta}$ are equipped with an inner product
\begin{align*}
(\ \cdot\ ,\ \cdot\ )_{F^{\mathfrak{ns}}_{\beta}}:=(\ \cdot\ ,\ \cdot\ )_{F_{\beta}}(\ \cdot\ ,\ \cdot\ )_{F^{\mathfrak{f}}}
\end{align*}
(see (\ref{inn-pro}) and (\ref{dual-f})).

In what follows, we use the abbreviations
\begin{align*}
&{\mid}\beta\rangle={\mid}\beta\rangle\otimes {\mid}{\rm NS}\rangle,
&\langle\beta{\mid}=\langle\beta{\mid}\otimes\langle {\rm NS} {\mid}.
\end{align*}
In the following, we omit the tensor product in $a_n\otimes\bold{1}$ and $\bold{1}\otimes b_{r}$ and simply denote them as $a_n$ and $b_r$.
We also use the following shorthand notation
\begin{align*}
Y(\ket{\beta},z)=Y({\mid}\beta;{\rm B}\rangle,z)\otimes Y({\mid}{\rm NS}\rangle,z).
\end{align*}

\subsection{$\mathfrak{ns}$-screening operators}

Fix $\tau\in \mathbb{C}$. Let $\tau_\pm$ $(\Re \tau_+\geq \Re \tau_-)$ be the solutions of 
\begin{align}
\label{second-order-tau}
x^2-\tau x-1=0.
\end{align}
For $r,s\in\mathbb{Z}$, we set
\begin{align*}
&\beta^{(\tau)}_{r,s}:=\frac{1-r}{2}\tau_++\frac{1-s}{2}\tau_-,
&F^{\mathfrak{ns}(\tau)}_{r,s}:=F^{\mathfrak{ns}(\tau)}_{\beta^{(\tau)}_{r,s}}.
\end{align*}
Define two fields
\begin{align*}
&\mathcal{S}^{(\tau)}_+(z):=b(z)Y(\ket{\tau_+},z),
&\mathcal{S}^{(\tau)}_-(z):=b(z)Y(\ket{\tau_-},z).
\end{align*}
These fields satisfy 
the operator product expansions
\begin{align}
&T^{(\tau;\mathfrak{ns})}(z)\mathcal{S}^{(\tau)}_{\pm}(w)\sim \partial_{w}\frac{\mathcal{S}^{(\tau)}_{\pm}(w)}{z-w},
&G^{(\tau;\mathfrak{ns})}(z)\mathcal{S}^{(\tau)}_{\pm}(w)\sim \frac{1}{\tau_{\pm}}\partial_{w}\frac{Y(\ket{\tau_{\pm}},w)}{z-w}.
\label{sc}
\end{align}
By (\ref{sc}), the operators  
\begin{align*}
&\mathcal{S}^{(\tau)}_+:=\oint_{z=0}\mathcal{S}^{(\tau)}_+(z)\frac{{\rm d}z}{2\pi i}:F^{\mathfrak{ns}(\tau)}_{{1,2k+1}}\rightarrow F^{\mathfrak{ns}(\tau)}_{{-1,2k+1}}\ \ (k\in\mathbb{Z}),\\
&\mathcal{S}^{(\tau)}_-:=\oint_{z=0}\mathcal{S}^{(\tau)}_-(z)\frac{{\rm d}z}{2\pi i}:F^{\mathfrak{ns}(\tau)}_{{2k+1,1}}\rightarrow F^{\mathfrak{ns}(\tau)}_{{2k+1,-1}}\ \ (k\in\mathbb{Z})
\end{align*}
become commutative with the $\mathfrak{ns}$-action of $\mathcal{F}^{\mathfrak{ns}}_{\tau}\mathchar`-{\rm Mod}$.
These zero-modes of $\mathcal{S}^{(\tau)}_\pm(z)$ are called $\mathfrak{ns}$-screening operators. 
As in Proposition~\ref{prop:qr-qs}, we can show that the following lemma (cf. \cite{AdamovicD/MilasA:2009-2,NT}).
\begin{prop}
\label{prop:sscommu}
 On the Fock module $F^{\mathfrak{ns}(\tau)}_{1,1}$, we have
$
 \lbrack\mathcal{S}^{(\tau)}_+,\mathcal{S}^{(\tau)}_-\rbrack=0.
$
\end{prop}

Let $m\in \mathbb{Z}_{\geq 1}$ and $k\in \mathbb{Z}$. For $u\in F^{\mathfrak{ns}(\tau)}_{{2k,2m}}$ and $\psi^\vee\in F^{\mathfrak{ns}(\tau)\vee}_{{2k,-2m}}$, consider the correlation function
\begin{align}
\label{Phi-cor-tau}
\begin{aligned}
&\Phi^{\mathfrak{ns}}(\psi^\vee,u;\bm{z})=\bigl( \psi^\vee,\mathcal{S}^{(\tau)}_-(z_1)\mathcal{S}^{(\tau)}_-(z_2)\cdots \mathcal{S}^{(\tau)}_-(z_{2m})u\bigr)_{F^{\mathfrak{ns}(\tau)}_{{2k,-2m}}}.
\end{aligned}
\end{align}
From (\ref{eq:opeVV}) and (\ref{eq:opebb}), we see that
\begin{align}
\label{prod-boson-ns}
\begin{aligned}
\Phi^{\mathfrak{ns}}(\bra{\beta^{(\tau)}_{2k,-2m}},\ket{\beta^{(\tau)}_{2k,2m}};\bm{z})=&\prod_{i=1}^{2m}z^{(1-2m)\frac{\tau^2_-}{2}+\frac{2k-1}{2}}_i\prod_{1\leq i\neq j\leq 2m}(z_i-z_j)^{\frac{\tau^2_--1}{2}}\\
&\cdot \varDelta(z)\bra{{\rm NS}}b(z_1)b(z_2)\cdots b(z_{2m})\ket{{\rm NS}},
\end{aligned}
\end{align}
where $\varDelta(z):=\prod_{1\leq i< j\leq 2m}(z_i-z_j)$.
It is known that the factor 
\begin{align}
\label{eq:bb-cor}
\varDelta(z)\bra{{\rm NS}}b(z_1)b(z_2)\cdots b(z_{2m})\ket{{\rm NS}}
\end{align}
of (\ref{prod-boson-ns}) can be expressed by means of Jack symmetric polynomials \cite[Proposition 3.8]{BMRW}.
Thus, in a similar way to (\ref{Tsuchiya-Kanie0}), we can apply the theorems in Subsection~\ref{sub:sus} to (\ref{Phi-cor-tau}).
Then we can define the field
\begin{align*}
&\mathcal{S}^{(\tau);[2m]}_{-}(z)\in {\rm Hom}_{\mathbb{C}}(F^{\mathfrak{ns}(\tau)}_{{2k,2m}},F^{\mathfrak{ns}(\tau)}_{{2k,-2m}})[[z,z^{-1}]]
\end{align*}
as follows
\begin{equation}
\label{Tsuchiya-Kanie0-tau}
\begin{split}
&\mathcal{S}^{(\tau);[2m]}_{-}(z)=\int_{[\Delta^{(\tau_-)}_{2m-1;\frac{1}{2}}]}\mathcal{S}^{(\tau)}_-(z)\mathcal{S}^{(\tau)}_-(zy_{1})\mathcal{S}^{(\tau)}_-(zy_2)\cdots \mathcal{S}^{(\tau)}_-(zy_{2m-1})z^{2m-1}{\rm d}\bm{y},
\end{split}
\end{equation}
where ${\rm d}\bm{y}={\rm d}y_1\wedge \cdots\wedge {\rm d}y_{2m-1}$ and we set
\begin{align}
[\Delta^{(\tau_-)}_{2m-1;\frac{1}{2}}]
:=[\Delta_{2m-1}\bigl(\frac{(1-2m)\tau^2_-}{2}-\frac{1}{2},\tau^2_-,\frac{\tau^2_-}{2}-\frac{1}{2}\bigr)].
\label{eq:twist-rho-tau}
\end{align}
\begin{prop}[\cite{IK2}]
The zero mode
\begin{align*}
&\mathcal{S}^{(\tau);[2m]}_{-}:={\rm Res}_{z=0}\mathcal{S}^{(\tau);[2m]}_{-}(z){\rm d}z \in {\rm Hom}_{\mathbb{C}}(F^{\mathfrak{ns}(\tau)}_{{2k,2m}},F^{\mathfrak{ns}(\tau)}_{{2k,-2m}})
\end{align*}
commute with every $\mathfrak{ns}$-mode of $\mathcal{F}^{\mathfrak{ns}}_{\tau}\mathchar`-{\rm Mod}$. 
\label{com-IK}
\end{prop}
This zero mode is also called $\mathfrak{ns}$-screening operator. The nontriviality of $\mathcal{S}^{(\tau);[2m]}_{-}$ is ensured by the results in \cite{BMRW}.

\subsection{The triplet $W$-superalgebra $\mathcal{SW}(m)$}
Let $m\in \mathbb{Z}_{\geq 1}$. 
In \cite{AdamovicD/MilasA:2009-2}, the triplet $W$-superalgebra $\mathcal{SW}(m)$ is realized by using $\mathfrak{ns}$-Fock modules whose central charge is 
\begin{align*}
c^{\mathfrak{ns}}_{1,2m+1}:=\frac{15}{2}-3(2m+1+\frac{1}{2m+1}).
\end{align*}
Before introducing the construction of \(\mathcal{SW}(m)\), we introduce some notation.

In what follows, we assume $m\in \mathbb{Z}_{\geq 1}$. We set
\begin{align*}
\beta_{+}=\sqrt{2m+1},\qquad \beta_-=-\sqrt{\frac{1}{2m+1}},\qquad\beta_0=\beta_++\beta_-.
\end{align*}
Note that $\beta_\pm$ are the solutions of (\ref{second-order-tau}) as $\tau=\beta_0$, and that $c^{\mathfrak{ns}}_{\beta_0}=\frac{3}{2}-3\beta^2_0=c^{\mathfrak{ns}}_{1,2m+1}$.
For $r,s,n\in\mathbb{Z}$, we set
\begin{align*}
&\beta_{r,s;n}:=\beta^{(\beta_0)}_{r,s}+\frac{n}{2}\beta_+=\frac{1-r}{2}\beta_++\frac{1-s}{2}\beta_-+\frac{n}{2}\beta_+,
&\beta_{r,s}:=\beta_{r,s;0},
\end{align*}
and we use the shorthand notation $F^{\mathfrak{ns}}_{r,s;n}=F^{\mathfrak{ns}(\beta_0)}_{\beta_{r,s;n}}$ and $F^{\mathfrak{ns}}_{r,s}=F^{\mathfrak{ns}(\beta_0)}_{\beta_{r,s}}$.
For $r,s,n\in\mathbb{Z}$, we introduce the notation
\begin{equation*}
\begin{aligned}
&h^{\mathfrak{ns}}_{r,s}:=h^{(\beta_0)}_{\beta_{r,s}}=\frac{1}{8}(r^2-1)(2m+1)-\frac{1}{4}(rs-1)+\frac{1}{8}(s^2-1)\frac{1}{2m+1},\\
&h^{\mathfrak{ns}}_{r,s;n}:=h^{(\beta_0)}_{\beta_{r,s;n}}=h^{\mathfrak{ns}}_{r-n,s}=h^{\mathfrak{ns}}_{r,s+(2m+1)n}
\end{aligned}
\end{equation*}
and denote by $L^{\mathfrak{ns}}(h)$ the simple $\mathfrak{ns}$-module whose lowest weight and central charge are $h$ and $c^{\mathfrak{ns}}_{1,2m+1}$, respectively.
\begin{dfn}
Let $L=\mathbb{Z}\beta_+=\mathbb{Z}\sqrt{2m+1}$ be an integral lattice.
The {\rm lattice vertex operator superalgebra} $\mathcal{V}^{\mathfrak{ns}}_L$ is the quadruple
\begin{equation*}
\Bigl(\bigoplus_{\beta\in L}F^{\mathfrak{ns}(\beta_0)}_{\beta},{\mid}0\rangle,T^{(\beta_0;\mathfrak{ns})},G^{(\beta_0;\mathfrak{ns})},Y\Bigr)
\end{equation*}
where the fields corresponding to ${\mid}0\rangle$, $a_{-1}{\mid}0\rangle$, $b_{-\frac{1}{2}}{\mid}0\rangle$, $T^{(\beta_0;\mathfrak{ns})}$, $G^{(\beta_0;\mathfrak{ns})}$, and $Y$ are those of $\mathcal{F}^{\mathfrak{ns}}_{\beta_0}$.
\end{dfn}
Note that the two screening operators $\mathcal{S}^{(\beta_0)}_\pm$ act on $\mathcal{V}^{\mathfrak{ns}}_L$. 
Using the classification results in \cite{IK2}, we see that ${\rm ker}\ \mathcal{S}^{(\beta_0)}_-{\mid}_{\mathcal{V}^{\mathfrak{ns}}_{L}}$ satisfies the following decomposition as $\mathfrak{ns}$-modules
\begin{equation*}
\begin{split}
{\rm ker}\ \mathcal{S}^{(\beta_0)}_-{\mid}_{\mathcal{V}^{\mathfrak{ns}}_{L}}&\simeq \bigoplus_{n\in\mathbb{Z}_{\geq 0}}(2n+1)L^{\mathfrak{ns}}(h^{\mathfrak{ns}}_{1,1;-2n}).
\end{split}
\end{equation*}

\begin{prop}[\cite{AdamovicD/MilasA:2009-2}]
Let
$
\mathcal{SW}(m):={\rm ker}\ \mathcal{S}^{(\beta_0)}_-{\mid}_{\mathcal{V}^{\mathfrak{ns}}_{L}}.
$ 
Then $\mathcal{SW}(m)$ has the structure of an $N=1$ vertex operator superalgebra.
\end{prop}
This vertex operator superalgebra is called $N=1$ $triplet$ $vertex$ $operator$ $superalgebra$ or $N=1$ $triplet$ $W\mathchar`-superalgebra$.

We define the following three vectors in $\mathcal{V}_L$
\begin{align*}
W^{-;\mathfrak{ns}}:={\mid}\beta_{1,1;-2}\rangle,\ \ \ W^{0;\mathfrak{ns}}:=\mathcal{S}^{(\beta_0)}_+W^{-;\mathfrak{ns}},\ \ \ W^{+;\mathfrak{ns}}:=(\mathcal{S}^{(\beta_0)}_+)^2W^{-;\mathfrak{ns}}.
\end{align*}
These vectors have the same $\mathcal{L}_0$-weight $h^{\mathfrak{ns}}_{3,1}=2m+\frac{1}{2}$.
We define the following three vectors 
\begin{align*}
\widehat{W}^{-;\mathfrak{ns}}:=b_{-\frac{1}{2}}{\mid}\beta_{1,1;-2}\rangle,\ \ \ \widehat{W}^{0;\mathfrak{ns}}:=\mathcal{S}^{(\beta_0)}_+\widehat{W}^{-;\mathfrak{ns}},\ \ \ \widehat{W}^{+;\mathfrak{ns}}:=(\mathcal{S}^{(\beta_0)}_+)^2\widehat{W}^{-;\mathfrak{ns}}.
\end{align*}
These vectors have the same $\mathcal{L}_0$-weight $2m+1$.
\begin{prop}[{\cite{AdamovicD/MilasA:2009-2}}]
\label{AM-gen}
The triplet super $W$-algebra $\mathcal{SW}(m)$ is generated by $Y(W^{\pm;\mathfrak{ns}},z),Y(W^{0;\mathfrak{ns}},z),G^{(\beta_0;\mathfrak{ns})}(z)$. Furthermore $\mathcal{SW}(m)$ is strongly generated by
\begin{align*}
G^{(\beta_0;\mathfrak{ns})}(z),\ T^{(\beta_0;\mathfrak{ns})}(z),\ Y(W^{\pm;\mathfrak{ns}},z),\ Y(W^{0;\mathfrak{ns}},z),\ Y(\widehat{W}^{\pm;\mathfrak{ns}},z),\ Y(\widehat{W}^{0;\mathfrak{ns}},z).
\end{align*}
\end{prop}
For $n\geq 0$ and $-n \leq k\leq n$, we define 
\begin{equation*}
w^{(n);\mathfrak{ns}}_{k}:=
(\mathcal{S}^{(\beta_0)}_+)^{n+k}\ket{\beta_{1,1;-2n}}.
\end{equation*}
We see that the set $\{w^{(n);\mathfrak{ns}}_{k}\}_{k=-n}^{n}$ gives a basis of the minimal conformal weight spaces of $(2n+1)L^{\mathfrak{ns}}(h^{\mathfrak{ns}}_{1,1;-2n})\subset \mathcal{SW}(m)$.
Note that $w^{(1);\mathfrak{ns}}_{\pm1}$ and $w^{(1);\mathfrak{ns}}_0$ agree with $W^{\pm;\mathfrak{ns}}$ and $W^{0;\mathfrak{ns}}$ respectively, up to scalar multiples.
We set
\begin{align*}
W^{\delta;\mathfrak{ns}}[n]&=\oint_{z=0}Y(W^{\delta;\mathfrak{ns}},z)z^{2m+\frac{1}{2}+n-1}\frac{{\rm d}z}{2\pi i},\qquad \delta=\pm,0,\ \ n\in \mathbb{Z}+\frac{1}{2},\\
\widehat{W}^{\delta;\mathfrak{ns}}[n]&=\oint_{z=0}Y(\widehat{W}^{\delta;\mathfrak{ns}},z)z^{2m+1+n-1}\frac{{\rm d}z}{2\pi i},\qquad \delta=\pm,0,\ \ n\in \mathbb{Z},
\end{align*}
which are the $n$-th mode of the fields $Y(W^{\delta;\mathfrak{ns}},z)$ and $Y(\widehat{W}^{\delta;\mathfrak{ns}},z)$.
It was shown in \cite{AdamovicD/MilasA:2009-2} that $\mathcal{SW}(m)$ is simple. The following proposition is also proved in \cite[Theorem 9.1]{AdamovicD/MilasA:2009-2}.
\begin{prop}[\cite{AdamovicD/MilasA:2009-2}]
The fields $Y(W^{\pm;\mathfrak{ns}},z)$, $Y(W^{0;\mathfrak{ns}},z)$, $Y(\widehat{W}^{\pm;\mathfrak{ns}},z)$ and $Y(\widehat{W}^{0;\mathfrak{ns}},z)$ act on the vectors $w^{(n);\mathfrak{ns}}_{k}$ as follows:
The vectors $w^{(n);\mathfrak{ns}}_{k}$ $(n\geq 0,-n\leq k\leq n)$ satisfy
\begin{align*}
w^{(n+1);\mathfrak{ns}}_{k\pm 1}&\in \mathbb{C}^\times W^{\pm;\mathfrak{ns}}[h^{\mathfrak{ns}}_{1,1;-2n}-h^{\mathfrak{ns}}_{1,1;-2n-2}]w^{(n);\mathfrak{ns}}_{k} +\sum_{l=0}^{n}U(\mathfrak{ns})w^{(l);\mathfrak{ns}}_{k\pm 1},
\end{align*}
where $w^{(n);\mathfrak{ns}}_{n+1}=w^{(n);\mathfrak{ns}}_{-n-1}=0$ and $U(\mathfrak{ns})$ is the universal enveloping algebra of $\mathfrak{ns}$. 
\label{sl2action2-s}
\end{prop}

Let $A(\mathcal{SW}(m))$ be the Zhu-algebra of $\mathcal{SW}(m)$ (for the definition of Zhu-algebras, see \cite{KW},\cite{Zh}).
In \cite[Theorem 11.2]{AdamovicD/MilasA:2009-2}, the structure of the Zhu-algebra $A(\mathcal{SW}(m))$ is determined. 
\begin{thm}[\cite{AdamovicD/MilasA:2009-2}]
\label{AM-ns}
The following holds for the Zhu-algebra $A(\mathcal{SW}(m))$.
\begin{enumerate}
\item The Zhu-algebra $A(\mathcal{SW}(m))$ is generated by $[\widehat{W}^{\pm ;\mathfrak{ns}}]$, $[\widehat{W}^{0 ;\mathfrak{ns}}]$ and $[T^{(\beta_0;\mathfrak{ns})}]$.
\item In the Zhu-algebra, we have
\begin{align*}
\begin{aligned}
&[\widehat{W}^{0 ;\mathfrak{ns}}]*[\widehat{W}^{\pm ;\mathfrak{ns}}]-[\widehat{W}^{\pm ;\mathfrak{ns}}]*[\widehat{W}^{0 ;\mathfrak{ns}}]=\pm2q([T^{(\beta_0;\mathfrak{ns})}])[\widehat{W}^{\pm ;\mathfrak{ns}}],\\
&[\widehat{W}^{+ ;\mathfrak{ns}}]*[\widehat{W}^{- ;\mathfrak{ns}}]-[\widehat{W}^{- ;\mathfrak{ns}}]*[\widehat{W}^{+ ;\mathfrak{ns}}]=-2q([T^{(\beta_0;\mathfrak{ns})}])[\widehat{W}^{0 ;\mathfrak{ns}}],
\end{aligned}
\end{align*}
where $q(x)$ is a certain polynomial
\item $q([T^{(\beta_0;\mathfrak{ns})}])$ is a unit in $A(\mathcal{SW}(m))$.
\end{enumerate}
In particular, $A(\mathcal{SW}(m))$ contains a Lie subalgebra isomorphic to $\mathfrak{sl}_2(\mathbb{C})$.
\end{thm}

\subsection{Derivations on $\mathcal{SW}(m)$} 
\label{subs:swm}
As in Proposition \ref{Felder complex2}, the following proposition holds.
\begin{prop}[\cite{IK2}]
\label{prop:IK-exact}
The $\mathfrak{ns}$-screening operators $\mathcal{S}^{(\beta_0)}_-$ and $\mathcal{S}^{(\beta_0);[2m]}_-$ define the complex
\begin{align*}
\cdots\xrightarrow{\mathcal{S}^{(\beta_0)}_-}F^{\mathfrak{ns}}_{1,2m;2k+1}\xrightarrow{\mathcal{S}^{(\beta_0);[2m]}_-}F^{\mathfrak{ns}}_{1,1;2k}\xrightarrow{\mathcal{S}^{(\beta_0)}_-}F^{\mathfrak{ns}}_{1,2m;2k-1}\xrightarrow{\mathcal{S}^{(\beta_0);[2m]}_-}\cdots
\end{align*}
and this complex is exact.
\end{prop}
Thus, in a similar way to Subsection \ref{cons-main0}, we can consider an $\epsilon$-lifting of the above complex.
To show this, let us introduce some notation.
We set
\begin{align}
\label{notep1-ns}
\widetilde{\beta}_+(\epsilon)=\beta_++\epsilon,\qquad
\widetilde{\beta}_-(\epsilon)=-\frac{1}{\widetilde{\beta}_+(\epsilon)},\qquad
\widetilde{\beta}_0(\epsilon)=\widetilde{\beta}_+(\epsilon)+\widetilde{\beta}_-(\epsilon),
\end{align}
for $\epsilon \in \mathbb{C}$.
Note that $\widetilde{\beta}_\pm(0)=\beta_\pm$ and $\widetilde{\beta}_0(0)=\beta_0$. The elements $\widetilde{\beta}_\pm(\epsilon)$ can be regarded as $\epsilon$-deformations of $\beta_\pm$ that preserve the relation $\beta_+\beta_-=-1$.
To simplify the notation, we will omit $(\epsilon)$ from each notation of (\ref{notep1-ns}).
We introduce the notation
\begin{align*}
\begin{aligned}
\widetilde{\beta}_{r,s}&:=\beta^{(\widetilde{\beta}_{0})}_{r,s},&\widetilde{F}^{\mathfrak{ns}}_{r,s}&:=\widetilde{F}^{(\widetilde{\beta}_0)}_{\widetilde{\beta}_{r,s}} \otimes{F}^{\mathfrak{f}},&\widetilde{\mathcal{S}}^{[2m]}_-&:=\mathcal{S}^{(\widetilde{\beta}_0);[2m]}_-,&\widetilde{\mathcal{S}}_\pm&:=\mathcal{S}^{(\widetilde{\beta}_0)}_\pm,
\end{aligned}
\end{align*}
and through the action of $T^{(\widetilde{\beta}_0;\mathfrak{ns})}$ and $G^{(\widetilde{\beta}_0;\mathfrak{ns})}$, we naturally regard $\widetilde{F}^{\mathfrak{ns}}_{r,s}$ as an $\mathfrak{ns}$-module with central charge $c^{\mathfrak{ns}}_{\widetilde{\beta}_0}$ (for the notation $\widetilde{F}^{(\widetilde{\beta}_0)}_{\widetilde{\beta}_{r,s}}$, see Subsection~\ref{cons-main0}).
We can regard $\widetilde{F}^{\mathfrak{ns}}_{r,s}$ as an $\epsilon$-deformation of $F^{\mathfrak{ns}}_{r,s}$.
As in (\ref{Q+Q}), we define a linear operator
\begin{align}
\label{Q+Q-ns}
\begin{aligned}
&N^{\mathfrak{ns}}_{-,k}(\epsilon):=\widetilde{\mathcal{S}}^{[2m]}_-\circ{\rm e}^{\widetilde{\beta}_{2k+2,2m}-\widetilde{\beta}_{2k+1,-1}}\circ\widetilde{\mathcal{S}}_-\circ{\rm e}^{\widetilde{\beta}_{2k+1,1}-\beta_{2k+1,1}}.
\end{aligned}
\end{align}
Note that the presence of the shift elements ${\rm e}^\bullet$ in (\ref{Q+Q-ns}) ensures that the actions of the deformed screening operators are well-defined.
By definition, we have
\begin{align}
\label{S-holom0}
\begin{aligned}
&\widetilde{\mathcal{S}}_-(\widetilde{F}^{\mathfrak{ns}}_{2k+1,1})\subset \widetilde{F}^{\mathfrak{ns}}_{2k+1,-1},\qquad 
\widetilde{\mathcal{S}}_+(\widetilde{F}^{\mathfrak{ns}}_{1,2k+1})\subset \widetilde{F}^{\mathfrak{ns}}_{-1,2k+1},\\
&\widetilde{\mathcal{S}}_-\circ {\rm e}^{\widetilde{\beta}_{2k+1,1}-{\beta}_{2k+1,1}}|_{{F}^{\mathfrak{ns}}_{2k+1,1}}|_{\epsilon=0}={\mathcal{S}}^{(\beta_0)}_-|_{{F}^{\mathfrak{ns}}_{2k+1,1}},\\
&\widetilde{\mathcal{S}}_+\circ {\rm e}^{\widetilde{\beta}_{1,2k+1}-{\beta}_{1,2k+1}}|_{{F}^{\mathfrak{ns}}_{1,2k+1}}|_{\epsilon=0}={\mathcal{S}}^{(\beta_0)}_+|_{{F}^{\mathfrak{ns}}_{1,2k+1}}
\end{aligned}
\end{align}
for $k\in \mathbb{Z}$.
Note that for sufficiently small $|\epsilon|\geq 0$, the parameters defining the twisted cycle $[\Delta^{(\widetilde{\beta}_-)}_{2m-1;\frac{1}{2}}]$ satisfy
\begin{align*}
\Bigl(\frac{(1-2m)\widetilde{\beta}^2_-}{2}-\frac{1}{2},\widetilde{\beta}^2_-,\frac{\widetilde{\beta}^2_-}{2}-\frac{1}{2}\Bigr)\notin \widehat{\mathcal{A}}_{2m-1}
\end{align*}
(for the definition of $[\Delta^{(\tau_-)}_{2m-1;\frac{1}{2}}]$, see (\ref{eq:twist-rho-tau})).
Then, as in Proposition \ref{prop:tw}, using Theorem \ref{sus-prop} and the fact that the factor (\ref{eq:bb-cor}) can be expressed in terms of Jack symmetric polynomials \cite[Proposition 3.8]{BMRW}, we obtain 
\begin{align}
\label{S-holom}
\begin{aligned}
&\widetilde{\mathcal{S}}^{[2m]}_-(\widetilde{F}^{\mathfrak{ns}}_{2k,2m})\subset \widetilde{F}^{\mathfrak{ns}}_{2k,-2m},\\
&\widetilde{\mathcal{S}}^{[2m]}_-\circ {\rm e}^{\widetilde{\beta}_{2k,2m}-{\beta}_{2k,2m}}|_{{F}^{\mathfrak{ns}}_{2k,2m}}|_{\epsilon=0}={\mathcal{S}}^{(\beta_0);[2m]}_-|_{{F}^{\mathfrak{ns}}_{2k,2m}}
\end{aligned}
\end{align}
for $k\in \mathbb{Z}$.
Then, by (\ref{S-holom0})-(\ref{S-holom}) and $\mathcal{S}^{(\beta_0);[2m]}_-\circ \mathcal{S}^{(\beta_0)}_-=0$, we get
\begin{align*}
\begin{aligned}
&{\rm e}^{\beta_{2k+3,1}-\widetilde{\beta}_{2k+2,-2m}}\circ N^{\mathfrak{ns}}_{-,k}(\epsilon)\in {\rm Hom}_{\mathbb{C}}( F^{\mathfrak{ns}}_{2k+1,1},\epsilon\cdot \mathcal{O}_{0;\epsilon}\otimes F^{\mathfrak{ns}}_{2k+1,1;-2}).
\end{aligned}
\end{align*}
Thus we can define the linear operator
\begin{align*}
G^{\mathfrak{ns}}_{-,k}:=\lim_{\epsilon\rightarrow 0}\hspace{-1mm}{}^F \frac{1}{\epsilon}N^{\mathfrak{ns}}_{-,k}(\epsilon)\in {\rm Hom}_{\mathbb{C}}( F^{\mathfrak{ns}}_{2k+1,1},F^{\mathfrak{ns}}_{2k+1,1;-2}),
\end{align*}
where the limit is naturally defined by replacing the bosonic Fock modules with $\mathfrak{ns}$-Fock modules in Definition~\ref{notlim}.
As in Lemma~\ref{lem:tw}, we can show that
\begin{align}
\label{G-iden-ns0}
\left.G^{\mathfrak{ns}}_{-,k}\right|_{{\rm ker}\mathcal{S}^{(\beta_0)}_-}=\left.\mathcal{S}^{(\beta_0);[2m]}_-\lim_{\epsilon\rightarrow 0}\hspace{-1mm}{}^F \frac{1}{\epsilon}\widetilde{\mathcal{S}}_-\circ{\rm e}^{\widetilde{\beta}_{2k+1,1}-\beta_{2k+1,1}}\right|_{F^{\mathfrak{ns}}_{2k+1,1}\cap {\rm ker}\mathcal{S}^{(\beta_0)}_-}.
\end{align}
Then, as in the discussion above Theorem~\ref{G+G-prop}, using (\ref{G-iden-ns0}), we obtain
\begin{align}
\label{G-iden-ns}
\begin{aligned}
\left.G^{\mathfrak{ns}}_{-,k}\right|_{{\rm ker}\mathcal{S}^{(\beta_0)}_-}&= \frac{1}{2m+1}\mathcal{S}^{(\beta_0);[2m]}_-\oint_{z=0} :\mathcal{S}^{(\beta_0)}_-(z)\phi_0(z):\frac{{\rm d}z}{2\pi i}.
\end{aligned}
\end{align}
Thus, we have the following proposition.
\begin{prop}
\label{G+G-prop-ns}
Define the linear operator $\mathbb{G}^{\mathfrak{ns}}_-\in {\rm End}_{\mathbb{C}}(\mathcal{V}^{\mathfrak{ns}}_L)$ by the right-hand side of (\ref{G-iden-ns}). Then, we have
\begin{align*}
&\left.\mathbb{G}^{\mathfrak{ns}}_-\right|_{F^{\mathfrak{ns}}_{2k+1,1}\cap \mathcal{SW}(m)}=\left.G^{\mathfrak{ns}}_{-,k}\right|_{\mathcal{SW}(m)}.
\end{align*}
\end{prop}
As in Theorem~\ref{G-hom}, using Proposition~\ref{prop:IK-exact}, (\ref{S-holom0}), and (\ref{S-holom}), we can show that $\mathbb{G}^{\mathfrak{ns}}_-|_{\mathcal{SW}(m)}$ commutes with the $U(\mathfrak{ns})$-action of $\mathcal{SW}(m)$.
Further, using the results for the Jantzen filtration of the Fock modules $F^{\mathfrak{ns}}_{r,s}$ determined in \cite{IK2} and arguing as in Theorem~\ref{non-Gop}, we see that $\mathbb{G}^{\mathfrak{ns}}_-$ acts by sending $w^{(n);\mathfrak{ns}}_{k}$ to a scalar multiple of $w^{(n);\mathfrak{ns}}_{k-1}$.
Summarizing the above, we obtain the following proposition.
\begin{prop}
\label{prop:nstransitive}
\begin{enumerate}
\item The restriction $\mathbb{G}^{\mathfrak{ns}}_-|_{\mathcal{SW}(m)}$ defines an $U(\mathfrak{ns})$ homomorphism.
\item For $n\in \mathbb{Z}_{\geq 1}$ and $k\geq -n+1$, we have
\begin{align*}
\mathbb{G}^{\mathfrak{ns}}_-w^{(n);\mathfrak{ns}}_{k}\in \mathbb{C}^\times w^{(n);\mathfrak{ns}}_{k-1},\quad\qquad 
\mathbb{G}^{\mathfrak{ns}}_-w^{(n);\mathfrak{ns}}_{-n}=0.
\end{align*}
\end{enumerate}
\end{prop}
In particular, by this proposition, $\mathbb{G}^{\mathfrak{ns}}_-$ is nontrivial.
The following lemma can be proved in the same way as Proposition~\ref{prop:qr-qs}.
\begin{lem}
 \label{lem:sscommu}
 We have
 $
 \lbrack\mathcal{S}^{(\beta_0);[2m]}_{-},\mathcal{S}^{(\beta_0)}_{+}\rbrack=0
 $ on $F^{\mathfrak{ns}}_{1,2m;2k+1}$ $(k\in \mathbb{Z})$.
 \end{lem}
\begin{prop}
The operator $\mathbb{G}^{\mathfrak{ns}}_-$ acts on $\mathcal{SW}(m)$ by a derivation.
\end{prop}
\begin{proof}
Fix $v_1,v_2\in \mathcal{SW}(m)$.
We denote
\begin{align*}
\widehat{\mathcal{S}}_-:=\frac{1}{2m+1}\oint_{z=0} :\mathcal{S}^{(\beta_0)}_-(z)\phi_0(z):\frac{{\rm d}z}{2\pi i}.
\end{align*}
Then $\mathbb{G}^{\mathfrak{ns}}_-=\mathcal{S}^{(\beta_0);[2m]}_-\circ \widehat{\mathcal{S}}_-$.
In the same way as in the proof of Theorem~\ref{G-deri} (see also Remark~\ref{rem:g-hat}), we have
\begin{align}
\label{deri-ns2}
\begin{aligned}
\mathbb{G}^{\mathfrak{ns}}_-Y(v_1,w)v_2
&=(-1)^{|v_1||v_2|}\mathcal{S}^{(\beta_0);[2m]}_-e^{wL_{-1}}Y(v_2,-w)\widehat{\mathcal{S}}_-{v}_1\\
&\qquad +\mathcal{S}^{(\beta_0);[2m]}_-Y({v}_1,w)\widehat{\mathcal{S}}_-{v}_2.
\end{aligned}
\end{align}
Using the commutativity of Lemma~\ref{lem:sscommu} together with Propositions~\ref{com-IK} and~\ref{AM-gen}, we see that $\mathcal{S}^{(\beta_0);[2m]}_-$ is an $\mathcal{SW}(m)$-homomorphism.
Thus, by (\ref{deri-ns2}), we get
\begin{align*}
\mathbb{G}^{\mathfrak{ns}}_-Y(v_1,w)v_2=Y(\mathbb{G}^{\mathfrak{ns}}_-v_1,w)v_2+Y(v_1,w)\mathbb{G}^{\mathfrak{ns}}_-v_2.
\end{align*}
\end{proof}

\begin{prop}
Define $\mathbb{E}^{\mathfrak{ns}},\mathbb{F}^{\mathfrak{ns}},h\in {\rm End}_{\mathfrak{ns}}(\mathcal{SW}(m))$ as follows:
\begin{align*}
\mathbb{E}^{\mathfrak{ns}}:=\mathcal{S}^{(\beta_0)}_+,\qquad
\mathbb{F}^{\mathfrak{ns}}:=\left.\mathbb{G}^{\mathfrak{ns}}_-\right|_{\mathcal{SW}(m)},\qquad
h=-\frac{2a_0}{(2m+1)\beta_-}.
\end{align*}
Then there exists a nonzero constant $c^{\mathfrak{ns}}_{\mathbb{E},\mathbb{F}}$ such that $\mathbb{E}^{\mathfrak{ns}}$, $h$, and $\mathbb{F}^{\mathfrak{ns}}$ satisfy the relation
\begin{align}
\label{eq:efh-ns}
[\mathbb{E}^{\mathfrak{ns}},\mathbb{F}^{\mathfrak{ns}}]=c^{\mathfrak{ns}}_{\mathbb{E},\mathbb{F}}h.
\end{align}
In particular, the Lie algebra $\mathfrak{sl}_2(\mathbb{C})=\mathbb{C}\mathbb{E}^{\mathfrak{ns}}\oplus \mathbb{C}h\oplus \mathbb{C}(c^{\mathfrak{ns}}_{\mathbb{E},\mathbb{F}})^{-1}\mathbb{F}^{\mathfrak{ns}}$ acts on $\mathcal{SW}(m)$ by derivations.
\label{thm:swmain}
\end{prop}
\begin{proof}
As in Theorem~\ref{thm:mainW1p}, by Proposition \ref{sl2action2-s}, it suffices to prove  
\begin{align}
\label{cru-ef-rel-ns}
[\mathbb{E}^{\mathfrak{ns}},\mathbb{F}^{\mathfrak{ns}}]|_{\mathcal{SW}(m)\cap F^{\mathfrak{ns}}_{1,1}}=0
\end{align}
to obtain (\ref{eq:efh-ns}).
Note that $F^{\mathfrak{ns}}_{1,1}\subset \widetilde{F}^{\mathfrak{ns}}_{1,1}$, and that the screening operators $\widetilde{\mathcal{S}}_\pm$ and the compositions $\widetilde{\mathcal{S}}_\pm\circ\widetilde{\mathcal{S}}_\mp$ are well-defined on $\widetilde{F}^{\mathfrak{ns}}_{1,1}$.
From Proposition~\ref{prop:sscommu}, we see that $[\widetilde{\mathcal{S}}_+,\widetilde{\mathcal{S}}_-]=0$ on $\widetilde{F}^{\mathfrak{ns}}_{1,1}$.
Thus, using Propositions~\ref{sl2action2-s} and \ref{prop:nstransitive}, together with (\ref{S-holom0})-(\ref{G-iden-ns0}), the proof is almost identical to that of Theorem~\ref{thm:mainW1p}, and we omit it.
\end{proof}
Note that $\mathcal{SW}(m)$ has the standard $\mathbb{Z}_2$-parity automorphism $\Pi : v \mapsto (-1)^{|v|}v$ $(v\in \mathcal{SW}(m))$, which fixes the even part and multiplies the odd part by $-1$, and that the even strong generators of $\mathcal{SW}(m)$ are $\widehat{W}^{\pm;\mathfrak{ns}}$, $\widehat{W}^{0;\mathfrak{ns}}$, and $T^{(\beta_0;\mathfrak{ns})}$.

\begin{thm}
\label{cor:sw(m)}
The full automorphism group of $\mathcal{SW}(m)$ is $PSL_2(\mathbb{C})\times \mathbb{Z}_2$.
\end{thm}
\begin{proof}
By integrating the $\mathfrak{sl}_2(\mathbb{C})$-action given in Proposition~\ref{thm:swmain}, we see that the full automorphism group ${\rm Aut}(\mathcal{SW}(m))$ contains $PSL_2(\mathbb{C})$ as a subgroup.
Therefore, by Theorem~\ref{AM-ns}, there exists a surjective homomorphism from ${\rm Aut}(\mathcal{SW}(m))$ to $PSL_2(\mathbb{C})$.
Then it follows from Proposition~\ref{AM-gen} that the kernel of the surjective homomorphism coincides with the parity automorphism. Hence, ${\rm Aut}(\mathcal{SW}(m))$ is isomorphic to $PSL_2(\mathbb{C})\times \mathbb{Z}_2$.
\end{proof}

\section*{Concluding remarks and further problems}
In this paper, we constructed the derivations $\mathbb{G}^{[p_\pm]}_{\pm }\in {\rm End}_{\mathfrak{Vir}}(\mathcal{W}_{p_+,p_-})$ based on the $\epsilon$-deformation method in \cite{TW} and showed that the $\mathfrak{sl}_2$-symmetry of $\mathcal{W}_{p_+,p_-}$ established in \cite{AdamovicD/LinX/MilasA:2013,McRaeR/SopinV:2026} can be naturally derived from our derivations. 
In addition, following the method used for $\mathbb{G}^{[p_\pm]}_{\pm }$, we constructed a derivation acting on $\mathcal{SW}(m)$ and
showed that the automorphism group of $\mathcal{SW}(m)$ is $PSL_2(\mathbb{C})\times \mathbb{Z}_2$.
Our method is mainly based on the structure of the Felder-complexes and the analyticity of the Selberg integrals. For this reason, we expect that our approach has a wide range of applicability.
For example, we expect that our construction can also be applied to the vertex operator superalgebras $\mathcal{SW}_{p,q}$ introduced in \cite{AdamovicD/MilasA:2009}, and may in particular be useful for studying the $C_2$-cofiniteness, the classification of simple $\mathcal{SW}_{p,q}$-modules, and the determination of the automorphism group.
As more advanced directions, the ADE classification of representations (cf.~\cite{AdamovicD/LinX/MilasA:2013,AdamovicD/LinX/MilasA:2014,AdamovicD/LinX/MilasA:2015,Lin}) and the theory of Borel actions (cf.~\cite{LS,Sug1,Sug2}) remain to be studied in the general cases of $\mathcal{W}_{p_+,p_-}$, $\mathcal{SW}_{p,q}$. It would be interesting to see whether our method can be applied to these problems.

From the character formulas in \cite{AdamovicD/MilasA:2008-2,AdamovicD/MilasA:2009-2} and the fusion structures (cf.~\cite{Nak2,TWFusion}), it is expected that $\mathcal{W}_{1,2m+1}$ and $\mathcal{SW}(m)$ are related at the level of tensor categories.
It is also known that \cite{McY} the tensor category structure of $\mathcal{W}_{1,p}$ obtained in \cite{TWFusion} can be rederived from the fact that ${\rm Aut}(\mathcal{W}_{1,p})\cong PSL_2(\mathbb{C})$ and the Virasoro tensor category structure at $c=c_{1,p}$ (see also \cite{CMY}). 
Combining our results on the automorphism group of $\mathcal{SW}(m)$ with the tensor-categorical results in \cite{CreutzigT/McRaeR/OroszHunzikerF/YangJ:2026}, we expect to better understand the relationship between $\mathcal{W}_{1,2m+1}$ and $\mathcal{SW}(m)$ and the tensor category structure of the $N=1$ super Virasoro algebra at $c^{\mathfrak{ns}}_{1,2m+1}$.

\section*{{Acknowledgement}}
We would like to thank Akihiro Tsuchiya and Simon Wood for valuable discussions on the ideas behind the construction of Frobenius homomorphisms.
We would also like to thank Hao Li, Robert McRae, and Jinwei Yang for valuable discussions.

\pdfbookmark[1]{References}{ref}

\vspace{10mm}
\ \ \ \ H.~Nakano, \textsc{Osaka City University Advanced Mathematical Institute}\par\nopagebreak
  \textit{E-mail address} : \texttt{hiromutakati@gmail.com}
  
\end{document}